\documentclass{article}

\usepackage{graphicx}
\usepackage{etex}

\usepackage[matrix,arrow,cmtip,rotate,curve,arc]{xy}
\usepackage{verbatim}
\usepackage{mathtools}
\usepackage{soul}
\usepackage{amsthm}

\usepackage{tikz}
\usetikzlibrary{matrix}
\usepackage{tikz-cd}

\usepackage{amsfonts}
\usepackage{youngtab}
\usepackage{amsmath}
\usepackage{amscd}
\usepackage{amssymb}					
\usepackage{tikz}
\usepackage{graphicx}
\usepackage{url}
\usepackage{hyperref}

\DeclareMathOperator*{\argmax}{arg\,max}
\DeclareMathOperator*{\dist}{dist}

\newtheorem{theorem}{Theorem}[section]

\newtheorem{proposition}[theorem]{Proposition}

\newtheorem{definition}[theorem]{Definition}
\newtheorem{lemma}[theorem]{Lemma}

\newcounter{sideremark}

\newcommand\redu{\Rightarrow\!\!\!\!\Rightarrow}
\newcommand\reduP{\stackrel{{\scriptscriptstyle \rm P~\,}}{\redu}}
\newcommand\lredu{\Leftarrow\!\!\!\!\Leftarrow}
\newcommand\lreduP{\stackrel{{\scriptscriptstyle \rm ~~~P\,}}{\lredu}}

\newcommand{\size}{\mathrm{size}}

\newcommand{\thedim}{{n}}

\newcommand{\thedimm}{{k}}

\newcommand{\stdsimp}[1]{\Delta^{#1}}
\renewcommand\:{\colon}

\newcommand{\source}{\mathsf{source} \,}
\newcommand{\target}{\mathsf{target} \,}
\newcommand{\map}{\mathsf{map}}
\newcommand{\sph}{\mathsf{sph}}
\newcommand{\comp}{\mathsf{c}}
\newcommand{\compred}{\mathsf{cr}}

\newcommand{\pt}{*}

\newcommand{\oLoop}{\overline \Omega}
\newcommand{\oI}{\overline I}
\newcommand{\oP}{\overline P}
\newcommand{\id}{\mathop\mathrm{id}\nolimits}
\newcommand{\ef}{\mathrm{ef}}

\newcommand{\Dom}{\mathsf{Dom}}
\newcommand{\oDom}{\overline{\Dom}}
\newcommand{\ooDom}{\overline{\oDom}}
\newcommand{\zero}{\mathbf{0}}

\newcommand{\sSet}{\mathsf{sSet}}

\newcommand{\heading}[1]{\vspace{1ex}\par\noindent{\bf\boldmath #1}}
\newcommand{\Esd}{\mathrm{Esd}}

\def\R{\mathbb{R}}
\def\N{\mathbb{N}}
\def\Z{\mathbb{Z}}

\def\s{\sigma}

\def\defeq{\coloneqq}

\newcommand{\pbsize}{15pt}
\newcommand{\pboffset}{.5}
\newcommand{\xycorner}[3]{\save #2="a";#1;"a"**{}?(\pboffset);"a"**\dir{-};#3;"a"**{}?(\pboffset);"a"**\dir{-}\restore}
\newcommand{\pb}{\xycorner{[]+<\pbsize,0pt>}{[]+<\pbsize,-\pbsize>}{[]+<0pt,-\pbsize>}}

\usetikzlibrary{decorations.pathreplacing}

\usetikzlibrary{calc}

\newcommand\deletion[1]{}

\begin{document}

\title{Computing simplicial representatives of homotopy group elements\thanks{The research 
leading to these results has received funding from Austrian Science Fund (FWF): M 1980.}}


\author{Marek~Filakovsk\'y \and
Peter~Franek \and
Uli~Wagner \and
Stephan~Zhechev
}

\maketitle

\begin{abstract}
A central problem of algebraic topology is to understand the \emph{homotopy groups}
$\pi_d(X)$ of a topological space $X$. For the computational version of the problem, it is
well known that there is no algorithm to decide whether the \emph{fundamental group} 
$\pi_1(X)$ of a given finite simplicial complex $X$ is trivial. On the other hand, there are 
several algorithms that, given a finite simplicial complex $X$ that is \emph{simply connected}
(i.e., with $\pi_1(X)$ trivial), compute the higher homotopy group $\pi_d(X)$ for any 
given $d\geq 2$. 

However, these algorithms come with a caveat: They compute the isomorphism type of
$\pi_d(X)$, $d\geq 2$ as an \emph{abstract} finitely generated abelian group given by 
generators and relations, but they work with very implicit representations of the elements of $\pi_d(X)$.
Converting elements of this abstract group into explicit geometric maps from the $d$-dimensional sphere $S^d$
to $X$ has been one of the main unsolved problems in the emerging field of computational homotopy theory.

Here we present an algorithm that, given a~simply connected space $X$, 
computes $\pi_d(X)$ and represents its elements as simplicial maps from a suitable triangulation 
of the $d$-sphere $S^d$ to $X$. For fixed $d$, the algorithm runs in time exponential in $\size(X)$, 
the number of simplices of $X$. Moreover, we prove that this is optimal: For every fixed $d\geq 2$, 
we construct a family of simply connected spaces $X$ such that for any simplicial map 
representing a generator of $\pi_d(X)$, the size of the triangulation of $S^d$ on which the map is defined,
is exponential in $\size(X)$.
\end{abstract}

\maketitle
\thispagestyle{empty}

\newpage

\section{Introduction}
One of the central concepts in topology are the \emph{homotopy groups} $\pi_d(X)$ of a topological space $X$. 
Similar to the \emph{homology groups} $H_d(X)$, the homotopy groups $\pi_d(X)$ provide a mathematically precise way of 
measuring the ``$d$-dimensional holes'' in $X$, but the latter are significantly more subtle and 
computationally much less tractable than the former. Understanding homotopy groups has been one of the main challenges propelling research in algebraic topology, with only partial results so far despite an enormous effort (see, e.g., \cite{Ravenel,Kochman}); the amazing complexity of the problem is illustrated by the fact that even for the $2$-dimensional sphere $S^2$, the higher homotopy groups 
$\pi_d(S^2)$ are nontrivial for infinitely many $d$ and \emph{known} only for a few dozen values of $d$.

For computational purposes, we consider spaces that have a combinatorial description as \emph{simplicial sets} (or, alternatively,
finite simplicial complexes) and maps between them as \emph{simplicial maps}. 

 A fundamental computational result about homotopy groups is negative: There is no algorithm to decide whether the \emph{fundamental group} $\pi_1(X)$ of a finite simplicial complex $X$ is trivial, i.e., whether every continuous map from the circle $S^1$ to $X$ can be continuously contracted to a point; this holds even if $X$ is restricted to be $2$-dimensional.\footnote{This follows via a standard reduction from a result of Adjan\cite{Adyan:AlgorithmicUnsolvabilityRecognitionPropertiesGroups-1955} and Rabin \cite{Rabin:RecursiveUnsolvabilityGroupTheoreticProblems-1958} on the algorithmic unsolvability of the triviality problem of a group given in terms of generators and relations; we refer to the survey \cite{Soare:ComputabilityDifferentialGeometry-2004} for further background.}

On the other hand, given a~space $X$ that is  \emph{simply connected} (i.e., path connected and with $\pi_1(X)$ trivial) 
there are algorithms that compute the higher homotopy group $\pi_d(X)$, for every given~$d \ge 2$.
The first such algorithm was given by Brown \cite{Brown}, and newer ones have been obtained 
as a part of general computational frameworks in algebraic topology; in particular, an algorithm based on the methods of Sergeraert et al.~\cite{Sergeraert:ComputabilityProblemAlgebraicTopology-1994,RubioSergeraert:ConstructiveAlgebraicTopology-2002}
was described by Real \cite{Real96}. 

More recently, \v{C}adek et al.~\cite{polypost} proved that, for any fixed $d$, the homotopy group $\pi_d(X)$ of a given $1$-connected finite 
simplicial set can be computed in polynomial time. On the negative side, computing $\pi_d(X)$ is \#P-hard if $d$ is part of the input \cite{Anick-homotopyhard,ext-hard} 
(and, moreover, W[1]-hard with respect to the parameter $d$ \cite{Mat-homotopyW1}), even if $X$ is restricted to be $4$-dimensional.
These results form part of a general effort to understand the {computational complexity} of topological questions concerning the classification of maps 
up to {homotopy} \cite{post,surv,ext-hard,VokriFil-homotopic} and related questions, such as the \emph{embeddability problem} for simplicial complexes (a higher-dimensional analogue of graph planarity) \cite{MatousekTancerWagner:HardnessEmbeddings-2011,Matousek:Embeddability-in-the-3-sphere-is-decidable-2014,aslep}.

\subsection{Our Results: Representing Homotopy Classes by Explicit Maps}
By definition, elements of $\pi_d(X)$ are equivalence classes of continuous maps from the $d$-dimensional sphere $S^d$ to $X$, with maps being considered equivalent 
(or lying in the same \emph{homotopy class}) if they are \emph{homotopic}, i.e., if they can be continuously deformed into one another (see Section~\ref{s:prelim} for more details).

The algorithms of \cite{Brown} or \cite{polypost} mentioned above compute $\pi_d(X)$ as an abstract abelian group, 
in terms of generators and relations.\footnote{That is, they compute integers $r,q_1,\ldots,q_k$ such that $\pi_d(X)$ is isomorphic to $\Z^r\oplus \Z_{q_1}\oplus \ldots \oplus \Z_{q_k}$.}
However, they work with very implicit representations of the elements of $\pi_d(X)$.

The main result of this paper is an algorithm that, given an element $\alpha$ of $\pi_d(X)$, computes a suitable triangulation $\Sigma^d$ of the sphere $S^d$ and an explicit simplicial map $\Sigma^d \to X$ representing the given homotopy class $\alpha$.

Apart from the intrinsic importance of homotopy groups, we see this as a first step towards the more general goal of computing explicit maps with specific topological properties; instances of this goal include computing explicit representatives of homotopy classes of maps between more general spaces $X$ and $Y$ 
(a problem raised in \cite{post}) as well as \emph{computing an explicit embedding} of a given simplicial complex into $\R^d$ (as opposed to \emph{deciding embeddability}). Moreover, these questions are also closely related to \emph{quantitative} questions in homotopy theory ~\cite{gromov_quantitative} and in the theory of embeddings \cite{Freedman:Geometric-complexity-of-embeddings-in-Rd-2014}. See Section~\ref{s:related work} for a more detailed discussion of these questions. 

Throughout this paper, we assume that the input $X$ is \emph{simply connected}, i.e., that it is connected and has trivial fundamental group 
$\pi_1(X)$. For the purpose of the exposition, we will assume that $X$ is given as a~$1$-reduced simplicial set, encoded as a~list of
its nondegenerate simplices and boundary operators given via finite tables.
We remark that the class of $1$-reduced simplicial sets contains 
standard models of $1$-connected topological spaces, such as spheres or complex projective spaces.
A~more general version of the theorem that also includes simply connected simplicial complexes is discussed in Section~\ref{s:proof_1}.

{\def\thetheorem{A}
\addtocounter{theorem}{-1} 
\begin{theorem}
\label{t:main} 
There exists an algorithm that, given $d\geq 2$ and a finite $1$-reduced simplicial set $X$, 
computes the generators $g_1,\ldots, g_k$ of $\pi_d(X)$ as simplicial maps $\Sigma_j^d\to X$, 
for suitable triangulations $\Sigma_j^d$ of $S^d$, $j=1,\ldots,k$.

For fixed $d$, the time complexity is exponential in the \emph{size} (number of simplices) of $X$; 
more precisely, it is $O(2^{P(\size(X))})$ where $P=P_d$ is a polynomial depending only on $d$.
\end{theorem}
}
Any element of $\pi_d(X)$ can be expressed as a~sum of~generators, and expressing the sum of two explicit maps 
from spheres into $X$ as another explicit map is a~simple operation. Hence, the algorithm in Theorem~\ref{t:main} 
can convert \emph{any} element of $\pi_d(X)$ into an~explicit simplicial map. 

Theorem~\ref{t:main}  also has the following \emph{quantitative} consequence: Fix some standard triangulation $\Sigma$ of the sphere
$S^d$, e.g., as the boundary of a $d+1$-simplex. By the classical \emph{Simplicial Approximation Theorem}~\cite[2.C]{Hatcher},
for any continuous map $f\colon S^d \to X$, there is a subdivision $\Sigma'$ of $\Sigma$ and a simplicial map $f'\colon \Sigma'\to X$ that
is homotopic to $f$. Theorem~\ref{t:main} implies that if $f$ represents a generator of $\pi_d(X)$, then the size of $\Sigma'$ can be bounded
by an exponential function of the number of simplices of $X$.



Furthermore, we can show that the exponential dependence on the number of simplices in $X$ is inevitable:

{\def\thetheorem{B}
\addtocounter{theorem}{-1} 
\begin{theorem}
\label{t:optimality} 
Let $d\geq 2$ be fixed. Then there is an infinite family of $d$-dimensional $0$-reduced $1$-connected simplicial sets $X$ such that
for any simplicial map $\Sigma\to X$ representing a generator of $\pi_d(X)$, the triangulation $\Sigma$ of  $S^d$ on which $f$ is defined
has size at least $2^{\Omega(\size(X))}$. If $d\geq 3$, we may even assume that $X$ are $1$-reduced.

Consequently, any algorithm for computing simplicial representatives of the generators of $\pi_d(X)$ for $1$-reduced 
simplicial set $X$ has time complexity at least $2^{\Omega(\size(X))}$. 
\end{theorem}
}
In the boundary case of $1$-reduced simplicial sets for $d=2$, we don't know whether the lower complexity bound is sub-exponential or not.
However, we can show that the algorithm from Theorem~\ref{t:main} is optimal in that case as well, see a~discussion 
in Section~\ref{s:proof_2}, page~\pageref{page:optimality}.

In Section~\ref{s:proof_1} and \ref{s:proof_2}, we state and prove generalizations of Theorem~\ref{t:main} and \ref{t:optimality} 
denoted as Theorem~\ref{t:main_a} and~\ref{t:optimality_a}. 
They remove the $1$-reduceness assumption and replace it by a~more flexible certificate of simply connectedness, allowing the input space
$X$ to be a~more flexible simplicial set or simplicial complex.

\subsection{Related and Future Work}
\label{s:related work}
\heading{Computational homotopy theory and applications.} 
This paper falls into the broader area of \emph{computational topology}, which has been a rapidly developing area (see, for instance, the textbooks \cite{EdelsbrunnerHarer:ComputationalTopology-2010,Zomorodian:TopologyComputing-2005,Matveev:AlgorithmicTopology-2007});
more specifically, as mentioned above, this work forms part of a general effort to understand the computational complexity 
of problems in \emph{homotopy theory}, both because of the intrinsic importance of these problems in topology and because of 
applications in other areas, e.g., to algorithmic questions regarding embeddability of simplicial complexes~\cite{MatousekTancerWagner:HardnessEmbeddings-2011,aslep}, to questions in topological combinatorics (see, e.g., 
\cite{MabillardWagner:Elim_II_SoCG-2016}), or to the robust satisfiability of equations \cite{nondec}.

A central theme in topology is to understand the set $[X,Y]$ of all homotopy classes of maps from a space $X$ to a space $Y$.
In many cases of interest, this set carries additional structure, e.g., an abelian group structure, as in the 
case $\pi_d(X)=[S^d,X]$ 
of higher homotopy groups 
that are the focus of the present paper.

Homotopy-theoretic questions have been at the heart of the development of algebraic topology since the 1940's. 
In the 1990s, three independent groups of researchers proposed general frameworks to make various more advanced methods of algebraic topology (such as spectral sequences) \emph{effective} (algorithmic): Sch\"on \cite{Schoen-effectivetop}, Smith
\cite{smith-mstructures}, and Sergeraert, Rubio, Dousson, Romero, and coworkers (e.g., \cite{Sergeraert:ComputabilityProblemAlgebraicTopology-1994,RubioSergeraert:ConstructiveAlgebraicTopology-2002,RomeroRubioSergeraert,SergRub-homtypes}; also see \cite{SergerGenova} for an exposition).
These frameworks yielded general \emph{computability} results for homotopy-theoretic questions (including
new algorithms for the computation of higher homotopy groups \cite{Real96}),
and in the case of Sergeraert et al., also a \emph{practical implementation} in form of the Kenzo software package~\cite{fKenzo}.

Building on the framework of \emph{objects with effective homology} by Sergeraert et al., in recent 
years a variety of new results in computational homotopy theory were obtained \cite{post,pKZ1,polypost,ext-hard,vokrinek:oddspheres,VokriFil-homotopic,aslep,SergRomEffHmtp,Romero2016}, including, in some cases,
the first \emph{polynomial-time algorithms}, by using a refined framework of \emph{objects with polynomial-time homology} 
\cite{pKZ1,polypost} that allows for a computational complexity analysis.
For an introduction to this area from a theoretical computer science perspective and an overview of some of these results, see, e.g.,
\cite{surv} and the references therein. 

\heading{Explicit maps.} As mentioned above, the above algorithms often work with rather \emph{implicit} representations of 
the homotopy classes in $\pi_d(X)$ (or, more generally, in $[X,Y]$) but does not yields explicit maps representing these 
homotopy classes.

For instance, the algorithm in \cite{Real96} computes $\pi_d(X)$ as the \emph{homology group} $H_d(F)$ of an auxiliary space 
$F=F_d(X)$ constructed from $X$ in such a way that $\pi_d(X)$ and $H_d(F)$ are isomorphic as groups.\footnote{Similarly, the 
algorithm in \cite{polypost} constructs an auxiliary chain complex $C$ such that $\pi_d(X)$ is isomorphic to the homology group 
$H_{d+1}(C)$ and computes the latter.}
 
More recently, Romero and Sergeraert~\cite{Romero2016} devised an algorithm that, given a $1$-reduced 
(and hence simply connected) simplicial set $X$ and $d\geq 2$, computes the homotopy group $\pi_d(X)$ as the homotopy group
$\pi_d(K)$ of an auxiliary simplicial set $K$ (a so-called \emph{Kan completion} of $X$) with $\pi_d(X)\cong \pi_d(K)$. Moreover, given 
an element of this group, the algorithm can compute an explicit simplicial map $\Sigma^d \to K$ from a suitable triangulation of $S^d$ to $K$ representing the given homotopy class. In this way, homotopy classes are represented by explicit maps, but as maps to the auxiliary space $K$, 
which is homotopy equivalent to but not homeomorphic to the given space $X$.

By contrast, our general goal is to is represent homotopy classes by maps into the given space; in the present paper, 
we treat, as an important first instance, the case $\pi_d(X)=[S^d,X]$. 

\heading{Open Problems and Future Work.} Our next goal is to extend the results here to the setting of \cite{post},
i.e., to represent, more generally, homotopy classes in $[X,Y]$ by explicit simplicial maps from some suitable 
subdivision $X'$ to $Y$ (under suitable assumptions that allow us to compute $[X,Y]$).\footnote{Similarly as before, the
algorithm in \cite{post} computes $[X,Y]$ as the set $[X,P]$ for some auxiliary space $P$ (a stage of a \emph{Postnikov system} for $Y$)
and represents the elements of $[X,Y]\cong [X,P]$ as maps from $X$ to $P$, but not as maps to $Y$.}

In a subsequent step, we hope to generalize this further to the \emph{equivariant} setting $[X,Y]_G$ of \cite{aslep}, in which a finite group $G$ of 
symmetries acts on the spaces $X,Y$ and all maps and homotopies are required to be \emph{equivariant}, i.e., to preserve the symmetries.

As mentioned above, one motivation is the problem of algorithmically constructing embeddings of simplicial complexes into $\R^d$.
Indeed, in a suitable range of dimensions ($d\geq \frac{3(k+1)}{2}$), 
the existence of an embedding of a finite $k$-dimensional simplicial complex $K$ into $\R^d$ is equivalent to the existence of an 
$\Z_2$-equivariant map from an auxiliary complex $\tilde{K}$ (the deleted product) into the sphere $S^{d-1}$, by a classical theorem 
of Haefliger and Weber \cite{Haefliger:PlongementsDifferentiablesDomaineStable-1962,Weber67}. The proof of the Haefliger-Weber Theorem is, in principle, constructive, but in order to turn this construction into
an algorithm to compute an embedding, one needs an explicit equivariant map into the sphere $S^{d-1}$.

\heading{Quantitative homotopy theory.}
Another motivation for representing homotopy classes by simplicial maps and complexity bounds for such algorithms 
is the connection to \emph{quantitative questions} in homotopy theory~\cite{gromov_quantitative,weinberger_quantitative}
and in the theory of embeddings \cite{Freedman:Geometric-complexity-of-embeddings-in-Rd-2014}. Given a suitable measure 
of \emph{complexity} for the maps in question, typical questions are: What is the relation between the complexity of a given 
null-homotopic map $f: X\to Y$ and the minimum complexity of a nullhomotopy witnessing this? What is the minimum complexity
of an embedding of a simplicial complex $K$ into $\R^d$? In quantitative homotopy theory, complexity is often quantified by
assuming that the spaces are metric spaces and by considering Lipschitz constants (which are  closely related to the sizes of the simplicial representatives of maps and homotopies~\cite{weinberger_quantitative}). For embeddings, the connection is even more direct: a typical measure
is the smallest number of simplices in a subdivision $K'$ or $K$ such that there exists a simplexwise linear-embedding $K' \hookrightarrow \R^d$.

\subsection{Structure of the paper.} 
The remainder of the paper is structured as follows: In Section~\ref{s:outline}, we give a~high-level description of the main ingredients
of the algorithm from Theorem~\ref{t:main}. 
In Section~\ref{s:prelim}, we review a number of necessary technical definitions regarding simplicial sets and the frameworks of effective and polynomial-time homology, in particular Kan's simplicial version of loop spaces and polynomial-time loop contractions for infinite simplicial sets.
In Section~\ref{s:proof_1}, we formally describe the algorithm from Theorem~\ref{t:main} and give a~high level proof based on a~number of lemmas
which are proved in in subsequent chapters. 
Section~\ref{s:proof_2} contains the proof of Theorem~\ref{t:optimality}.
The rest of the paper contains several 
technical parts needed for the proof of Theorem~\ref{t:main}: in Section~\ref{s:berger}, 
we describe Berger's effective Hurewicz inverse and analyze its running time (Theorem~\ref{t:eff_hur}), 
in Section~\ref{a:loop_contraction}, we prove that the stages of the Whitehead tower have polynomial-time contractible loops (Lemma~\ref{l:contr_loop}).
Finally, in Section~\ref{sec:map-into-complex}, 
we show how to reduce the case when the input is a~simplicial complex $X^{sc}$ to the case of an associated 
simplicial set $X$ and convert a~map $\Sigma\to X$ into a~map from a~subdivision $Sd(\Sigma)$ into $X^{sc}$
(Lemma~\ref{l:lift_to_X}).

\section{Outline of the Algorithm}
\label{s:outline}
In this section we present a high-level description of the main steps and ingredients involved in the algorithm from Theorem~\ref{t:main}. 

\heading{The algorithm in a nutshell.}
\begin{enumerate}
%
\item In the simplest case when the space ${X}$ is $(d-1)$-connected (i.e., $\pi_i({X})=0$ for all $i\leq d-1$.),
the classical Hurewicz Theorem ~\cite[Sec. 4.2]{Hatcher} yields an isomorphism $\pi_d(X)\cong H_d(X)$ between the $d$th homotopy group and the $d$th homology group of $X$. Computing generators of the homology group is known to be a computationally easy task 
(it amounts to solving a linear system of equations over the integers). The key is then converting the homology generators into 
the corresponding homotopy generators, i.e., to compute an inverse of the Hurewicz isomorphism. This was described 
in the work of Berger~\cite{Berger_thesis,Berger_paper}. 
We analyze the complexity of Berger's algorithm
in detail and show that it runs in exponential time in the size of $X$ (assuming that the dimension $d$ is fixed).
\item For the general case, we construct an auxiliary simplicial set $F_d$ together with a~simplicial map $\psi_d: F_d\to {X}$ that has the following properties:
\begin{itemize}
\item $F_d$ is a simplicial set that is $d-1$ connected, and
\item $\psi_d\: F_d \to {X}$ induces an isomorphism $\psi_{d*} \: \pi_d(F_d) \to \pi_d({X})$.
\end{itemize}
Our construction of $F_d$ is based on computing stages of the Whitehead tower of ${X}$~\cite[p. 356]{Hatcher};
this is similar to Real's algorithm, which computes $\pi_d(X)$ as $H_d(F_d)$ as an abstract abelian group.

The overall strategy is to use Berger's algorithm on the space $F_d$ and compute generators of $\pi_d(F_d)$ as simplicial maps.
Then we use the simplicial map $\psi_d$ to convert each generator of $\pi_d(F_d) $ into a map $\Sigma^d\to{X}$, and these maps
generate $\pi_d({X})$. The main technical task for this step is to show that Berger's algorithm can be applied to $F_d$. For this, we need to
construct a~polynomial algorithm for explicit contractions of loops in $F_d$ (this space is $1$-connected but not $1$-reduced in
general).

\end{enumerate}

\heading{Our contributions.}
The main ingredients of the algorithm outlined above are the computability of stages of the Whitehead tower~\cite{Real96} as simplicial sets with polynomial-time homology and Berger's algorithmization of the inverse Hurewicz isomorphism~\cite{Berger_thesis,Berger_paper}. 

The idea that these two tools can be combined to compute explicit representatives of $\pi_d(X)$ is rather natural and is also mentioned, for the special case of $1$-reduced simplicial sets, in~\cite[p. 3]{Romero2016}; however, there are a number of technical challenges to overcome in order to carry out this program (as remarked in ~\cite[p. 3]{Romero2016}: ``Clemens Berger's algorithm, quite complex, has never been implemented, severely limiting the current scope of this approach, same comment with respect to the theoretical complexity of such an algorithm.'').
On a technical level, our main contributions are as follows:
\begin{itemize}
\item We give a~complexity analysis of Berger's algorithm to compute the inverse of the Hurewicz isomorphism (Theorem~\ref{t:eff_hur}).
\item We show that the homology generators of the Whitehead stage $F_d$ can be computed in polynomial time (Lemma~\ref{l:Fn}).
\item Berger's algorithm requires an explicit algorithm for loop contraction---a certificate of $1$-connectedness of the space $F_d$. 
While $F_d$ is not $1$-reduced in general, we describe an explicit algorithm for contracting its loop and show that Berger's algorithm
can be applied.
\end{itemize}
We remark that the Whitehead tower stages are simplicial sets with infinitely many simplices, and we need the machinery of objects with polynomial-time homology to carry out the last two steps.

\section{Definitions and Preliminaries}
\label{s:prelim}
In this section, we give the necessary technical definitions that will be used throughout this paper. In the first part, we recall the standard definitions for simplicial sets and the toolbox of effective homology. 

Afterwards, we present Kan's definiton of a loop space and further formalize our definition of (polynomial-time) loop contractions.

\subsection{Simplicial Sets and Polynomial-Time Effective Homology}
\heading{Simplicial sets and their computer representation.}
A simplicial set $X$ is a graded set $X$ indexed by the non-negative integers together with a collection of mappings 
$d_i \colon X_{\thedim} \to X_{\thedim-1}$ and $s_i \colon  X_\thedim \to X_{\thedim +1}, \, 0\leq i\leq \thedim$ called the \emph{face} and \emph{degeneracy} operators. 
They satisfy the following identities:
\[
\begin{array}{ll}
d_i d_j = d_{j-1}d_i  &\quad\text{for  } i<j, \\  
d_i s_i = d_{i +1} s_i = \id & \quad\text{for } 0\leq i< n, \\
d_i s_j = s_j d_{i-1}  &\quad\text{for  } i>j+1, \\
d_i s_j = s_{j-1} d_i  &\quad\text{for  }i<j,\\
s_i s_j = s_{j+1}s_i &\quad\text{for  } i\leq j.\\
\end{array}
\]
More details on simplicial sets and the motivation behind these formulas can be found in~\cite{may,goerssjardine}.

Simplicial maps between simplicial sets are maps of  graded sets which commute with the face and degeneracy operators.
The elements of $X_\thedim$ are called $\thedim$-\emph{simplices}. 
We say that a simplex $x \in X_\thedim$ is \emph{(non-)degenerate} if it can(not) be expressed as $x = s_i y$ for some $y \in X_{\thedim-1}$. 
If a simplicial set $X$ is also a~graded (Abelian) group and face and degeneracy operators are group homomorphisms, we say that $X$ is a simplicial (Abelian) group.

A~simplicial set is called $k$-reduced for $k\geq 0$, if it has a~single $i$-simplex for each $i\leq k$.

For a simplicial set $X$, we define the chain complex  $C_* (X)$ to be a~free Abelian group enerated by the elements of $X_n$ with differential 
$$\partial(c)  = \sum_{i = 0} ^n (-1)^i d_i (c).$$ 

A~simplicial set is \emph{locally effective}, if its simplices have a~specified finite encoding and algorithms are given that compute the face and degeneracy operators.
A~simplicial map $f$ between locally effective simplicial sets $X$ and $Y$ is \emph{locally effective}, if an algorithm is given that for the encoding of any given 
$x\in X$ computes the encoding of $f(x)\in Y$.

We define a~simplicial set to be \emph{finite} if it has finitely many non-degenerate simplices.
Such simplicial set can be algorithmically represented in the following way. The encoding of non-degenerate simplices
can be given via a~finite list and the encoding of a~degenerate simplex $s_{i_k}\ldots s_{i_1} y$ for $i_1<i_2<\ldots <i_k$ and a non-degenerate $y$ 
can be assumed to be a~pair consisting of the sequence $(i_1,\ldots, i_k)$ and the encoding of $y$. 
The face operators are fully described by their action on non-degenerate simplices and can be given via finite tables.
In this way, any simplicial set with finitely many non-degenerate simplices is naturally locally effective. 
Any choice of an~implementation of the encoding and face operators is called a~\emph{representation} of the simplicial set.
The \emph{size of a~representation} is the overall memory space one needs to store the data which represent the simplicial set.

\heading{Geometric realization.} To each simplicial set $X$ we assign a~topological space $|X|$ called its geometric realization. The construction
is similar to that of simplicial complexes. Let $\Delta_j$ be the geometric realization of a~standard $j$-simplex for each $j\geq 0$. 
For each $k$, we define $D_i: \Delta_{k-1}\hookrightarrow \Delta_k$ to be the inclusion of a $(k-1)$-simplex into the $i$'th face of
a~$k$-simplex and $S_i: \Delta_k\to \Delta_{k-1}$ be the geometric realization of a~simplicial map that sends the vertices 
$(0,1,\ldots,k)$ of $\Delta_k$ to the vertices $(0,1,\ldots,i,i,i+1,\ldots,k-1)$. 
The geometric realization $|X|$ is then defined to be a~disjoint union of all simplices $X$ factored by the relation $\sim$
$$
|X|\defeq (\bigsqcup_{n=0}^\infty X_n\times \Delta_n)/\sim
$$
where $\sim$ 
is the equivalence relation generated by the relations $(x,D_i(p))\sim (d_i(x),p)$ for $x \in  X_{n+1}$, 
$p \in \Delta_n$ and the relations $(x, S_i(p))\sim (s_i(x), p)$ for $x \in  X_{n-1}$, $p\in\Delta_n$.

Similarly, a~simplicial map between simplicial complexes naturally induces a~continuous map between their geometric realizations.

\heading{Simplicial complexes and simplicial sets.}
In any simplicial complex $X^{sc}$, we can choose an ordering of vertices and define a~simplicial sets $X^{ss}$ that consists
of all non-decrasing sequences of points in $X^{sc}$: the dimension of $(V_0,\ldots, V_d)$ equals $d$. 
The face operator is $d_i$ omits the $i$'th coordinate and the degeneracy $s_j$ doubles the $j$'th coordinate.
Moreover, choosing a~maximal tree $T$ in the $1$-skeleton of $X$ enables us to construct a~simplicial set 
${X}:=X^{ss}/T$ in which all vertices and edges in the tree, as well as their degeneracies, are considered to
be a~base-point (or its degeneracies). The geometric realizations of $X^{sc}$ and ${X}$ are homotopy equivalent
and ${X}$ is $0$-reduced, i.e. it has one vertex only.

\heading{Homotopy groups.}
Let $(X,x_0)$ be a pointed topological space. 
The $k$-th homotopy group $\pi_k(X,x_0)$ of $(X,x_0)$ is defined as the set of pointed homotopy\footnote{A homotopy 
$F: {S}^k \times I \rightarrow X$ is pointed if $F(*,t) = x_0$ for all $t \in I$.}
classes of pointed continuous maps $({S}^k , *) \rightarrow (X,x_0)$, where $* \in {S}^k$ is a distinguished point. 
In particular, the $0$-th homotopy group has one element for each path connected component of $X$. For $k=1$, $\pi_1(X,x_0)$ 
is the fundamental group of $X$, once we endow it with the group operation that concatenates loops starting and ending in $x_0$.
The group operation on $\pi_k(X,x_0)$ for $k>1$ assigns to $[f],[g]$ the homotopy class of the 
composition $S^k\stackrel{\pi}{\to} S^k\vee S^k \stackrel{f\vee g}{\to} X$ where $\pi$ factors an~equatorial $(k-1)$-sphere containing $x_0$ 
into a point. Homotopy groups $\pi_k$ are commutative for $k>1$.

If the choice of base-points is understood from the context or unimportant, we will use the shorter notation $\pi_k(X)$.
For a~simplicial set $X$, we will use the notation $\pi_k(X)$ for the $k$'th homotopy group of its geometric realization $|X|$.

An important tool for computing homotopy groups is the \emph{Hurewicz theorem}. It says that whenever $X$ is $(d-1)$-connected,
then there is an isomorphism $\pi_d(X)\to H_d(X)$. Moreover, if the element of $\pi_d(X)$ is represented by a~simplicial map
$f: \Sigma^d\to X$ and $\sum_j k_j \sigma_j$ represents a homology generator of $H_d(\Sigma^d)$, 
then the Hurewicz isomorphism maps $[f]$ to the homology class of the formal sum $\sum_{j} k_j f(\sigma_j)$
of $d$-simplices in $X$.

\heading{Effective homology.} 
We call a chain complex $C_*$ \emph{locally effective} if the elements $c\in C_*$ have finite (agreed upon) encoding and there are algorithms computing 
the addition, zero, inverse and differential for the elements of $C_*$. 

A locally effective chain complex $C_*$ is called \emph{effective} if there is an algorithm that for given $n \in \mathbb{N}$ generates a finite basis $c_\alpha \in C_n$ and an algorithm that for every $c\in C_*$ outputs the unique decomposition of $c$ into a linear combination of $c_\alpha$'s. 

Let $C_*$ and $D_*$ be chain complexes. A~\emph{reduction} $C_*\redu D_*$ is a triple $(f,g,h)$ of maps such that 
$f: C_*\to D_*$ and $g: D_*\to C_*$ are chain homomorphisms, $h: C_*\to C_*$ has degree $1$, $fg=\mathrm{id}$ and $fg-\mathrm{id}=h\partial + \partial h$,
and further $hh=hg=fh=0$.

A locally effective chain complex $C_*$ has \emph{effective homology} ($C_*$ is a \emph{chain complex with effective homology}) if there is a~locally effective chain complex
$\tilde{C}_*$, reductions $C_* \lredu \tilde{C}_* \redu C_* ^\ef $ where $C_* ^\ef $ is an effective chain complex, and all the reduction maps are computable.

\heading{Eilenberg-MacLane spaces.} Let $d\geq 1$ and $\pi$ be an Abelian group. An Eilenberg-MacLane space $K(\pi, d)$ is a~topological space with the properties 
$\pi_d(K(\pi,d))\simeq \pi$ and $\pi_j(K(\pi,d))=0$ for $0<j\neq d$. It can be shown that such space $K(\pi,d)$ exists and, under certain natural restrictions,
has a unique homotopy type. If $\pi$ is finitely generated, then $K(\pi,d)$ has a~locally effective simplicial model~\cite{pKZ1}.
\heading{Globally polynomial-time homology and related notions.} In many auxiliary steps of the algorithm, we will construct various spaces and maps.
To analyse the overall time complexity, we need to parametrize all these objects by the very initial input, which is in our case an encoding
of a~finite $1$-reduced simplicial set (or, in Theorem~\ref{t:main_a}, a~more general space endowed with certain explicit 
certificate of $1$-connectedness).

More generally, let $\mathcal{I}$ be a parameter set so that for each
$I\in\mathcal{I}$ an integer $\size(I)$ is defined.
We say that $F$ is a parametrized simplicial set (group, chain group, \ldots), if for each $I\in\mathcal{I}$, a locally effective simplicial set 
(group, chain group, \ldots) $F(I)$ is given. The simplicial set $F$ is \emph{locally polynomial-time}, if there exists a locally effective
model of $F(I)$ such that for each $k\in \N$ and an encoding of a $k$-simplex $x\in F(I)$, 
the encoding of $d_i(x)$ and $s_j(x)$ can be computed in time polynomial in $\size(\text{enc}(x))+\size(I)$. The polynomial,
however, may depend on $k$.
A polynomial-time map between parametrized simplicial sets $F$ and $G$ is an algorithm that for each $k\in\N$, $I\in \mathcal{I}$ and 
an encoding of an $k$-simplex $x$ in $F(I)$ computes the encoding of $f(x)$ in time polynomial in $\size(\text{enc}(x))+\size(I)$: again,
the polynomial may depend on $k$.

Similarly, a~locally polynomial-time (parametrized) chain complex is an assignment of a~computer representation $C_*(I)$ of a~chain complex
with a distinguished basis in each gradation, such that all these basis elements have some agreed-upon encoding.
A~chain $\sum_j k_j \sigma_j$ is assumed to be represented as a list of pairs
$(k_j, \text{enc}(\sigma_j))_j$ and has size $\sum_j (\size(k_j)+\size(\text{enc}(\sigma_j)))$, where we assume that the size of an~integer $k_j$ is 
its bit-size. Further, an algorithm is given that computes the differential of a chain $z\in C_k(I)$ in time polynomial in $\size(z)+\size(I)$, the
polynomial depending on $k$. The notion of a~polynomial-time chain map is straight-forward.

A \emph{globally polynomial-time chain complex} is a locally polynomial-time chain complex $EC$ that in addition has all chain groups $EC(I)_k$
finitely generated and an~additional algorithm is given that for each $k$ computes the encoding of the generators of $EC(I)_k$ in time polynomial in
$\size(I)$. Finally, we define a~\emph{simplicial set with globally polynomial-time homology} to be a~locally polynomial-time parametrized simplicial set $F$
together with reductions $C_*(F)\lredu \tilde{C} \redu EC$ where $\tilde{C},EC$ are locally polynomial-time chain complexes, $EC$ is a~globally polynomial-time chain complex
and the reduction data are all polynomial-time maps, as usual the polynomials depending on the grading $k$.

The name ``polynomial-time homology'' is motivated by the following:
\begin{lemma}
\label{l:hom_gen}
Let $F$ be a~parametrized simplicial set with polynomial-time homology and $k\geq 0$ be fixed. Then all generators of $H_k(F(I))$ can be computed
in time polynomial in $\size(I)$.
\end{lemma}
\begin{proof}
For the globally polynomial-time chain complex $EF$ and each fixed $j$, we can compute the matrix of the differentials $d_j : EF(I)_j \to  EF(I)_{j-1}$ 
with respect to the distinguished bases in time polynomial in $\size(I)$: we just evaluate $d_k$ on each element of the distinguished basis of $EF(I)_k$. 
Then the homology generators of $H_k(EC)$ can be computed using a~Smith normal form algorithm applied to the matrices of $d_k$ and $d_{k+1}$, 
as is explained in standard textbooks (such as~\cite{Munkres}). 
Polynomial-time algorithms for the Smith normal form are nontrivial but known~\cite{KannanBachem}.

Let $x_1,\ldots,x_m$ be the cycles generating $H_k(EF(I))$. We assume that reductions 
$$
C_*(F)\stackrel{(f,g,h)}{\lredu} \tilde{F}\stackrel{(f',g',h')}{\redu} EF
$$ 
are given and all the reduction maps are polynomial. 
Thus we can compute the chains 
$$fg'(x_1),fg'(x_2),\ldots, fg'(x_m)$$ in polynomial time and it is a matter of elementary computation to verify that
they constitute a set of homology generators for $H_k(F(I))$.
\end{proof}

\subsection{Loop Spaces and Polynomial-Time Loop Contraction}

\heading{Principal bundles and loop group complexes.}
In the text we will frequently deal with principal twisted Cartesian products: these are simplicial analogues of principal fiber bundles. 
The definitions in this section come from Kan's article~\cite{Kan:1957}.

We first define the Cartesian product $X \times Y$ of simplicial sets $X,Y$:
The set of $n$-simplices $(X \times Y)_n$ consists of tuples $(x,y)$, where $x \in X_n, x\in Y_n$. 
The face and degeneracy operators on  $X \times Y$ are given by $d_i (x,y) =  (d_i x, d_i y)$, $s_i (x,y) =  (s_i x, s_i y)$.

\begin{definition}[Principal Twisted Cartesian product]\label{d:twistprod}
Let $B$ be a simplicial set with a basepoint $b_0\in B_0$ and $G$ be a~simplicial group. We call a graded map (of degree -1) $\tau\: B_{n+1} \to G_{n}, n \geq 0$ a \emph{twisting operator} if the following conditions are satisfied:
\begin{itemize}
\item $d_{n}\tau(\beta)=\tau(d_{n+1}b)^{-1}\tau(d_n b)$
\item $d_i\tau(\beta)=\tau(d_{i}b )$
for $0\leq i < {n}$
\item $s_i\tau(b)=\tau(s_{i} b)$,  $i < \thedim$, and
\item $\tau(s_\thedim b)=1_\thedim$ for all $b \in B_\thedim$
where $1_\thedim$ is the unit element of $G_\thedim$.
\end{itemize}
Let $B$, $G$, $\tau$ be as above. We will define a~\emph{twisted Cartesian product} $B\times_\tau G$ to be a~simplicial set
$E$ with $E_\thedim=B_\thedim\times G_\thedim$, and the face and degeneracy
operators are also as in the Cartesian product, i.e. $d_i(b,g) = (d_i b, d_i g)$ ,
with the sole exception of $d_\thedim$, which is given by
\[
d_\thedim(b,g)\defeq(d_\thedim b, \tau(b) d_\thedim (g)),
\ \ \ \ (b,g)\in B_\thedim\times G_\thedim.
\]
\end{definition}
It is not trivial to see why this should be the right way of
representing fiber bundles simplicially, but for us, it is only important
that it works, and we will have explicit formulas available for the
twisting operator for all the specific applications.

We remark that in the literature one can find multiple definitions of twisted operator and twisted product~\cite{may,Kan:1957,Berger_thesis}
and that they, in essence differ from each other based on the decision whether the twisting ``compresses'' the first two or the last two face operators. Here, we follow the same notation as in \cite{Berger_thesis}.

\begin{definition} 
\label{d:G-constr}
Let $X$ be a~$0$-reduced simplicial set. 
Then we define $GX$ to be a~(non-commutative) simplicial group such that 
\begin{itemize}
\item 
$GX_n$ has a generator $\overline{\sigma}$ for each $(n+1)$-simplex $\sigma\in X$ and a relation $\overline{s_{n} y}=1$ for each simplex in the image of the last degeneracy $s_{n}$.
\item The face operators are given by $d_i \overline{\sigma}\defeq \overline{d_i \sigma}$ for $i<n$ and $d_n \overline{\sigma}\defeq (\overline{d_{n+1}\sigma})^{-1} \overline{d_n \sigma}$
\item The degeneracy operators are $s_i \overline{\sigma}\defeq \overline{s_i \sigma}$.
\end{itemize}
\end{definition}
We use the multiplicative notation, with $1$ being the neutral element. It is shown in~\cite{Kan:1957} that $GX$ is a discrete simplicial analog of the
loop space of $X$.

For algorithmic puroposes, we assume that an elements $\prod_{j} \overline{\sigma}_j^{k_j}$ of $GX$ is represented as a~list of pairs 
$(\sigma_j, k_j)$ and has size $\sum_j \size(\sigma_j)+\size(k_j)$.

\begin{definition}
\label{d:loop_c}
Let $X$ be a $0$-reduced simplicial set.
We say that a map $c_0: GX_0\to GX_1$ is a~\emph{contraction} of loops in $X$, if $d_0 c_0(x)=x$ and $d_1 c_0(x)=1$ for each $x\in GX_0$.

In case where $X$ has finitely many nondegenerate $1$-simplices, we define the size $\size(c_0)$ to be the sum 
$$
\sum_{\gamma\in X_1} \size(c_0(\gamma)).
$$
\end{definition}

\heading{Loop contraction for simplicial complexes.}
Let $X^{sc}$ be a~simplicial complex. 
Let $T$ be a~spanning tree in the $1$-skeleton of $X^{sc}$ and $R$ a~chosen vertex. 
For each oriented edge $e=(v_1 v_2)$ we define a~formal inverse to be $e^{-1}:=(v_2 v_1)$ and we also consider degenerate edges $(v,v)$.
A~\emph{loop} is defined as a~sequence $e_1,\ldots, e_k$ of oriented edges in $X^{sc}$ such that
\begin{itemize}
\item The end vertex of $e_i$ equals the initial vertex of $e_{i+1}$, and
\item The initial vertex of $e_1$ and the end vertex of $e_k$ equal $R$.
\end{itemize}
Every edge $e$ that is not contained in $T$ gives rise to a~unique loop $l_e$.
Further, every loop in $X^{sc}$ is either a~concatenation of such $l_e$'s, or can be derived from such concatenation
by inserting and deleting consecutive pairs $(e,e^{-1})$ and degenerate edges.  
Before we formally define our combinatorial version of loop contraction, we need the following
definition.

\begin{definition}
Let $S$ be a~set,  $U\subseteq S$, $F(S)$ and $F(U)$ be free groups generated by $S$, $U$, respectively.\footnote{Formally, elements
of $F(S)$ are sequences of symbols $s^{\epsilon}$ for $\epsilon\in \{1,-1\}$ and $s\in S$ with the relation $s^1 s^{-1}=1$, where $1$ represents the
empty sequence. The group operation is concatenation.}
Let $h_U: F(S)\to F(S)$ be a~homomorphism that sends each $u\in U$ to $1$ and each $s\in S\setminus U$ to itself.
We say that an element $x$ of $F(S)$ equals $y$ modulo $U$, if $h_U(x)=y$.
\end{definition}
An example of an element that is trivial modulo $U$ is the word $s \, u \, s^{-1}$, where $s\in S$ and $u\in U$.

\begin{definition}
\label{d:contraction_sc}
Let $S$ be the set of all oriented edges and oriented degenerate edges in $X^{sc}$ and assume that a~spanning tree $T$ is chosen. 
Let $U$ be the set of all oriented edges in $T$, including all degenerate edges.
A~\emph{contraction of an edge~$\alpha$} is a~sequence of vertices $A_0,A_1,\ldots,A_s$ and $B_1,\ldots, B_{s}$ such that
\begin{itemize}
\item for each $i$, $\{A_i,A_{i+1}, B_{i+1}\}$ is a~simplex of $X^{sc}$, and
\item the element of $F(S)$
\begin{equation}
\label{e:long_loop}
(A_0 B_1) (B_1 A_1) (A_1 B_2) (B_2 A_2) \ldots (B_s A_s) (A_s A_{s-1}) (A_{s-1} A_{s-2}) \ldots (A_1 A_0)
\end{equation}
equals $\alpha$ modulo $U$.

A~\emph{loop contraction in a~simplicial complex} is the choice of a~contraction of $\alpha$ for each edge $\alpha\in X^{sc}\setminus T$.
\end{itemize}
The size of the contraction of $\alpha$ is defined to be the number of vertices in (\ref{e:long_loop}) and the 
size $\size(c)$ of the loop contraction on $X^{sc}$ is the sum of the sizes over all $\alpha\in X^{sc}\setminus T$.
\end{definition}
\begin{figure}[!htb]
\begin{center}
\includegraphics{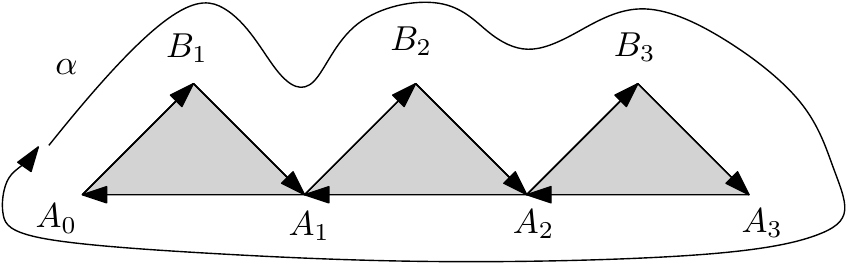}
\end{center}
\caption{The loop ranging over the boundary of this geometric shape equals $\alpha$, after ignoring edges in the maximal tree
and canceling pairs $(e,e^{-1})$. The interior of the triangles gives rise to a~contraction.} 
\label{fig:contraction_sc}
\end{figure}
The geometry behind this definition is displayed in Figure~\ref{fig:contraction_sc}. The sequence of $A_i$'s and $B_j$'s
gives rise to a~map from the sequence of (full) triangles into $X^{sc}$. The big loop around the boundary is combinatorially
described by (\ref{e:long_loop}). We can continuously contract all of its parts that are in the tree $T$ to a~chosen 
basepoint, as the tree is contractible. Further, we can continuously contract all pairs of edges $(e,e^{-1})$ and 
what remains is the original edge $\alpha$: with all the tree contracted to a~point, it will be transformed into a loop
that geometrically corresponds to $l_\alpha$. 
The interior of the full triangles then constitutes its ``filler'', hence a~certificate of the contractibility of $l_\alpha$.

A~loop contraction in the sense of Definition~\ref{fig:contraction_sc} exists iff the space $X^{sc}$ is simply connected.
One could choose different notions of loop contraction. For instance, we could provide, for each $\alpha$,
a~simplicial map from a~triangulated $2$-disc into $X^{sc}$ such that the oriented boundary of the disc would be mapped exactly to $l_\alpha$. 
The description from Definition~\ref{d:contraction_sc} could easily be converted into such map. 
We chose the current definition because of its canonical and algebraic nature.
The connection between Definitions~\ref{d:loop_c} and~\ref{d:contraction_sc} is the content of the following lemma.
\begin{lemma}
\label{l:contr2contr}
Let $X^{sc}$ be a~$1$-connected simplicial complex with a~chosen orientation of all simplices, 
$X^{ss}$ the induced simplicial set, $T$ a~maximal tree in $X^{sc}$, and ${X}:=X^{ss}/T$ 
the corresponding $0$-reduced simplicial set.
Assume that a~loop contraction in the simplicial complex $X^{sc}$ is given, such as described
in Definition~\ref{d:contraction_sc}. Then we can algorithmically compute $c_0(\alpha)\in G{X}_1$ such that $d_0 c_0(\alpha)=\alpha$ 
and $d_1 c_0(\alpha)=1$, for every generator $\alpha$ of $G{X}_0$. 
Moreover, the computation of $c_0(\alpha)$
is linear in the size of $X^{sc}$ and the size of the simplicial complex contraction data.
\end{lemma}
\begin{proof}
For each $i$, the triangle $\{A_i,A_{i+1},B_{i+1}\}$ from Def.~\ref{d:contraction_sc} is in the simplicial complex $X^{sc}$.
There is a~unique oriented $2$-simplex in $X^{ss}$ of the form $(V_0,V_1,V_2)$ (possibly degenerate) such that 
$\{V_0,V_1,V_2\}=\{A_i,A_{i+1},B_{i+1}\}$. Let as denote such oriented simplex by $\sigma_i$, and its image in $G{X}_1$ 
by $\overline{\sigma}_i$.
We will define an element $g_i\in G{X}_1$ such that it satisfies 
\begin{equation}
\label{e:properties_of_gi}
d_0 g_i\simeq \overline{(A_i, A_{i+1})} \quad\text{ and }\quad d_1 g_i\simeq \overline{(A_i, B_{i+1})}\,\, \overline{(B_{i+1},A_{i+1})}
\end{equation}
where $\simeq$ is an equivalence relation that identifies any element $\overline{(U,V)}\in G{X}_1$ with $\overline{(V,U)}^{-1}$ 
(note that only one of the symbols $(U,V)$ and $(V,U)$ is well defined in $X^{ss}$, resp. ${X}$.)
Explicitly, we can define $g_i$ with these properties as follows:
\begin{itemize}
\item If $\sigma=(B_{i+1}, A_i, A_{i+1})$, then $g_i:=\overline{\sigma}_i$,
\item If $\sigma=(A_i, A_{i+1}, B_{i+1})$, then $g_i:=s_0\overline{(d_2 {\sigma})} \, \overline{\sigma}_i \,s_0 d_0(\overline{\sigma}_i)^{-1}$
\item If $\sigma=(A_{i+1}, B_{i+1}, A_{i})$, then $g_i= s_0 d_0\overline{\sigma_i}^{-1} \,\overline{\sigma_i} \, s_0 (\overline{d_1 {\sigma}_i})^{-1}$ 
\item If $\sigma=(B_{i+1}, A_{i+1}, A_i)$, then $g_i:=\overline{\sigma_i}^{-1}$
\item If $\sigma=(A_{i+1}, A_i, B_{i+1})$, then $g_i:= s_0 d_0\overline{\sigma_i} \,\, \overline{\sigma_i}^{-1}\, s_0 (\overline{d_2 \sigma_i})^{-1}$
\item If $\sigma=(A_{i}, B_{i+1}, A_{i+1})$, then $g_i:=s_0(\overline{d_1 \sigma_i})\,\overline{\sigma_i}^{-1} s_0 d_0\overline{\sigma_i}$.
\end{itemize}
Let $g:=g_0\ldots, g_s$. The assumption (\ref{e:long_loop}) together with equation (\ref{e:properties_of_gi}) immediately implies that 
$d_1 g (d_0 g)^{-1}=\overline{\alpha}$. Thus we define $c_0(\overline{\alpha}):=s_0 d_1(g)\,g^{-1}$.
Algorithmically, to construct $g$ amounts to going over all the triples $(A_{i}, A_{i+1},B_{i+1})$ from a~given sequence of $A_i's$ and $B_j$'s,
checking the orientation and computing $g_i$ for every $i$.
\end{proof}

\heading{Polynomial-time loop contraction.} Let $F$ be a~parametrized simplicial set such that each $F(I)$ is $0$-reduced.
Using constructions analogous to those defined above, $GF$ is a~parametrized locally-polynomial simplicial group whereas
we assume a~simple encoding of elements of $GF_i$ as follows. If $x=\prod_j \overline{\sigma_j}^{k_j}\in GF(I)_k$ where
$\sigma_j$ are $(k+1)$-simplices in $F(I)$, not in the image of $s_k$, then we assume that $x$ is stored in the memory as a list of pairs
$(k_j,\text{enc}(\sigma_j))$ and has size $\sum_j (\size(k_j)+\size(\sigma_j))$ where some $\sigma_i$ may be equal to $\sigma_j$ for $i\neq j$.
Face and degeneracy operators are defined in Definition (\ref{d:G-constr}) and it is easy to see that for any locally polynomial-time simplicial set $F$,
$GF$ is a locally polynomial-time simplicial group.
\begin{definition}
\label{d:poly-loop}
Let $F$ be a~locally polynomial simplicial set. We say that $F$ has \emph{polynomially contractible loops}, if 
there exists an algorithm that for a $0$-simplex $x\in GF(I)$ computes a $1$-simplex $c_0(x)\in GF(I)$ such that $d_0 x=x$, $d_1 x=1 \in GF(I)_0$, 
and the running-time is polynomial in $\size(x)+\size(I)$.
\end{definition}

\section{Proof of Theorem 1}
\label{s:proof_1}
We will prove a~stronger statement of Theorem~\ref{t:main} formulated as follows.
{\def\thetheorem{A.1}
\addtocounter{theorem}{-1} 
\begin{theorem}
\label{t:main_a} 
There exists an algorithm that, given $d\geq 2$ and a finite $0$-reduced simplicial set $X$ 
(alternatively, a~finite simplicial complex) with an explicit loop contraction
$c_0$ (such as in Definition~\ref{d:loop_c} or~\ref{d:contraction_sc}) 
computes the generators $g_1,\ldots, g_k$ of $\pi_d(X)$ as simplicial maps $\Sigma_j^d\to X$, 
for suitable triangulations $\Sigma_j^d$ of $S^d$, $j=1,\ldots,k$.

For fixed $d$, the time complexity is exponential in the \emph{size} of $X$ and the size of the loop contraction $c_0$; 
more precisely, it is $O(2^{P(\size(X)+\size(c_0))})$ where $P=P_d$ is a polynomial depending only on $d$.
\end{theorem}
}
This immediately implies Theorem~\ref{t:main}, as for a~$1$-reduced simplicial set, the contraction $c_0$ is 
trivial, given by $c_0(1)=1$.

The proof of Theorem~\ref{t:main_a} is based on a~combination of four statements presented here as Lemma~\ref{l:Fn}, Theorem~\ref{t:eff_hur}, 
Lemma~\ref{l:contr_loop} and Lemma~\ref{l:lift_to_X}. Each of them is relatively independent and their proofs are 
delegated to further sections.

First we present an algorithm that, given a $1$-connected finite simplicial set $X$ and a positive integer $d$, outputs a simplicial set $F_d$ and a simplicial map $\psi_d$ such that
\begin{itemize}
\item the simplicial set $F_d$ is $d-1$ connected, it has  polynomial-time effective homology and polynomially contractible loops.
\item the simplicial map $\psi_d\: F_d \to X$ is polynomial-time and induces an~isomorphism \\$\psi_{d*} \: \pi_d(F_d) \to \pi_d(X)$.
\end{itemize}

\heading{Whitehead tower.}
We construct simplicial sets $F_d$ as stages of a so-called \emph{Whitehead tower} for the simplicial set $X$. 
It is a~sequence of simplicial sets and maps
\[
\xymatrix{
\cdots \ar[r] & F_d  \ar[r]^-{f_d} & F_{d-1}  \ar[r]^{f_{d-1}}&\cdots  \ar[r]^{f_4} & F_3  \ar@{>>}[r]^-{f_3} & F_2 = X.
}
\]
where $f_i$ induces an isomorphism $\pi_j(F_{i+1}) \to \pi_j(F_{i})$ for $j>i$
and $\pi_j(F_i) = 0$ for $j< i$. 
We define $\psi_d = f_d f_{d-1} \ldots f_3$. One can see that $F_d, \psi_d$ satisfy the desired properties.

\begin{lemma}~\label{l:Fn}
Let $d\geq 2$ be a~fixed integer.
Then there exists a polynomial-time algorithm that, for a~given 
$1$-connected finite simplicial set $X$, constructs the stages $F_2, \ldots, F_d$ of the Whitehead tower of $X$. 

The simplicial sets $F_k(X)$, parametrized by $1$-connected finite simplicial sets $X$, have polynomial-time homology and 
the~maps $f_k$ are polynomial-time simplicial maps.
\end{lemma}
\begin{proof}
The proof is by induction. The basic step is trivial as $F_2 = X$. We describe how to obtain $F_{k+1}, f_{k+1}$ assuming that we have computed $F_{k}$, $2\leq k <d$.

\begin{enumerate}
\item We compute simplicial map $\varphi_{k}\: F_{k} \to K(\pi_{k}(X), k) = K(\pi_k(F_k), k)$ 
that induces an isomorphism $\varphi_{k*}\: \pi_{k}(F_{k}) \to \pi_{k}(K(\pi_{k}(X), k))\cong \pi_{k}(X)$. This is done using the algorithm in~\cite{polypost}, as $K(\pi_{k}(X), k)$ is the first nontrivial stage of the Postnikov tower for the simplicial set $F_{k}$.

For the simplicial set $K(\pi_{k}(X), k)$ and for such simplicial sets there is a classical principal bundle (twisted Cartesian product) (see \cite{may}):
\[
\xymatrix{
K(\pi_{k}(X), k-1) \ar[d] \\
E(\pi_{k}(X), k-1) = K(\pi_{k}(X), k) \times_{\tau}  K(\pi_{k}(X), k-1) \ar@{>>}[d]^\delta \\
K(\pi_{k}(X), k)
}
\]
\item
We construct $F_{k+1}$ and $f_{k+1}$ as a pullback of the twisted Cartesian product:
\[
\xymatrix{
K(\pi_{k}(X), k-1)  \ar[rr]^{\cong} \ar[d]&&K(\pi_{k}(X), k-1) \ar[d] \\
F_{k+1} \defeq  F_{k} \times_{\tau'} K(\pi_{k}(X), k-1) \ar@{.>}[rr]\ar@{.>>}[d]_{f_{k+1}} \pb &&K(\pi_{k}(X), k) \times_{\tau}  K(\pi_{k}(X), k-1) \ar@{>>}[d]^\delta \\
F_{k}\ar[rr]^{\varphi_{k}} && K(\pi_{k}(X), k).
}
\]
\end{enumerate}
It can be shown that the pullback, i.e. simplicial subset of pairs $(x,y) \in F_k  \times E(\pi_{k}(X),k-1)$ such that $\delta (y) = \varphi_k(x)$, can be identified with the twisted product as above~\cite{may}, where the twisting operator $\tau'$ is defined as $\tau \varphi_k$.

To show correctness of the algorithm, we assume inductively, that $F_{k}$ has polynomial-time effective homology. According to  \cite[Section 3.8]{polypost}, the simplicial sets $K(\pi_{k}(X), {k-1})$, $E(\pi_{k}(X), k-1)$, $K(\pi_{k}(X), k)$ have polynomial-time effective homology and maps $\varphi_k, \delta$ are polynomial-time. Further, they are all obtained by an algorithm that runs in polynomial time.

As $F_{k+1}$ is constructed as a twisted product of $F_k$ with $K(\pi_{k}(X), k)$, Corollary~3.18 of~\cite{polypost} 
implies that $F_{k+1}$ has polynomial-time effective homology and $f_{k+1}$ is a~polynomial-time map.\footnote{We remark that the paper~\cite{polypost} uses a different formalization of twised cartesian product than the one employed by us. However, the paper \cite{Filakovsky-tensor}, on which the~Corollary~3.18 of~\cite{polypost} is based, can be reformulated in context of the definition used here. We do not provide full details, only remark that one has to make a choice of \emph{Eilenberg-Zilber reduction data} that corresponds to the definition of twisted cartesian product.}

The sequence of simplicial sets $\xymatrix{ F_{k+1} \ar[r]^{f_{k+1}}&F_{k} \ar[r]^-{\varphi_{k}} &K (\pi_{k}(X), k)}$ 
induces the long exact sequence of homotopy groups
\[
\xymatrix{
\cdots \ar[r] &\pi_i( F_{k+1})  \ar[r]^{f_{k +1*}}& \pi_i(F_{k} )\ar[r]^-{\varphi_{k*}} & \pi_i(K (\pi_{k}(X), k)) \ar[r] &\pi_{i-1}(F_{k+1})\ar[r] & \cdots
}
\]
The reason why this is the case follows from a rather technical argument that identifies the~simplicial set $F_{k +1}$ with a so called \emph{homotopy fiber} of the map $\varphi_{k}\: F_{k} \to K (\pi_{k}(X), k)$. 
In more detail, the category of simplicial sets is right proper \cite[II.8.6–7]{goerssjardine} and map $\delta$ is a so-called Kan fibration \cite[§~23]{may}. This makes the pullback $F_{k +1 }$ coincide with so-called homotopy pullback. Further, the~simplicial set $E(\pi_{k}(X),k-1)$ is contractible, hence the homotopy pullback is a homotopy fiber. The induced exact sequence is due to~\cite[chapter I.3]{quillen1967homotopical}.

The inductive assumption, together with the fact that $\varphi_k$ induces an isomorphism $\varphi_{k*}\: \pi_{k}(F_{k}) \to \pi_{k}(K(\pi_{k}(X),k))$ imply that $f_k$ induces an isomorphism $\pi_j(F_{k +1}) \to \pi_j(F_{k})$ for $j>k$ and $\pi_j(F_{k +1}) = 0$ for $j\leq k$. 
\end{proof}
The lemma implies that the~simplicial sets $F_k$ have polynomial-time effective homology and maps $\psi_k = f_k f_{k-1} \ldots f_3$ are polynomial-time as they are defined as a composition of polynomial-time maps $f_i$.

The following theorem is a~key ingredient of our algorithm. 

\begin{theorem}[Effective Hurewicz Inverse]
\label{t:eff_hur}
Let $d>1$ be fixed and $F$ be an $(d-1)$-connected $0$-reduced simplicial set parametrized by a~set $\mathcal{I}$,
with polynomial-time homology and polynomially contractible loops.

Then there exists an algorithm that, for a~given $d$-cycle $z\in Z_{d}(F(I))$, 
outputs a~simplicial model $\Sigma^{d}$ of the $d$-sphere and a simplicial map $\Sigma^{d}\to F(I)$ 
whose homotopy class is the Hurewicz inverse of $[z]\in H_{d}(F(I))$. 

Moreover, the time complexity is bounded by an exponential of a polynomial function in $\size(I)+\size(z)$.
\end{theorem}

The construction of an effective Hurewicz inverse is the main result of~\cite{Berger_thesis} 
and further details are provided in Section~\ref{s:berger}. 
It exploits a~combinatorial version of Hurewicz theorem given by Kan in~\cite{Kan:Hurewicz} where $\pi_d(F)$ 
is described in terms of $\pi_{d-1}(\widetilde{GF})$ where $\widetilde{GF}$ is a non-commutative simplicial group that models the loop space of $F$.
Kan showed that the Hurewicz isomorphism can be identified with a map $H_{d-1}(\widetilde{GF})\to H_{d-1}(\widetilde{AF})$ induced by Abelianization.
Berger then describes the inverse of the Hurewicz homomorphism as a composition of the maps $1,2,3$ in the diagram
\[
\xymatrix{
\pi_{d}(F) && H_{d}(F)\ar@{.>}[ll]_{h^{-1}} \ar[d]^1\\
H_{d-1}(\widetilde{GF}) \ar[u]^{3} &&H_{d-1}(\widetilde{AF}). \ar[ll]^{2}
}
\]
Arrow $1$ is induced by a chain homotopy equivalence and arrow $3$ by Berger's explicit geometric model of the loop space. 
To algorithmize arrow 2, we need an algebraic machinery that includes an explicit contraction of $k$-loops in $\widetilde{GF}$ for all $k<d-1$.
Those are based partially on linear computations in the Abelian group $\widetilde{AF}$ and partially on explicit inductive formulas dealing with commutators.
The lowest-dimensional contraction operation, however, cannot be algorithmized, without some external input.
The possibility of providing it 
is is the content of the following claim:
\begin{lemma}
\label{l:contr_loop}
Let $d\geq 2$ be a fixed integer and $\mathcal{I}$ be the set of all $1$-connected $0$-reduced finite simplicial sets with 
an~explicit loop contraction $c_0$.
Then the simplicial set $F_d$ from Lemma~\ref{l:Fn}, parametrized by $\mathcal{I}$, has polynomial-time contractible loops.
\end{lemma}
The proof is constructive, based on explicit formulas in our model of $F_d$: the details are in Section~\ref{a:loop_contraction}.

The core of the algorithm we will describe works with simplicial sets and simplicial maps between them. If our input is a~simplicial complex,
we need tools to convert them into maps between simplicial complexes. The next two lemmas address this. 
\begin{lemma}
\label{l:complexifying}
Let $Y$ be a~finite simplicial set. Then there exists a polynomial-time algorithm that computes a~simplicial complex $Y^{sc}$ with a~given orientation
of each simplex, and a~map $\gamma: Y^{sc}\to Y$ (still understood to be a~map between simplicial sets) such that the geometric realization of $\gamma$
is homotopic to a~homeomorphism.
\end{lemma}
This construction is given in~\cite[Appendix B]{ext-hard}.\footnote{A version of this lemma is given as~\cite[Proposition 3.5]{ext-hard}.
However, we also need the fact that $|Y^{sc}|$ is homeomorphic to $|Y|$, which is not explicitly mentioned in this reference,
but follows easily from the construction.}
Explicitly, the simplicial complex $Y^{sc}$ is defined to be $Y^{sc}:=B_*(Sd(Y))$, where $Sd$ is the barycentric subdivision
functor and $B_*$ a functor introduced in~\cite{Jardine:SimplicialApproximation-2004}: $Y^{sc}$ can be constructed recursively by
adding a vertex $v_\sigma$ for each nondegenerate simplex $\sigma\in Sd(Y)$ and replacing $\sigma$ by the cone with apex $v_\sigma$ 
over $B_*(\partial \sigma)$.
The subdivision $Sd(Y)$ is a~regular simplicial set and
$B_*(Sd(Y))$ coincides with the flag simplicial complex of the poset of nondegenerate simplices of $Sd(Y)$.
It follows that the geometric realisations $|Y^{sc}|$ is homeomorphic\footnote{The subdivision $Sd(Y)$ has geometric
realisation homeomorphic to $|Y|$ by~\cite[Thm 4.6.4]{FritschPiccinini:CellularStructures-1990}. The realisation
of $Sd(X)$ is a~regular CW complex and $B_*(Sd(Y))$ coincides with the first derived subdivision of this
regular CW complex, as defined in~\cite[p. 137]{geoghegan2007}.  
The geometric realisation of the resulting simplicial complex is still homeomorphic to $|Y|$ and $|Sd(Y)|$ by~\cite[Prop. 5.3.8]{geoghegan2007}.}
to $|Y|$.
Simplices of $Y^{sc}$ are naturally oriented and the explicit description of $\gamma$ is given 
in~\cite[p. 61]{ext-hard} and the references therein.

In our main algoritm, $Y=\Sigma^d$ will be a~triangulation of the $d$-sphere and $X$ a~simplicial set derived from a~simplicial complex
$X^{sc}$ by contracting its spanning tree into a~point. The following lemma shows that we can convert a~map $\Sigma^{sc}\to X$ into
a~map $(\Sigma^{sc})'\to X^{sc}$ between simplicial complexes. 
\begin{lemma}
\label{l:lift_to_X}
Let $d>0$ be fixed. Assume that $X^{sc}$ is a~given simplicial complex with a chosen ordering of vertices and a~maximal spanning tree $T$; we denote the
underlying simplicial set by $X^{ss}$. Let $p: X^{ss} \to X:=X^{ss}/T$ be the projection to the 
associated $0$-reduced simplicial set. Let $\Sigma$ be a~given $d$-dimensional simplicial complex with a~chosen orientation of each simplex,
$\Sigma^{ss}$ the induced simplicial set, and $f:\Sigma^{ss}\to {X}$ a~simplicial map.

Then there exists a~subdivision $\mathrm{Sd}(\Sigma)$ and a~simplicial map $f':\mathrm{Sd}(\Sigma)\to X^{sc}$ between 
\emph{simplicial complexes}\footnote{The constructed map $f$ does not necessarily preserves orientations: 
it only maps simplices to simplices.} such that
$$
|\Sigma|=|\mathrm{Sd}(\Sigma)| \stackrel{|f'|}{\to} |X^{sc}| \stackrel{|p|}{\to} |X|
$$
is homotopic to $|\Sigma^{ss}|\stackrel{|f|}{\to} |{X}|$. 
Moreover, $f'$ can be computed in polynomial time, assuming an~encoding of the input $f,\Sigma,X^{sc}$, ${X}$ and $T$.
\end{lemma}

Thus if $\Sigma$ is a~sphere and $f$ corresponds to a~homotopy generator,
$f'$ is the corresponding homotopy generator represented as a~simplicial map between simplicial complexes.
We remark that the algorithm we describe works even if $d$ is a~part of the input, but the time complexity would be exponential in general,
as the number of vertices in our subdivision $\mathrm{Sd}(\Sigma)$ would grow exponentially with $d$.

The proof of~Lemma~\ref{l:lift_to_X} is given in~Section~\ref{sec:map-into-complex}.
\begin{proof}[Proof of Theorem~\ref{t:main_a}]
First assume that a~finite simplicial complex $X^{sc}$ is given together with a~loop contraction. 
Then the algorithm goes as follows.
\begin{enumerate}
\item We choose an ordering of vertices and convert $X^{sc}$ into a~simplicial set. 
Choosing a~spanning tree and contracting it to a~point creates a~$0$-reduced 
simplicial set $X$ homotopy equivalent to $X^{sc}$. 
By~Lemma~\ref{l:contr2contr}, we can convert the input data into
a~list $c_0(\alpha)$ for all generators $\alpha$ of $GX_0$ in polynomial time.
\item
We construct the simplicial set $F_d$ from Lemma~\ref{l:Fn} as simplicial set with polynomial-time effective homology. 
Hence by Lemma~\ref{l:hom_gen} we can compute the generators of $H_d(F_d)$ in time polynomial in $\size(X)$. 
Due to Lemma~\ref{l:contr_loop} and Theorem~\ref{t:eff_hur}, we can convert these homology generators to homotopy generators
$\Sigma_j^d \to F_d$ in time exponential in $P(\size(X)+\size(c_0))$ where $P$ is a~polynomial.
\item We compose the representatives of $\pi_d(F_d)$ with $\psi_d$ to obtain representatives $\Sigma_j^d\to X$ of the generators of $\pi_d(X)$,
another polynomial-time operation. This way, we compute explicit homotopy generators as maps into the simplicial set $X$.
\item We use Lemma~\ref{l:complexifying} to compute simplicial complexes $\Sigma_j^{sc}$ and maps $\Sigma_j^{sc}\to \Sigma^d$
homotopic to homeomorphisms. The compositions $\Sigma_j^{sc}\to\Sigma_j^d\to X$ still represent a~set of homotopy generators.
Finally, by Lemma~\ref{l:lift_to_X}, we can compute, for each $j$, a~subdivision
of the sphere $\Sigma_j^{sc}$ and a~simplicial map from this subdivision 
into the simplicial \emph{complex} $X^{sc}$, in time polynomial in the size of the representatives
$\Sigma_j^{sc}\to X$.
\end{enumerate}
In case when the input is a~$0$-reduced simplicial set $X$ with a~loop contraction $c_0$, only steps $2$ and $3$ are performed.
In either case, the overall exponential complexity bound comes from Berger's Effective Hurewicz inverse theorem.
\end{proof}
%
\section{Proof of Theorem \ref{t:optimality}}
\label{s:proof_2}
Similarly as in the proof of Theorem~\ref{t:main}, we prove a~slightly stronger version of Theorem~\ref{t:optimality} that
also includes finite simplicial complexes.
{\def\thetheorem{B.1}
\addtocounter{theorem}{-1} 
\begin{theorem}
\label{t:optimality_a} 
Let $d\geq 2$ be fixed. Then
\begin{enumerate}
\item 
there is an infinite family of $d$-dimensional $1$-connected finite simplicial complexes $X$ such that
for any simplicial map $\Sigma\to X$ representing a generator of $\pi_d(X)$, the triangulation $\Sigma$ 
of  $S^d$ on which $f$ is defined
has size at least $2^{\Omega(\size(X))}$.
\item 
there is an infinite family of $d$-dimensional $(d-1)$-connected and $(d-2)$-reduced simplicial sets $X$ such that
for any simplicial map $\Sigma\to X$ representing a generator of $\pi_d(X)$, the triangulation $\Sigma$ of  $S^d$ on which $f$ is defined
has size at least $2^{\Omega(\size(X))}$.
\end{enumerate}
Consequently, any algorithm for computing simplicial representatives of the generators of $\pi_d(X)$
has time complexity at least $2^{\Omega(\size(X))}$. 
\end{theorem}
}
The second item immediately implies Theorem~\ref{t:optimality}. 

In the first item, we don't assume any certificate
for $1$-connectedness. However, we suspect that any algorithm that computes representatives of $\pi_d(X)$ for simplicial complexes $X$ \emph{must} 
necessarily use some explicit certificate of simple connectivity, but so far we have not been able to verify this.

\begin{lemma}
\label{l:big_homology}
Let $d\geq 2$.
\begin{enumerate}
\item 
There exists a~sequence $\{X_k\}_{k \geq 1}$ of $d$-dimensional $(d-1)$-connected simplicial complexes, such that 
$H_{d}(X_k)\simeq \Z$ 
for all $k$ and for any choice of a~cycle
$z_k\in Z_{d}(X_k)$ generating the homology group, the largest coefficient in $z_k$ grows exponentially in $\size(X_k)$.
\item 
There exists a~sequence $\{X_k\}_{k \geq 1}$ of $d$-dimensional $(d-1)$-connected and $(d-2)$-reduced simplicial sets, such that 
$H_{d}(X_k)\simeq \Z$ 
for all $k$ and for any choice of cycles
$z_k\in Z_{d}(X_k)$ generating the homology, the largest coefficient in $z_k$ grows exponentially\footnote{With a slight abuse of language,
we assume that each $X_k$ not only a~simplicial set but also its algorithmic representation with a specified size such as explained 
in Section~\ref{s:prelim}.} in $\size(X_k)$.

\end{enumerate}
\end{lemma}

\begin{proof}[Proof of Theorem 2 based on Lemma~\ref{l:big_homology}] 
Let $\{X_k\}_{k \geq 1}$ be the sequence of simplicial sets or simplicial complexes from Lemma~\ref{l:big_homology}. 
Since they are $(d-1)$-connected, by the theorem of Hurewicz, 
$\pi_{d}(X_k)\simeq H_{d}(X_k)\simeq \Z$. 
For each $k$, let $\Sigma_k$ be a~simplicial set or simplicial complex with $|\Sigma_k|=S^{d}$, and 
$f_k: \Sigma_k\to X_k$ a~simplicial map representing a~generator of $\pi_{d}(X_k)$. 
The generator of $H_d(\Sigma_d)$ contains each non-degenerate $d$-simplex with a~coefficient $\pm 1$ (this follows from the fact
that $\Sigma_k$ is a triangulation of the $d$-sphere and the $d$-homology of the $d$-sphere is generated by its fundamental class).
The Hurewicz isomorphism $\pi_{d}(X_k)\to H_{d}(X_k)$
maps such a representative to the  formal sum of simplices
$$
f_k \mapsto \sum_{\sigma\,\text{ is a } d-\text{simplex in } (\Sigma_k)} \pm f_k(\sigma) \in C_{d}(X_k) \, ,
$$
which represents a generator of $H_{d}(X_k)$.
It follows from Lemma~\ref{l:big_homology} that the number of $d$-simplices in $\Sigma_k$ grows exponentially in $\size(X_k)$. 
Moreover, the complexity of any algorithm that computes $f_k: \Sigma_k\to X_k$ is at least the size of $\Sigma_k$, which completes the proof.
\end{proof}

It remains to define the sequence from Lemma~\ref{l:big_homology}:
\begin{proof}[Proof of Lemma~\ref{l:big_homology}.] \hfill  
\begin{enumerate}
\item
We begin by constructing for every $d\geq 2$, a~sequence of $\{X_k\}_{k \geq 1}$ of $(d-1)$-connected simplicial complexes, such that 
$H_{d}(X_k)\simeq \Z$ 
for all $k$, and for any choice of a~cycle
$z_k\in Z_{d}(X_k)$ generating the homology group, the largest coefficient in $z_k$ grows exponentially in $\size(X_k)$.

We start with $d=2$.
The idea is to glue $X_k$ out of $k$ copies of a triangulated mapping cylinders of a degree $2$ map ${S}^1 \rightarrow {S}^1$, i.e. $k$ M\"obius bands, and then fill in the two open ends with one triangle each ($A$ and $B$ in Figure~\ref{fig:Moebius}). 
The case $k=1$ is shown in Figure~\ref{fig:Moebius}. For $k \geq 2$, we take $k$ copies of the triangulated M\"obius band and identify the middle circle of each one to the boundary of the next one.

\begin{figure}[t!]
\begin{center}
\includegraphics[scale=0.6]{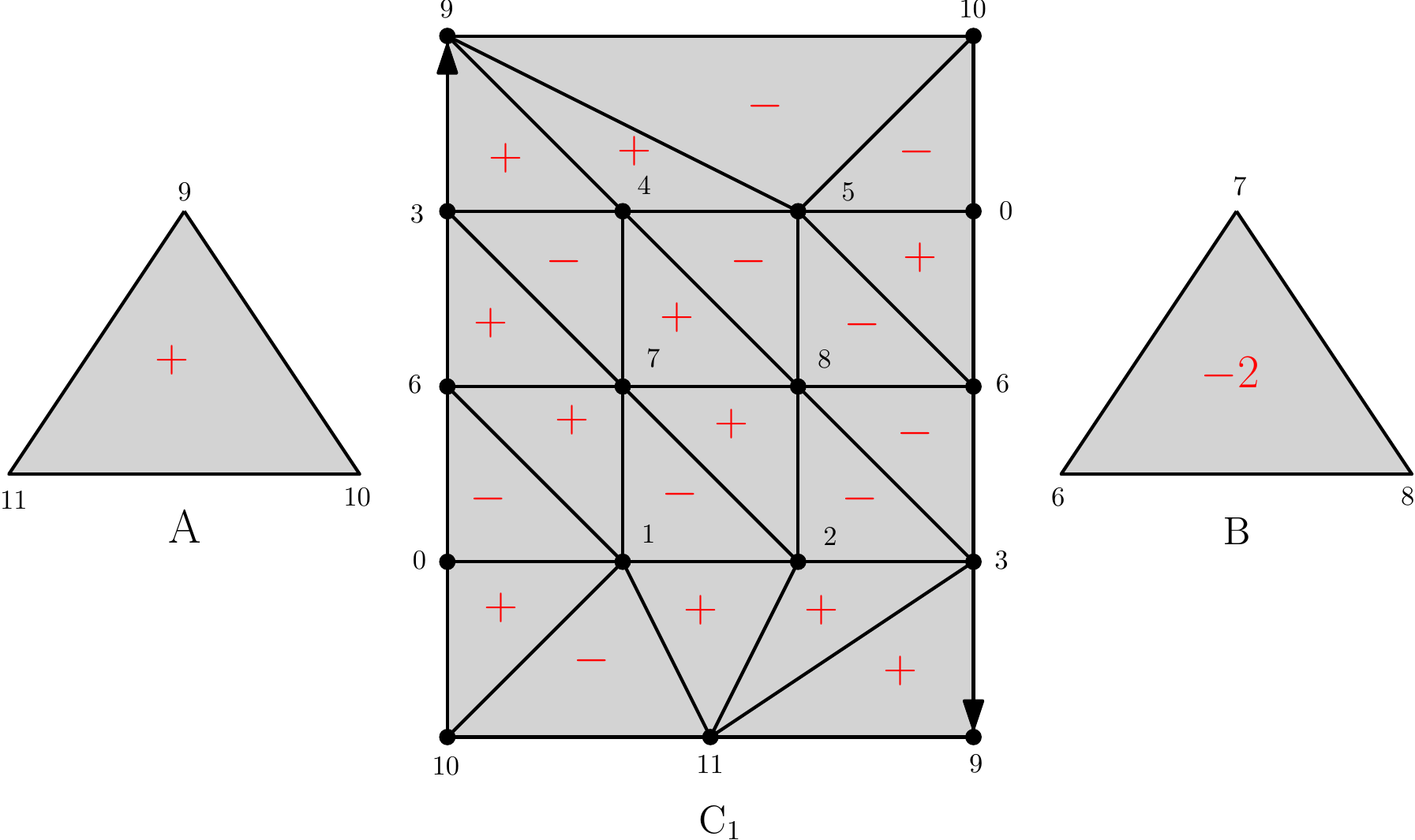}
\end{center}
\caption{The M\"obius band is the mapping cylinder of a degree $2$ map $S^1 \rightarrow S^1$. The triangulation has four layers because starting from the boundary, which is a triangle, we first need to pass to a hexagon in order to cover the middle triangle twice, obtaining the desired degree $2$ map. Connecting $k$ copies of the M\"obius band creates a mapping cylinder of a degree $2^k$ map, using only linearly (in $k$) many simplices. Gluing the full triangles $A$ and $B$ to the ends of this mapping cylinder finishes the construction of $X_k$.
The red coefficients exhibit a generator $\xi$ of $H_2(X_1) = Z_2 (X_1) \simeq \Z $ given as a~formal sum of $2$-simplices.}
\label{fig:Moebius}
\end{figure}

We observe that, up to homotopy equivalence,
$X_k$ consists of a~$2$-disc with another $2$-disc which is attached to it via the boundary map $S^1\to S^1$ of degree $2^k$.
Therefore, $X_k$ is simply connected and has $H_2(X_k) \simeq \mathbb{Z}$ and any homology generator will contain the $2$-simplex
$A$ with coefficient $\pm 1$ and $B$ with coefficient $\pm 2^k$.

Similarly for $d>2$, the simplicial complex $X_k$ is obtained by glueing $k$ copies of a triangulated mapping cylinder of a degree $2$ map 
$S^{d-1} \rightarrow S^{d-1}$, and the two open ends are filled in with two triangulated $d$-balls. 

\item
For every $k \geq 1$ we define the simplicial sets $X_k$ to have one vertex $*$, no non-degenerate simplices up to dimension $d-2$, $k$ non-degenerate $(d-1)$-simplices $\sigma_1,\ldots, \sigma_k$ that are all spherical (that is, for all $i,j$, $d_i \sigma_j=*$ is
the degeneracy of the only vertex of $X_k$), and $k+1 \,$  $d$-simplices $A,C_1,C_2,\ldots, C_{k-1},B$ such that
\begin{itemize}
\item $d_0 A=\sigma_1$, $d_j A=*$ for $j>0$,
\item $d_0 C_i=\sigma_i$, $d_1 C_i=\sigma_{i+1}$, $d_2 C_i=\sigma_i$ and $d_j C_i=*$ for $j>2$, and
\item $d_{0} B=\sigma_k$, $d_j B=*$ for $j>0$.
\end{itemize}
$X_k$ does not have any non-degenerate simplices of dimension larger than $d$. 
The relations of a~simplicial set are satisfied, because $d_i d_j$ is trivial in all cases. 

The boundary operator in the associated normalised chain complex $C_*(X_i)$ 
acts on basis elements as 
\begin{itemize}
\item $\partial A=\sigma_1$
\item $\partial C_i=2\sigma_i-\sigma_{i+1}$, and
\item $\partial B=\sigma_k$.
\end{itemize}
To see that $X_k$ is $(d-1)$-connected for $d>2$, it is enough to prove that $H_{d-1}(X_k)$ is trivial (by $1$-reduceness and Hurewicz theorem). This is true, because $\sigma_1$ is the boundary of $A$
and for $i>1$, $\sigma_i$ is the boundary of the chain $$2^{i-1} A - 2^{i-2} C_1 - \ldots - 2 C_{i-2}-C_{i-1}.$$
In the case $d=2$, $X_k$ is not $1$-reduced, but we can show $1$-connectedness similarly as in the proof of the first part: up to homotopy,
$X_k$ consists of two discs with boundaries together via a~map of degree $2^{k-1}$.

There are no non-degenerate $(d+1)$-simplices, so $H_{d}(X_k)\simeq Z_{d}(X_k)$ and a~simple computation shows that every cycle is a~multiple of
$$
2^{k-1} A - 2^{k-2} C_1 - 2^{k-3} C_2 -\ldots - C_{k-1} - B.
$$
The computer representation of $X_k$ has size that grows linearly with $k$, but 
the coefficients of homology generators grow exponentially with $k$, so they grow exponentially with $\size(X_k)$.
\end{enumerate}
\end{proof}

\heading{Discussion on optimality.}
\label{page:optimality}
If $d=2$ and $X$ is a~$1$-reduced simplicial set, then generators of $H_2(X)$ can be computed via the Smith normal form of the differential
$\partial_3: C_3(X)\to C_2(X)$. Using canonical bases, the matrix of $\partial_3=d_0-d_1+d_2-d_3$ satisfies that the sum of absolute values 
over each column is at most $4$. We were not able to find any infinite family of such matrices so that the smallest coefficient in any set
of homology generating cycles grows exponentially with the size of $X$ (that is, the size of the matrix). 
However, if a~set of homology-generating cycles with subexponential coefficients always exists and can be found algorithmically in polynomial time, our main algorithm given as Theorem~\ref{t:main} is optimal in this case as well. This is because the exponential complexity of the algorithm only appears in the geometric realisation of an element of $GX_1^{sph}$ with large (exponential) exponents (see ``Arrow 3'' in Section~\ref{s:berger}), and the only source of such exponents is the homology $H_1(AX)\simeq H_2(X)$. 

\section{Effective Hurewicz Inverse}
\label{s:berger}
Here we will prove Theorem~\ref{t:eff_hur} by directly describing the algorithm proposed in~\cite{Berger_thesis} and analysing its running time.
\begin{definition}
Let $G$ be a simplicial group. 
Then the Moore complex $\tilde{G}$ is a~(possibly non-abelian) 
chain complex defined by $\tilde{G}_i\defeq G_i\cap (\bigcap_{j>0} \ker{d_j})$ endowed with the differential 
$d_0: \tilde{G}_i\to\tilde{G}_{i-1}$.
\end{definition}
It can be shown that $d_0 d_0=1$ in $\tilde{G}$ and that $\mathrm{Im}(d_0)$ is a~normal subgroup of 
$\ker d_0$ so that the homology $H_*(\tilde{G})$ is well defined. 
\begin{definition}
Let $F$ be a~$0$-reduced simplicial set, $GF$ the associated simplicial group from Def.~\ref{d:G-constr}, and $\widetilde{GF}$ its Moore complex.
We define $AF$ to be the Abelianization of $GF$ and $\widetilde{AF}$ to be the Moore complex of $AF$.
The simplicial group $AF$ is also endowed with a~chain group structure via $\partial = \sum_j (-1)^j d_j$.
If $\sigma\in F_k$, we will denote by $\overline{\sigma}$ the corresponding simplex in $GF_{i-1}$, resp. $AF_{i-1}$.
\end{definition}
Note that, following Def.~\ref{d:G-constr}, the ``last'' differential $d_k \overline{\sigma}$ in $AF_k$ equals $\overline{d_k \sigma} - \overline{d_{k+1}\sigma}$.
Clearly, the Abelianization map $p: GF\to GF/[GF,GF]=AF$ takes $\widetilde{GF}$ into $\widetilde{AF}$. 

Kan showed in~\cite{Kan:Hurewicz} that for $d>1$ and a~$(d-1)$-connected simplicial set $F$, the Hurewicz isomorphism can be identified with
the map $H_{d-1}(\widetilde{GF})\to H_{d-1}(\widetilde{AF})$ induced by Abelianization, whereas these groups are naturally isomorphic 
to $\pi_d(F)$ and $H_d(F)$, respectively. 
Our strategy is to construct maps representing the isomorphisms $1,2,3$ in the commutative diagram
\begin{equation}\label{e:arrows}
\xymatrix{
\pi_{d}(F) && H_{d}(F)\ar@{.>}[ll]_{h^{-1}} \ar[d]^1\\
H_{d-1}(\widetilde{GF}) \ar[u]^{3} &&H_{d-1}(\widetilde{AF}). \ar[ll]^{2}
}
\end{equation}
Here $h$ stands for the Hurewicz isomorphism, $1$ is induced by a homotopy equivalence of chain complexes, $2$ is the inverse of $H_{d-1}(p)$ where 
$p$ is the Abelianization,
and $3$ represents an isomorphism between the $(d-1)$'th homology of $\widetilde{GF}$ (that models the loop space of $F$) and $\pi_{d}(F)$.
The algorithms representing $1,2,3$ will act on representatives, that is, $1$ and $2$ will convert cycles to cycles and $3$ will convert a cycle to a simplicial map
$\Sigma^d\to F$ where $|\Sigma^d|=S^d$. 
In what follows, we will explicitly describe the effective versions of $1,2,3$ and show that the underlying algorithms are polynomial 
for arrows 1,2 and exponential for arrow~3.

\subsection*{Arrow 1.}
Let $F$ be a~$0$-reduced simplicial set, $C_*(F)$ be the (unreduced) chain complex of $F$ and $AF_{*-1}$ the shifted chain complex of $AF$ defined by 
$(AF_{*-1})_i\defeq AX_{i-1}$. 
As a chain complex, $AF_{*-1}$ is a subcomplex of $C_*(F)$ generated by all simplices that are not in the image of the last degeneracy. 
Let $\widetilde{AF}_{*-1}$ be the Moore complex of $AF_{*-1}$.
\begin{lemma}
\label{l:red_to_Moore}
There exists a~polynomial-time strong chain deformation retraction 
$(f,g,h): C_*(F)\to \widetilde{AF}_{*-1}$.  That is, $f: C_*(F)\to \widetilde{AF}_{*-1}$, $g: \widetilde{AF}_{*-1}\to C_*(F)$ are polynomial-time chain-maps 
and $h: C_*(F)\to C_{*+1}(F)$ is a polynomial map such that $fg=\mathrm{id}$ and $gf-\mathrm{id}=h\partial + \partial h$.
\end{lemma}
\begin{proof}
First we will define a chain deformation retraction from $C_*(F)$ to $AF_{*-1}$ represented by $f_0: C_*(F)\to AF_{*-1}$, $g_0: AF_{*-1}\to C_*(F)$ and
$h_0: C_*(F)\to C_{*+1}(F)$.

The chain complex $AF_{*-1}$ consists of Abelian groups $AF_{k-1}$ freely generated by $k$-simplices 
in $F$ that are not in the image of the last degeneracy $s_{k-1}$.
On generators, we define $f_0(\sigma)=\overline{\sigma}$ whenever $\sigma$ is a $k$-simplex not in $\mathrm{Im}(s_{k-1})$ and $f_0(x)=0$ otherwise. 
Deciding whether $\sigma$ is in the image of $s_{k-1}$ amounts to deciding $\sigma=s_{k-1} d_k \sigma$ which can be done in time polynomial
in $\size(I)+\size(\sigma)$, the polynomial depending on $k$. It follows that $f_0$ is a~locally polynomial map.

The remaining maps are defined by $g_0(\overline{\sigma})\defeq \sigma-s_{k-1} d_k \sigma$ and $h_0(\sigma)\defeq (-1)^k s_k \sigma$. These maps are locally polynomial as well
and it is a~matter of straight-forward computations to check that $f_0$ and $g_0$ are chain maps, $f_0 g_0=\mathrm{id}$ and $g_0 f_0-\mathrm{id}=h_0\partial + \partial h_0$.

Further, we define another chain deformation retraction from $AF$ to $\widetilde{AF}$.
For each $p\geq 0$, let $A^p$ be a chain subcomplex of $AF$ defined by

$$(A^p)_k\defeq \{x\in AF_{k}: d_i x=0\quad \text{for} \quad i>\max\{k-p,0\} \,\}$$
that is, the kernel of the $p$ last face operators, not including $d_0$ ($d_i$ refers here to the face operators in $AF$). 
Then $A^{p+1}$ is a chain subcomplex of $A^p$ and we define the maps 
$f_{p+1}: (A^p)_k\to (A^{p+1})_k$ by $f_{p+1}(x)=x-s_{k-p-1} d_{k-p} x$ whenever $k-p>0$, and $f_{p+1}(x)=x$ otherwise; $g_{p+1}: A^{p+1}\to A^p$ will be an inclusion, 
and $h_{p+1}: (A^p)_k\to (A^p)_{k+1}$ via $h_{p+1}(x)=(-1)^{k-p} s_{k-p}x$ if $k-p>0$ and $0$ otherwise. A simple calculation shows that $f_{p+1}, g_{p+1}$ are
chain maps, $f_{p+1} g_{p+1} = \mathrm{id}$, $g_{p+1} f_{p+1}-\mathrm{id}=h_{p+1}\partial + \partial h_{p+1}$ and it is clear that $f_{p+1}, g_{p+1}, h_{p+1}$ 
are polynomial-time maps.

By definition, the Moore complex $\widetilde{AF}=\cap_{p>0} A^p$. The strong chain deformation retraction $(f,g,h)$ from $C_*(F)$ to $\tilde{AF}_{*-1}$ is then defined by
the infinite compositions
\begin{align*}
& f\defeq \ldots f_{k+1} f_k \ldots f_1 f_0 \\
& g\defeq g_0 g_1 \ldots g_k g_{k+1} \ldots
\end{align*}
and the infinite sum
$$
h=h_0 + g_1 h_1 f_1 + (g_1 g_2) h_2 (f_2 f_1) + \ldots 
$$
which are all well-defined, because when applying them to an element $x$, only finitely many of $f_j, g_j$ differ from the identity map and only finitely many $h_j$ are nonzero. 
These are the maps $f,g,h$ from the lemma and we need to show that if the degree $k$ is fixed, then we can evaluate $f,g,h$ on $C_k(F)$ resp. $\tilde{AF}_{k-1}$ in
time polynomial in the input size. However, for fixed $k$ , the definition of $f,g,h$ includes only $f_i, g_i, h_i$ for $i<k$. Then $f,g$ are composed of $k$ polynomial-time maps
and $h$ is a sum of $k$ polynomial-time maps.
\end{proof}

The polynomial-time version of arrow 1 is then induced by applying the map $f$ from Lemma~\ref{l:red_to_Moore}.

\subsection*{Arrow 2.}
\begin{lemma}[Boundary certificate]
\label{l:boundary_cer}
Let $d>1$ be fixed and let $F$ be a~$(d-1)$-connected simplicial set with polynomial-time homology.
There is an algorithm that, for $j<d-1$ and a~cycle $z\in Z_{j}(\widetilde{AF})$, computes an element $c^A(z)\in \widetilde{AF}_{j+1}$ such that
$d_0 c^{A}(z)=z$. The running time is polynomial in $\size(z)+\size(I)$.
\end{lemma}

\begin{proof}
First note that the $(d-1)$-connectedness of $F$ implies that $H_{j+1}(F)\simeq H_{j}(\widetilde{AF})$ are trivial for $j<d-1$, so each cycle in these dimensions is a boundary.

By assumption, $F$ has a~polynomial-time homology, which means that there exists a~globally polynomial-time chain complex $E_*F$, a locally polynomial-time 
chain complex $Y$ and polynomial-time reductions from $Y$ to $C_*(F)$ and $E_*F$
$$
E_* F \lreduP Y \reduP C_*(F).
$$
Let $(f',g',h')$ be chain homotopy equivalence of $Y$ and $\widetilde{AF}_{*-1}$ defined as the composition of $Y\redu C_*(F)$ and the chain homotopy
equivalence of $C_*(F)$ and $\widetilde{AF}_{*-1}$ described in Lemma~\ref{l:red_to_Moore}.
Further, let 
$f'', g'', h''$ be the maps defining the reduction $Y\redu E_*F$: all of these maps are polynomial-time.

Let $j<d-1$ and $z\in Z_j(\widetilde{AF})$, $z=\sum_j k_j y_j$. 
Then the element $f'' g'(z)$
is a cycle in $E_{j+1}F$ and can be computed in time polynomial in $\size(z)+\size(I)$. 
In particular, the size of $f''g'(z)$ is bounded by such polynomial. 
The number of generators of $E_{j+2}F$ and $E_{j+1}F$ is polynomial in $\size(I)$
and we can compute, in time polynomial in $\size(I)$, the boundary matrix of $\partial: E_{j+2}F\to E_{j+1}F$ with respect to the generators.

Next we want to find an element $t\in E_{j+2}F$ such that $\partial t=f''g'(z)$. 
Using generating sets for $E_{j+2}F$, $E_{j+1}F$, this reduces to
a~linear system of Diophantine equations and can be solved in time polynomial
in the size of the $\partial$-matrix and the right hand side $f''g'(z)$~\cite{KannanBachem}.

Finally, we claim that $c^A(z)\defeq f'g''(t)-f'h''g'(z)$ is the desired element mapped to $z$ by the differential in $\widetilde{AF}$.
This follows from a direct computation
\begin{align*}
\partial c^A(z)\defeq  & \partial f'g''(t) - \partial f'h''g'(z) =\\
=&f'g''(\partial t)-\partial f'h'' g'(z) =\\
=&  f'g''f''g'(z)-\partial f'h''g'(z)=\\
=& f'(h''\partial + \partial h''+\mathrm{id})g'(z)-\partial f'h''g'(z)=\\
=& f'h''g'\partial z + \partial f'h''g'(z) + f'g'(z) - \partial f'h''g'(z)=\\
=& 0 + f'g'(z) = z
\end{align*}
The computation of $t$ as well as all involved maps are polynomial-time, hence the computation of $c^A(z)$ is polynomial too. 
\end{proof}

The next lemma will be needed as an auxiliary tool later.
\begin{lemma}
\label{l:comm_decomposition}
Let $S$ be a countable set with a given encoding, $G$ be the free (non-abelian) group generated by $S$, and define 
$\size(\prod_j s_j^{k_j})\defeq \sum_j (\size(s_j)+\size(k_j))$. 
Let $G'\defeq [G,G]$ be its commutator subgroup.

Then there exists a~polynomial-time algorithm that for an element 
$ g=\prod_j s_j^{k_j} $ in $G'\subseteq G$, computes elements $a_i, b_i\in G$ such that $g=\prod_j [a_j, b_j]$.
\end{lemma}
In other words, we can decompose commutator elements into simple commutators in polynomial-time at most.
\begin{proof}
Let us choose a~linear ordering on $S$ and let $g=\prod_j s_j^{k_j}$ be in $G'$: that is, for each $j$, the exponents $\{k_{j'}:\,s_{j'}=s_j\}$
sum up to zero. We will present a~bubble-sort type algorithm for sorting elements in $g$. 
Going from the left to right, we will always swap $s_j^{k_j}$ and $s_{j+1}^{k+1}$ whenever $s_{j+1}<s_j$. Such swap always creates a commutator, but that will immediately be moved to the initial segment of commutators. 

More precisely, assume that Init is the initial segment, $x=s_j^{k_j}$ and $y=s_{j+1}^{k_{j+1}}$ should be swapped, Rest
represent the segment behind $y$, and Commutators is a final segment of commutators.
The swapping will consists of these steps:
\begin{align*}
& \hbox{}\hskip 15pt \text{Init} \,\, x \,\, y \,\, \text{Rest} \,\, \text{Commutators} \\
& \mapsto \text{Init} \,\, y \,\, x \,\, [x^{-1}, y^{-1}] \,\,  \text{Rest} \,\,\text{Commutators} \\
& \mapsto \text{Init} \,\, y \,\, x \,\, \text{Rest} \,\, \big([x^{-1}, y^{-1}] \,\, [[y^{-1}, x^{-1}], \text{Rest}^{-1}] \,\, \text{Commutators} \big)
\end{align*}
where the parenthesis enclose a new segment of commutators. Before the parenthesis, $x$ and $y$ are swapped. 
Each such swap requires enhancing the commutator section by two new commutators of size at most $\size(g)$, hence each such swap
has complexity linear in $\size(g)$. 

Let as call everything before the commutator section a~``regular section''.
Going from left to right and performing these swaps will ensure that the largest element will be at the end of the regular section.
But no later then that, the largest element $y_{\text{largest}}$ disappears from the regular section completely, because all of its exponents add up to $0$.
Again, starting from the left and performing another round of swaps will ensure that the second-largest elements disappear from the regular section;
repeating this, all the regular section will eventually disappear which will happen in at most $\size(g)^2$ swaps in total. Each swap has complexity linear
in $\size(g)$ and the overall time complexity is not worse than cubic.
\end{proof}

\begin{lemma}
\label{l:simple_contraction}
Assume that $F$ is a parametrized simplicial set with polynomially contractible loops.
Let $k>0$, $\gamma\in GF_k$ be spherical and $\alpha\in GF_k$ is arbitrary. There is a~polynomial-time algorithm that computes $\delta\in GF_{k+1}'$ such that
$d_0 \delta=[\alpha, \gamma]$ and $d_i \delta=1$ for all $i>0$.
\end{lemma}
In other words, a~simple commutator of a spherical element with another element can always be ``contracted'' in $GF'$ in polynomial time. 
Our proof roughly follows the construction in Kan~\cite[Sec. 8]{Kan:Hurewicz}
\begin{proof}
For $x\in GF_0$, we will denote by $c_0 x$ the element of $\widetilde{GF}_1$ such that $d_0 c_0 x=x$: this can be computed in polynomial-time 
by the assumption on polynomial loop contractions. For the simplex $\alpha\in GF_k$, we define $(k+1)$-simplices $\beta_0,\ldots, \beta_k$ by
$\beta_k\defeq s_0^k c_0 d_0^k \alpha$ and inductively $\beta_{j-1}\defeq (s_j d_j \beta_j) \,\cdot \, (s_j \alpha^{-1})\,\cdot\,(s_{j-1} \alpha)$ for $j<k$.
Then the following relations hold:\footnote{Kan uses a~slightly different convention in~\cite{Kan:Hurewicz} but the
resulting properties are the same. The sequence $\beta_0,\ldots,\beta_k$ can be interpreted as a discrete path from $\alpha$ to the identity element.}
\begin{itemize}
\item $d_0 \beta_0=\alpha$.
\item $d_j \beta_j=d_j \beta_{j-1}$, $1\leq j\leq k$
\item $d_{k+1} \beta_k=1$.
\end{itemize}
The second and third equations are a matter of direct computation, while the first follows from the more general relation $d_0^{j+1} \beta_j=d_0^j \alpha$
which can be proved by induction. If $k$ is fixed, then all $\beta_0,\ldots,\beta_k$ can be computed in polynomial time.

The desired element $\delta$ is then the alternating product
$$
\delta\defeq [\beta_0, s_0 \gamma] \, [\beta_1, s_1\gamma]^{-1} \, \ldots \,[\beta_k, s_k\gamma]^{\pm 1}.
$$
\end{proof}
\begin{lemma}
\label{l:contraction}
Under the assumptions of Theorem~\ref{t:eff_hur}, there exist homomorphisms $c_j: {GF}_j\to {GF}_{j+1}$ for $0\leq j < d-1$
such that 
\begin{enumerate}
\item $d_0 c_j=\mathrm{id}$,
\item $d_i c_j=c_{j-1}d_{i-1}$, $0<i\leq j+1$, and
\item $c_{j} s_{i} =  s_{i+1} c_{j-1}$ for $0<j<d-1$ and $0\leq i< j$,
\item $d_1 c_0(x)=1$ for all $x\in GF_0$.
\end{enumerate}
If $d$ is fixed and $x\in GF_j$, $j<d-1$,
then $c_j(x)$ can be computed in polynomial time.
\end{lemma}

\begin{proof}
The homomorphism $c_0$ can be constructed directly from the assumption on polynomial contractibility of loops.
We have a~canonical basis of $GF_0$ consisting of all non-degenerate $1$-simplices of $F$. For $\sigma\in F_1$, we denote by $\overline{\sigma}$ the corresponding
generator of  $GF_0$.
The we define $c_0(\prod \overline{\sigma}_j^{k_j})$ to be $\prod b_j^{k_j}$ where $b_j$ is the element of $GF_1$ such that
$d_0 b_j=\overline{\sigma}_j$ and $d_1 b_j=1$.

In what follows, assume that $1\leq k<d-1$ and $c_i$ have been defined for all $i<k$. 
We will define $c_k$ in the following steps.

{\bf Step 1.} Contractible elements. 

Let $x\in GF_k$. We will say that $x$ is \emph{contractible} and $y\in GF_{k+1}$ is a~\emph{contraction of $x$}, if $d_0 y=x$ and $d_i y=c_{k-1} d_{i-1} x$ for all $i>0$.

The general strategy for defining $c_k$ will be to find a contraction $h$ for each basis element ($(k+1)$-simplex) $g\in GF_k$ and define $c_k(g)\defeq h$.
This will enforce properties $1$ and $2$. Moreover, in case when $g$ is degenerate, the contraction will be chosen in such a way that property $3$ holds too;
otherwise it holds vacuously. 
Property $4$ only deals with $c_0$ and is satisfied by the definition of loop contraction (a~polynomial-time $c_0$ 
is given as an input in Theorem~\ref{t:eff_hur}).

{\bf Step 2.} Contraction of degenerate elements.

Let $g=s_i y$ be a basis element of $GF_k$, $y\in GF_{k-1}$. Then $g$ can be uniquely expressed as $s_{j} z$ where $j$ is the maximal $i$ such that $g\in \mathrm{Im}(s_i)$.
We then define $c_k(g)\defeq s_{j+1} c_{k-1}(z)$. Note that
$$d_0 c_k(g)=d_0 s_{j+1} c_{k-1}(z)=s_j d_0 c_{k-1}(z)=s_j z=g,$$
so property $1$ is satisfied. To verify property $2$, first assume that $i\in \{j+1, j+2\}$. Then we have
$$
d_i c_k(g)= d_i s_{j+1} c_{k-1} (z)=c_{k-1} (z)=c_{k-1} d_{i-1} s_{j} z=c_{k-1} d_{i-1} g.
$$
This fully covers the case $k=1$, because then the only possibility is $j=0$ and $i\in \{1,2\}$.
Further, let $k>1$. If $i\leq j$, then we have
\begin{align*}
d_i c_k g=&  d_i c_k s_j z=d_i s_{j+1} c_{k-1}(z)=s_j d_i c_{k-1}(z) = s_j c_{k-2} d_{i-1} z=\\
=&c_{k-1} s_{j-1} d_{i-1} z =c_{k-1} d_{i-1} s_j z=c_{k-1} d_{i-1} g
\end{align*}
and if $i>j+2$, then the computation is completely analogous, using the relation $d_i s_{j+1}=s_{j+1} d_{i-1}$ instead.

So far, we have shown that $c_k(g)\defeq s_{j+1} c_{k-1} g$ is a contraction of $g$. It remains to show property $3$.
That is, we have to show that if $g=s_j z$ can also be expressed as $s_i y$, then $c_k(s_i y)=s_{i+1} c_{k-1} y$. 

The degenerate element $g$ has a unique expression $g=s_{i_u} \ldots s_{i_1} s_{i_0} v$ where $i_0<i_1<\ldots <i_u=j$ and
is expressible as $s_i x$ iff $i=i_j$ for some $j=0,1,\ldots, u$. Choosing such $i<j$, we can rewrite $g$ as
$g=s_j s_i v$ for some $v$ and then $g=s_i s_{j-1} v$, so that $y=s_{j-1} v$ and $z=s_i v$.
Then we again use induction to show
\begin{align*}
c_k(s_i y)=&s_{j+1} c_{k-1}(z)=s_{j+1} c_{k-1} s_i v=s_{j+1} s_{i+1} c_{k-2} v=  \\
=& s_{i+1} s_j c_{k-2} v=s_{i+1} c_{k-1} s_{j-1} v=s_{i+1} c_{k-1} y
\end{align*}
as required.

{\bf Step 3.} Decomposition into spherical and conical parts. 

We will call an element $\hat{x}\in GF_k$ to be \emph{conical}, if it is a product of elements that are either degenerate or in the image of $c_{k-1}$.
Let $x\in GF_k$ be arbitrary.
We define $x_k\defeq x$ and inductively $x_{i-1}\defeq x_i (s_{i-1} d_i x_i)^{-1}$. In this way we obtain $x_0,\ldots, x_n$
such that $x_i$ is in the kernel of $d_j$ for $j>i$ and $x=x_0 y$ where $y$ is a product of degenerate simplices. Further, let $x^s\defeq x_0 (c_{k-1} d_0 x_0)^{-1}$.
A simple computation shows that $x^s$ is \emph{spherical}, that is, $d_i x^s=1$ for all $i$. 
We obtain an equation $x=x^s \hat{x}$ where $\hat{x}=(c_{k-1} (d_0 x_0) y$; this is a decomposition of $x$ into a~spherical part $x^s$ and a~conical element $\hat{x}$.

We will define $c_k$ on non-degenerate basis elements $g=\overline{\sigma}$ by first decomposing 
$g=g^S \hat{g}$ into a~spherical and conical part, finding contractions $h_1$ of $g^S$ and $h_2$ of $\hat{g}$, and defining $c_k(g)\defeq h_1 h_2$. Then $c_k(g)$
is a contraction of $g$ and hence satisfies properties $1$ and $2$: property $3$ is vacuously true once $g$ is non-degenerate.

{\bf Step 4.} Contraction of the conical part.

Let $\hat{x}\defeq c_{k-1} (d_0 x_0)\, y$ be the conical part defined in the previous step. By construction, $x_0\in \tilde{GF}_k$ and $y$ is a product of degenerate elements
$s_{i_1} u_1 \ldots s_{i_l} u_l$.
We define the contraction of $c_{k-1} (d_0 x_0)$ to be 
$$\tilde{c}_{k}(c_{k-1}(d_0 x_0))\defeq s_0 c_{k-1}(d_0 x_0).$$ Note that this satisfies property $1$ as
$d_0 \tilde{c}_k c_{k-1} (d_0 x_0)=c_{k-1} (d_0 x_0)$. For property $2$, we first verify 
$$d_1 \tilde{c}_k c_{k-1} (d_0 x_0)=d_1 s_0 c_{k-1}(d_0 x_0)=c_{k-1}(d_0 x_0)=c_{k-1} d_0 c_{k-1} (d_0 x_0).$$
Not let $i\geq 2$. If $k=1$, then the remaining face operator is $d_2$ and we have
$$
d_2 \tilde{c}_1 c_0 (d_0 x_0)=d_2 s_0 c_0(d_0 x_0)=s_0 d_1 c_0(d_0 x_0)=1=c_0 d_1 c_0(d_0 x_0)
$$
using axiom $4$ for $c_0$. Finally, if $i\geq 2$ and $k\geq 2$, we have 
\begin{align*}
d_i \tilde{c}_k c_{k-1} (d_0 x_0)=&d_i s_0 c_{k-1}(d_0 x_0)=s_0 d_{i-1} c_{k-1} (d_0 x_0)=s_0 c_{k-1} d_{i-2} d_0 x_0 =\\
=& s_0 c_{k-1} d_0 d_{i-1} x_0=s_0 c_{k-1} d_0 1= 1 = c_{k-1} c_{k-2} d_0 d_{i-1} x_0=\\
=& c_{k-1} c_{k-2} d_{i-2} d_0  x_0=c_{k-1} d_{i-1} c_{k-1}(d_0 x_0),
\end{align*}
where we exploited the fact that $x_0\in \widetilde{GF}_k$ and hence $d_{u} x_0=1$ for $u\geq 2$.

The contraction of degenerate elements $y$ has already been defined in Step 2, so we can define a contraction of $c_{k-1}(d_0 x_0) y$ to be
$s_0 c_{k-1}(d_0 x_0)\,c_k(y)$.

{\bf Step 5.} Contraction of commutators. 

Let $g'\in GF_k'$ be an~element of the commutator subgroup. By Lemma~\ref{l:comm_decomposition}, we can algorithmically
decompose $g'$ into a product of simple commutators, so to find a contraction of $g'$, it is sufficient to 
find a~contraction of each simple commutator $[x,y]$ in this decomposition.

Let $x=x^S\,\hat{x}$ and $y=y^S\,\hat{y}$ be the decompositions into spherical and conical parts described in Step 3.
Using the notation $^b a\defeq bab^{-1}$, we can decompose $[x,y]$
as follows~\cite[p. 60]{Berger_thesis}:
\begin{equation}
\label{e:decom_com}
[x,y]=\,([x,y] [\hat{y}, x])\,([x,\hat{y}] [\hat{y}, \hat{x}]) \,[\hat{x}, \hat{y}] = 
[^{xy} x^{-1}, ^{xy}(y^{-1} \hat{y})] \,[^x\hat{y}, ^x(x^{-1}\hat{x})]\,[\hat{x},\hat{y}].
\end{equation}
Both $x^{-1}\hat{x}$ and $y^{-1}\hat{y}$ are spherical simplices and so are their conjugations. It follows that equation (\ref{e:decom_com}) can be rewritten to
$[x,y]=[\alpha_1, \gamma_1]\,[\alpha_2,\gamma_2]\,[\hat{x}, \hat{y}]$ where $\gamma_{1}$ and $\gamma_2$ are spherical. All of these decompositions
are done by elementary formulas and are polynomial-time in the size of $x$ and $y$.

By Lemma~\ref{l:simple_contraction} we can find an elements $\lambda_i\in \widetilde{GF}_{k+1}$ such that $d_0 \lambda_i=[\alpha_i,\gamma_i]$,
$i=1,2$, in polynomial time. Further, both $\tilde{x}$ and $\tilde{y}$ are conical and they are in the form $\tilde{x}=c_0 (d_0 x_0) x_{deg}$ 
where $x_0\in \widetilde{GF}_k$ and $x_{deg}$ is degenerate; similar decomposition holds for $y$. In Step 4 we showed how to compute elements $c^x$ and $c^y$
such that $c^x$, $c^y$ is a contraction of $\hat{x}$, $\hat{y}$, respectively. Then $[c^x, c^y]$ is a contraction of $[\hat{x}, \hat{y}]$ and
$ \lambda_1 \lambda_2 [c^x, c^y] $ is a contraction of $[x,y]$.

{\bf Step 6.} Contraction of spherical elements.

The last missing step is to compute a~contraction of the spherical element $g^S$ where $g^S$ is the spherical part of a~basis element $g\in GF_k$.

Let us denote by $p$ the projection $GF\stackrel{p}{\to} AF$. 
The projection $z\defeq p(g^S)$ is in the kernel of all face operators and hence a~cycle in $\widetilde{AF}_{k}$. By Lemma~\ref{l:boundary_cer}, 
we can compute $t\defeq c_k^A(z)\in \widetilde{AF}_{k+1}$ such that $d_0 t=z$,
in polynomial time.  Let $h\in {GF}_{k+1}$ be any $p$-preimage\footnote{For $t=\sum_j k_j \overline{\sigma}_j$, 
we may choose $h=\prod_j \overline{\sigma}_j^{k_j}$ (choosing any order of the simplices).} of $t$.
Let $h_k\defeq h$ and inductively define $h_{j-1}\defeq h_{j} (s_{j-1} d_j h_{j})^{-1}$ for $j<k$. Then $h_0$ is in the kernel of all faces except $d_0$, that is,
$h_0\in \widetilde{GF}_{k+1}$. It follows that $p(h_0)\in \widetilde{AF}_{k+1}$ is in the kernel of all faces except $d_0$.
We claim that $p(h_0)=t$.This can be shown as follows: assume that $p(h_j)=t$, then $p(h_{j-1})=p(h_j) + p(s_{j-1} d_j h_j^{-1})=t + s_{j-1} d_j t=t+0=t$.

We have the following commutative diagram:
\[
\xymatrix@R=1em{
{} & h_0 \ar@{|->}[r] & t\\
\widetilde{GF}_{k+1}' \ar@{^{(}->}[r] \ar[d]^{d_0}& \widetilde{GF}_{k+1} \ar@{->>}[r]^p \ar[d]^{d_0}& \widetilde{AF}_{k+1} \ar[d]^{d_0}\\
\widetilde{GF}_{k}'  \ar@{^{(}->}[r] & \widetilde{GF}_{k} \ar@{->>}[r]^p & \widetilde{AF}_{k}\\
{} & g^S  \ar@{|->}[r] & z
}
\]
Both $g^S$ and $d_0 h_0$ are mapped by $p$ to the same element $z$: it follows that $g^S (d_0 h_0)^{-1}$ is mapped by $p$ to zero and hence is 
an element of the commutator subgroup.
Let $\tilde{h}$ be the contraction of $g^S (d_0 h_0)^{-1} $, computed in Step 5, and finally let $h:=\tilde{h} h_0$.
Then $h$ is an element of $\widetilde{GF}_{k+1}$ and a direct computation shows that
$d_0 {h}=g^S$ as desired.

This completes the construction of $c_k$: for each non-degenerate basis element $g$ of $GF_k$, $c_k(g)$ is defined to be the product of the contraction of $g^S$ and 
the contraction\footnote{The connectivity assumption on $F$ was exploited in the existence of the contraction $c_j^A$ on the Abelian part.} of~$\hat{g}$.

All the subroutines described in the above steps are polynomial-time. Thus we showed that 
if there exists a~polynomial-time algorithm for $c_{k-1}$, then there also exists a~polynomial-time algorithm for $c_k$. 
The existence of a~polynomial-time $c_0$ follows from the assumption on polynomial loop contractibility and $d$ is fixed, 
thus there exists a~polynomial-time algorithm that for $x\in GF_{j}$ computes $c_j(x)$ for each $j<d-1$.
\end{proof}

\begin{lemma}[Construction of arrow 2]\label{l:arrow2}
Under the assumption of Theorem~\ref{t:eff_hur}, let $z\in Z_{d-1}(\widetilde{AF})$ be a~cycle. Then there exists a~polynomial-time algorithm
that computes a~cycle $x\in Z_{d-1}(\widetilde{GF})$ such that the Abelianization of $x$ is $z$.
\end{lemma}
The assignment $z\mapsto x$ is hence an effective inverse of the isomorphism 
$$H_{d-1}(\widetilde{GF})\to H_{d-1}(\widetilde{AF})$$ on the level of representatives.
\begin{proof}
Let $c_{d-2}$ be the contraction from Lemma~\ref{l:contraction} and $z\in Z_{d-1}(\widetilde{AF})$ be a cycle. 
First choose $y\in GF_{d-1}$ such that $p(y)=z$. Creating the sequence $y_n\defeq y$, $y_{j-1}\defeq y_j s_{j-1} d_j y_j^{-1}$ for decreasing $j$,
yields an element $y_0\in \widetilde{GF}_{d-1}$ that is still
mapped to $z$ by $p$, similarly as in Step 4 of Lemma~\ref{l:contraction}. The equation $p d_0 (y_0)=d_0 p (y_0)=d_0 z=0$ shows that
$d_0 y_0$ is in the commutator subgroup $\widetilde{GF}_{d-2}'$. We define $x\defeq y_0 c_{d-2} (d_0 y_0)^{-1}$: this is already a cycle in $\widetilde{GF}_{d-1}$
and $p(x)=p(y_0)=z$.
\end{proof}
\subsection*{Arrow 3.}
The construction of map $3$ is one of the main results from \cite{Berger_paper} and involves further definitions. Here, we describe the main points of the construction only while details are given in later sections.

Given a 0-reduced simplicial set $F$, there exists a simplicial group $\oLoop F$ that is a discrete analog of a loopspace of $F$ i.e. 
$\pi_{d-1} (\oLoop F) \cong \pi_{d}(F)$. Further, there is a homomorphism of simplicial groups $t\:GF \to \oLoop F$ that induces an isomorphism on the level of homotopy groups. 
This is described in~\cite[Proposition 3.3]{Berger_paper}. 

The homomorphism $t$ is given later by formula \eqref{e:homomorphismt} and the simplicial set $\oLoop F$ is described in the next section. Here, we remark that the size of $t(g)$ is exponential in size of~$g$. 

Finally, Lemma~\ref{l:arrow3} describes an algorithm that for a~spherical element $\gamma \in \oLoop F_{d-1}$ constructs a simplicial map $\gamma_\sph\: \Sigma^{d}(\gamma) \to F$ such that $\pi_{d-1}(\oLoop F)\ni [\gamma] \simeq [\gamma_\sph] \in \pi_{d}(F)$.

The size of $\gamma_\sph$ is polynomial in $\size(\gamma)$. Hence, given a spherical $g \in \widetilde{GF}_{d-1}$, the algorithm produces $t(g)_\sph\:\Sigma^{d}(t(g)) \to F$ that is exponential with respect to $\size (g)$.
\heading{Berger's model of the loop space.}
\begin{definition}[Oriented multigraph on $X_n$]\label{def:multigraph}
Let $X$ be a 0-reduced simplicial set. We define a directed multigraph $MX_n = (V_n,E_n)$, where the set of vertices $V_n = X_n$ and the set of edges $E_n$ is given by
\[
E_n = \{[x, i]^\epsilon \mid x\in X_{n+1}, 0\leq i \leq n, \epsilon \in \{1, -1\}\}.
\]
We define maps $\source, \target\: E_n \to V_{n}$ by setting $\source [x,i] = d_{i+1} x$, $\target [x,i] = d_i x$ and $\source [x,i]^{-1} = \target [x,i]$ and $\target [x,i]^{-1} = \source [x,i]$.

An edge $[x, i]^\epsilon \in E_n$ is called \emph{compressible}, if $x = s_i x'$ for some $x' \in X_{n}$.
\end{definition}
\begin{definition}[Paths]
Let $X \in \sSet$. A sequence of edges in $MX_{n}$ 
\begin{equation}
\label{e:gamma}
\gamma = [x_1,i_1]^{\epsilon_1}[x_2,i_2]^{\epsilon_2}\cdots [x_k,i_k]^{\epsilon_k}
\end{equation}
is called an $n$-\emph{path}, if  
$\target [x_j,i_j]^{\epsilon_j} = \source [x_{j+1},i_{j+1}]^{\epsilon_{j+1}}$, $1\leq j<k$.

Moreover, for every $x\in V_n = X_n$ we define a path of length zero $1_x$ with the property $\source 1_x = x = \target 1_x$ and relations
$a 1_x=a$ whenever $\target{a}=x$ and $1_x b=b$ whenever $\source{b}=x$.

The set of paths on $MX_{n}$ is denoted by 
$IX_n$. Let $\gamma \in IX_n$ by as in (\ref{e:gamma}).
We define $\source \gamma = \source [x_1,i_1]^{\epsilon_1}$ and $\target \gamma =\target [x_k,i_k]^{\epsilon_k}$. 
The \emph{inverse} of $\gamma$, denoted $\gamma^{-1}$, is defined as
\[
\gamma^{-1} = [x_k,i_k]^{- \epsilon_k}\cdots  [x_1,i_1]^{ - \epsilon_1}.\]
if $\gamma$ = $1_x$, then $\gamma^{-1} = \gamma$. Note that each path is either equal to $1_x$ for some $x$ or can be represented
in a form such as (\ref{e:gamma}), without any units.
\end{definition}
For algorithmic purposes, we assume that a~path 
$\gamma = [x_1,i_1]^{\epsilon_1}[x_2,i_2]^{\epsilon_2}\cdots [x_k,i_k]^{\epsilon_k}$
is represented as a~list of triples $(x_j, i_j, \epsilon_j)$ and has size 
$$\size(\gamma):=\sum_j \size(x_j)+\size(i_j)+\size(\epsilon_j),$$
which is bounded by a~linear function in~$\sum_j \size({x_j})$.

Given an edge $[x,i]^\epsilon \in MX_n$, we define operators $$d_0, \ldots d_{n}\: E_n \to IX_{n-1} \,\, \text{and}\,\, s_0, \ldots, s_n\:  E_n \to IX_{n+1}$$ called \emph{face} and \emph{degeneracy} operators, respectively. These are given as follows
\[
d_j [x,i]^\epsilon = \left\{
\begin{array}{lll}
{[}d_j x, i - 1{]}^\epsilon,& & j< i;\\ 
1_{d_ i d_{i+1} x},&  & i=j;\\
{[}d_{j+1} x, i{]}^\epsilon,& & j>i.
\end{array}
\right. \qquad
s_j [x,i]^\epsilon = \left\{
\begin{array}{lll}
{[}s_j x, i + 1{]}^\epsilon,& & j< i;\\ 
{[}s_i x, i+1{]}{[}s_{i+1} x, i{]})^\epsilon,&  & i=j;\\
{[}s_{j+1} x, i{]}^\epsilon,& & j>i.
\end{array}
\right.
\]

One can now extend the definition of face and degeneracy operators to paths, i.e. we define operators $d_0, \ldots d_{n}\: IX_n \to IX_{n-1}$ and $s_0, \ldots, s_n\:  IX_n \to IX_{n+1}$ 
\[
d_j \gamma = \left\{
\begin{array}{lll}
d_j([x_1,i_1]^{\epsilon_1}) d_j([x_2,i_2]^{\epsilon_2})\cdots d_j([x_k,i_k]^{\epsilon_k})& \text{if} & \gamma =  [x_1,i_1]^{\epsilon_1}[x_2,i_2]^{\epsilon_2}\cdots [x_k,i_k]^{\epsilon_k},\\
1_{d_j x}& \text{if} &  \gamma = 1_x, x\in X_n.
\end{array}
\right.
\]
\[
s_j \gamma = \left\{
\begin{array}{lll}
s_j([x_1,i_1]^{\epsilon_1}) s_j([x_2,i_2]^{\epsilon_2})\cdots s_j([x_k,i_k]^{\epsilon_k})& \text{if} & \gamma =  [x_1,i_1]^{\epsilon_1}[x_2,i_2]^{\epsilon_2}\cdots [x_k,i_k]^{\epsilon_k}\\
1_{s_j x}& \text{if} &  \gamma = 1_x, x\in X_n.
\end{array}
\right.
\]
With the operators defined above, one can see that $IX$ is in fact a simplicial set. 

For any $\gamma, \gamma' \in IX$ such that $\target{\gamma} = \source{\gamma'}$, we define a composition $\gamma\cdot \gamma'$ in an obvious way.

If the simplicial set $X$ is $0$-reduced, we denote the unique basepoint $\pt\in X_0$. Abusing the notation, we denote the iterated degeneracy of the basepoint $\underbrace{{s_0} \cdots s_0 \pt}_{k-\text{times}}\in X_{k}$ by~$\pt$ as well. With that in mind, we define simplicial subsets $P X$, $\Omega X$ of $IX$ as follows:
\[
PX= \{\gamma \in IX \mid \target\gamma = \pt\} \quad \Omega X= \{\gamma \in IX \mid \source \gamma = \pt = \target \gamma \}.
\]
We remark that simplicial sets $PX, \Omega X$ intuitively capture the idea of pathspace and loopspace in a simplicial setting.

\begin{definition}
A path $\gamma = [x_1,i_1]^{\epsilon_1}[x_2,i_2]^{\epsilon_2}\cdots [x_k,i_k]^{\epsilon_k} \in IX$ is called \emph{reduced}, if for every $1\leq j<k$ the following condition holds:
\[
(x_j = x_{j+1}\,\& \,i_j = i_{j+1}) \Rightarrow \epsilon_j = \epsilon_{j+1}.
\]
e.g. an edge in the path $\gamma$ is never followed by its inverse. 

An edge $[x, i]^\epsilon \in E_n$ is called \emph{compressible}, if $x = s_i x'$ for some $x' \in X_{n}$. A path is \emph{compressed} if it does not contain any compressible edge. 
\end{definition}
We define relation $\sim_R$ on $IX$ (or rather on each $IX_n$) as a relation generated by
\[ [x,i]^{\epsilon} [x,i]^{-\epsilon} \sim_R 1_{\source( [x,i]^{\epsilon})},\quad  n\in \mathbb{N}_0, [x,i]^\epsilon \in E_n.\]
Similarly, we define  $\sim_C$ on $IX$ as a relation generated by
\[
 [x,i]^{\epsilon} \sim_C 1_{\source( [x,i]^{\epsilon})},\quad \text{if}\,  [x,i]^{\epsilon} \in E_n \,\text{is compressible}.
\]
We finally define $\oI X = (IX/\sim_C)/ \sim_R$. Similarly, one defines $\oP X, \oLoop X$.

For $\gamma, \gamma'\in IX_n$, we write $\gamma \sim \gamma'$ if they represent the same element in $\oI X_n$. The symbol $\overline{\gamma}$, denotes the 
(unique) compressed and reduced path such that $\gamma \sim \overline{\gamma}$. 
One can see $\oI X$ ($\oP X, \oLoop X$) as the set of reduced and compressed paths in $IX (P X, \Omega X)$.

In a natural way, we can extend the definition of face and degeneracy operators $d_i, s_i$ on sets $\oI X (\oP X$,$\oLoop X$) by setting $d_i \gamma = \overline {d_i \gamma}$ and $s_i \gamma = \overline {s_i \gamma}$. One can check that this turns $\oI X$, $ \oP X$ and $\oLoop X$ into simplicial sets.

Similarly, we define operation $\cdot \: \oLoop X_n \times \oLoop X_n \to \oLoop X_n$ by $\gamma\cdot \gamma' \mapsto \overline{\gamma\gamma'}$, i.e. we first compose the loops and then assign the appropriate compressed and reduced representative. With the operation defined as above, $\oLoop X$ is a simplicial group.
\heading{Homomorphism $t\: GX \to \oLoop X$.}
We first describe how to any given $x \in X_n$ assign a path $\gamma_x \in \oP X_n$ with the property $\source \gamma_x = x$ and $\target \gamma_x = *$:

For $x \in X_n$, $n>0$, the $0$-reducedness of $X$ gives us $d_{i_1} d_{i_2}\cdots d_{i_{n}} x = \pt$, here $i_{j} \in \{0, \ldots, j\}$, $0<j\leq n$. In particular, $d_{0} d_{1}\cdots d_{n-1} x = \pt$. Using this, we define
\[
\gamma_x = [s_{n}x,{n-1}][s_{n}s_{n-1}d_{n-1}x,{n-2}]\cdots[s_{n}s_{n-1}\cdots s_{1}d_1 d_{2}\cdots d_{n-1}x,0].
\]
Ignoring the degeneracies, one can see the sequence of edges as a path 
\[x \to d_{n-1}x \to d_{n-2}d_{n-1}x \to \cdots \to d_{0} d_{1}\cdots d_{n-1} x.\]

We define the homomorphism $t$ on the generators of $GX_n$, i.e. on the elements $\overline{x}$, where $x \in X_{n+1}$ as follows:
\begin{equation}\label{e:homomorphismt}
 t(\overline x) =  \overline{\gamma_{d_{n+1}x}^{-1}[x,n]\gamma_{d_{n}x}}.
\end{equation}
This is an element of $\oLoop X_n$.

The algorithm representing the map $t$ has \emph{exponential time complexity} due to the fact that an~element $\overline{\sigma}^{k}$
with size $\size(\sigma)+\size(k)$ is mapped to 
$$
\underbrace{\overline{\gamma_{d_{n+1}x}^{-1}[x,n]\gamma_{d_{n}x} \,\,\, \ldots \,\,\,\gamma_{d_{n+1}x}^{-1}[x,n]\gamma_{d_{n}x}}}_{k \,\,\text{times}}
$$
which in general can have size proportional to $k$. Assuming an encoding of integers such that $\size(k)\simeq \ln(k)$, this amounts
to an exponential increase.

\heading{Universal preimage of a path.}
Intuitively, one can think of the simplicial set $IX$ of paths as of a discretized version of space of continuous maps $|X|^{[0,1]}$. 
In particular,  $\gamma \in IX_{d-1}$ is a walk through a sequence of $d$-simplices in $X$ that connect $\source \gamma$ with $\target \gamma$. 
However, in the continuous case an element $\mu \in |X|^{[0,1]}$ corresponds to a continuous map $\mu \: [0,1] \to |X|$. 
We want to push the parallels further, namely, given any nontrivial\footnote{By nontrivial we mean that $\gamma\neq 1_x$ for any $x\in X_{d-1}$.}
$\gamma \in IX_{d-1}$, we aim to define a simplicial set $\Dom(\gamma)$ and a~simplicial map 
$\gamma_\map: \Dom(\gamma)\to X$ with the following properties:
\begin{enumerate}
\item $|\Dom(\gamma)| = D^{d}$.
\item $\gamma_\map$ maps $\Dom(\gamma)$ to the set of simplices contained in the path ${\gamma}$.
\end{enumerate}
We will utilize the following construction given in \cite{Berger_paper}.

\begin{definition}
Let $\gamma \in IX_{d-1}$ . We define $\Dom(\gamma)$ and $\gamma_\map$ as follows.
Suppose, that $\gamma = [y_1,i_1]^{\epsilon_1}[y_2,i_2]^{\epsilon_2}\cdots [y_k,i_k]^{\epsilon_k}$. 
For every edge $[y_j, i_j]^{\epsilon_j}$, let $\alpha_j$ be the simplicial map  $\stdsimp{d} \to y_j$ sending the nondegenerate $d$ simplex in $\stdsimp{d}$ to $y_j$.

We define $\Dom(\gamma)$ as a quotient of the disjoint union of $k$ copies of $\stdsimp{d}$:
\[
\Dom(\gamma) = \bigsqcup_{i = 1} ^{k} \stdsimp{d} /\sim
\]
where each copy of $\stdsimp{d}$ corresponds to a domain of a unique $\alpha_j$ and the relation is given by  
\[
(\alpha_j)^{-1} \target ([y_j,i_j]^{\epsilon_j}) \sim (\alpha_{j+1})^{-1}  \source ([y_{j+1},i_{j+1}]^{\epsilon_{j+1}}) . 
\]
The map $\gamma_\map$ is induced by the collection of maps $\alpha_1, \ldots, \alpha_k$:
\[\xymatrix{
{\bigsqcup_{i = 1} ^{k} \stdsimp{d}}\ar[rrd]^{\alpha_1, \ldots, \alpha_k} \ar@{>>}[d]&&\\ 
\Dom(\gamma) \ar[rr]^{\gamma_\map} && X.
}
\]
\end{definition}

We recall that simplicial set $\oI X$ was defined as the set of ``reduced and compressed'' paths in $IX$. Similarly, one introduces a reduced and compressed versions of the construction $\Dom$. As a final step we then get
\begin{lemma}[Section 2.4 in~\cite{Berger_paper}]
\label{l:arrow3}
Let $\gamma \in \oLoop X_{d-1}$ such that $d_i \gamma = 1\in \oLoop X$ for all $i$. 
Then the map $\gamma_\map\: \Dom(\gamma)\to X$ factorizes through a simplicial set model of the sphere $\Sigma^{d}(\gamma)$ as follows:
\[\xymatrix{
\Dom(\gamma) \ar[rrd]^{\gamma_\map} \ar@{>>}[d]&&\\ 
\Sigma^{d}(\gamma) \ar[rr]^{\gamma_\sph} && X.
}
\]
Further, $\pi_{d-1}(\oLoop X)\ni [\gamma] \simeq [\gamma_\sph] \in \pi_{d}(X)$.
\end{lemma}
We will not give the proof of correctness of Lemma~\ref{l:arrow3} (it can be found in \cite{Berger_paper}). 
Instead, in the next section, we only describe the algorithmic construction of $\gamma_\sph\: \Sigma^{d}(\gamma) \to X$ and give a running time estimate.

\heading{Algorithm from Lemma~\ref{l:arrow3}.}

The algorithm accepts an element $\gamma \in \oLoop X_{d-1}$ such that $d_i \gamma = 1\in \oLoop X$ for all $i$, a~spherical element. We divide the algorithm into four steps that correspond to the four step factorization in the following diagram:
\[\xymatrix{
\Dom(\gamma) \ar[rrdd]^{\gamma_\map} \ar@{>>}[d]&&\\ 
\oDom(\gamma) \ar[rrd]^{\gamma_\comp} \ar@{>>}[d]&&\\
\ooDom(\gamma) \ar[rr]_{\gamma_\compred} \ar@{>>}[d]&&X\\
\Sigma^{d}(\gamma) \ar[rru]_{\gamma_\sph} &&
}
\]

\begin{enumerate}
\item[$\Dom(\gamma)$:] We interpret $\gamma$ as an element in $IX$ and construct $\gamma_\map\: \Dom(\gamma) \to X$. This is clearly linear in the size of $\gamma$.
\item[$\oDom(\gamma)$:] The algorithm checks, whether an edge $[y,j]^\epsilon$ in $d_{i_1}d_{i_2}\ldots d_{i_\ell} \gamma$, where $0\leq {i_1} < {i_2}< \ldots < {i_\ell}<(d-\ell -2)$ is \emph{compressible}, i.e. $y = s_j d_j y$. If this is the case, add a corresponding relation on the preimages: $\alpha^{-1}(y) \sim s_j d_j \alpha^{-1}(y)$. Factoring out the relations, we get a map $\gamma_\comp\:  \oDom(\gamma)  \to X$.

Although the number of faces we have to go through is exponential in $d$, this is not a problem, since $d$ is deemed as a constant in the algorithm and so is $2^d$. Hence the number of operations is again linear in the size of $\gamma$.

\item[$\ooDom(\gamma)$:] Let $k<d$. We know that $\overline{d_k \gamma}=1_*$, so after removing all compressible elements from the path $d_k \gamma$,
it will contain a sequence of pairs ($[y_i,j_i]^{\epsilon_i}, [y_i, j_i]^{-\epsilon_i})$ such that, after removing all $[y_u, j_u]^{\pm 1}$ for all $u<v$, then
$[y_v, j_v]^{\epsilon_v}$ and $[y_v, j_v]^{-\epsilon_v}$ are next to each other.\footnote{For example, $[a,1] [b,2] [b,2]^{-1} [a,1]^{-1}$ can be split
into a sequence $([b,2], [b,2]^{-1}), ([a,1], [a,1]^{-1})$.}
Each such pair  $([y_i,j_i]^{\epsilon_i}, [y_i, j_i]^{-\epsilon_i})$ corresponds to a pair of indices $(l_i,m_i)$ corresponding to the positions of those
edges in $d_k \gamma$.
These sequences are not unique, but can be easily found in time linear in $\mathrm{length}(\gamma)$. 
Then we glue $\alpha_{l_i}^{-1}(y_i)$ with $\alpha_{m_i}^{-1}(y_i)$ for all $i$. Performing such identifications for all $k$ defines the new simplicial set $\ooDom(\gamma)$.
\item[$\Sigma^{d}(\gamma)$:] It remains to identify  $\alpha^{-1}(\source \gamma)$ and $\alpha^{-1} (\target \gamma)$ with the appropriate degeneracy 
of the (unique) basepoint. The resulting space $|\Sigma^{d}(\gamma)|$ is a~$d$-sphere.

\end{enumerate}

\section{Polynomial-Time Loop Contraction in $F_d$}
\label{a:loop_contraction}
In this section, we show that simplicial sets $F_\thedimm$, $2\leq \thedimm \leq \thedim$ constructed algorithmically in Section~\ref{s:proof_1}
have polynomial-time contractible loops, thus proving Lemma~\ref{l:contr_loop}. We first give the contraction on $F_2$  and show that the contraction $F_i, i>3$ follows from the contraction on $F_3$. The majority of the effort in this section is then concentrated on the description of the contraction $c_0$ on $F_3$.

\heading{Notation.}
We will further use the following shorthand notation: For a $0$-reduced simplicial set $X$  we will denote the iterated degeneracy $s_0 \cdots s_0 *$ of its unique basepoint $*$  by $*$ and we set $\pi_i = \pi_i(X)$. For any Eilenberg-Maclane space $K(\pi_i, i-1)$, $i\geq 2$, we denote its basepoint and its degeneracies by $0$. From the context, it will always be clear which simplicial set we refer to. 

\heading{Loop contraction on $F_2$.}
Assuming that $X$ is a $0$-reduced, $1$-connected simplicial set with a given algorithm that computes the contraction on loops $c_0\: (GX)_0 \to (GX)_1$, the contraction $c_0$ on $F_2$ is automatically defined, as $X = F_2$. 

\heading{Loop contraction on $F_i$, $i>3$.}
Suppose we have defined the contraction on the generators of $G_0(F_3)$. i.e. for any $(x,k) \in (X\times_{\tau'}K(\pi_2, 1))_1$ we have
\[
c_0(\overline{(x,k)}) = \overline{(x_1,k_1)}^{\epsilon_1} \cdots \overline{(x_\thedim,k_\thedim)}^{\epsilon_\thedim}\qquad (x_j,k_j) \in (F_3)_2, \epsilon_j \in \Z, 1\leq j\leq \thedim
\]
such that $d_0 c_0 (\overline{(x,k)}) = \overline{(x,k)}$ and $d_1 c_0 (\overline{(x,k)}) = 1$. In detail, we get the following:
\begin{align}
\overline{(x,k)} &= d_0 c_0(\overline{(x,k)}) =\overline{(d_0 x_1,d_0 k_1)}^{\epsilon_1} \cdots \overline{(d_0 x_\thedim,d_0 k_\thedim)}^{\epsilon_\thedim} \label{e:dnull}\\
1 &= d_1 c_0(\overline{(x,k)}) =\big(\overline{(d_2 x_1 ,  \tau'(x_1) d_2 k_1 )}^{-1}\cdot \overline{(d_1 x_1 ,d_1 k_1)}\big)^{\epsilon_1} \cdots \label{e:dfirst} \\  
{ } &{}\big(\overline{(d_2 x_\thedim ,  \tau'(x_\thedim) d_2 k_\thedim )}^{-1}\cdot \overline{(d_1 x_\thedim ,d_1 k_\thedim)}\big)^{\epsilon_\thedim}\nonumber
\end{align}
We now aim to give a reduction on the generators of $G_0(F_i)$, $i>3$. Simplicial set ${F_i}$ is an iterated twisted product of the form
\[
 \big((( X\times_{\tau'} K(\pi_2, 1)) \times_{\tau'} K(\pi_3, 2) )\times_{\tau'} \cdots \times_{\tau'}  K(\pi_{i-2}, i-3)\big)  \times_{\tau'} K(\pi_{i-1}, i-2)
\]
As simplicial sets $K(\pi_{i-1}, i-2)$ are $1$-reduced for $i>3$, we can identify elements of $({F_i})_1$ with vectors $(x,k,0, \ldots,0)$, where $k\in K(\pi_{2}, 1)_1, x\in X_1$. We further shorthand the series of $i-3$ zeros in the vector with $\zero$. Hence generators $G_0(F_i)$ are of the form $\overline{(x,k,\zero)}$. The $1$-reducedness also implies that $\tau' (\alpha) = 0$ whenever $\alpha \in (F_i)_2$, $i>2$.

Finally, we set 
\[
c_0(\overline{(x,k,\zero)}) = \overline{(x_1, k_1, \zero)}^{\epsilon_1} \cdots \overline{(x_\thedim, k_\thedim, \zero)}^{ \epsilon_\thedim}\qquad (x_j, k_j, \zero) \in (F_i)_2, \epsilon_j \in \Z, 1\leq j\leq \thedim
\]

The (almost) freeness of $G_0(F_i)$, the fact that $K(\pi_{i-1}, i-2)$ are $1$-reduced for $i>3$ and equations \eqref{e:dnull}, \eqref{e:dfirst} give that $d_0 c_0(\overline{(x,k,\zero)}) = \overline{(x,k,\zero)}$ and $d_1 c_0(\overline{(x,k,\zero)}) = 1$.

Before the definition of contraction on simplicial set $F_3$, we remind the basic facts involving the simplicial model of Eilenberg-MacLane spaces we are using.
\heading{Eilenberg--MacLane spaces.}
As noted in Section~\ref{s:prelim}, given a group $\pi$ and an integer $i\geq 0$ an Eilenberg--MacLane space $K(\pi, i)$ is a space satisfying
\[
\pi_j (K(\pi, i)) = 
\left\{
	\begin{array}{ll}
		\pi & \mbox{for } j = i,\\
		0 & \mbox{else}.
	\end{array}
\right.
\]
In the rest of this section, by $K(\pi, i)$ we will always mean the simplicial model which is defined in \cite[page 101]{may}
\[
K(\pi, i)_q = Z^i (\stdsimp{q}; \pi),
\] 
where $\stdsimp{q} \in \sSet$ is the standard $q$-simplex and $Z^i$ denotes the cocycles. This means that each $q$--simplex is regarded as a labeling of the $i$--dimensional faces  of $\stdsimp{q}$ by elements of $\pi$ such that they add up to $0\in \pi$ on the boundary of every $(i+1)$-simplex in $\stdsimp{q}$, hence elements of $K(\pi, q)_q$ are in bijection with elements of $\pi$. The boundary and degeneracy operators in $K(\pi, \thedimm)$ are given as follows: For any  $\s \in K(\pi, i)_q$, $d_j(\s) \in K(\pi, \thedimm)_{q-1}$ is given by a restriction of $\s \in K(\pi, i)$ to the $j$-th face of $\stdsimp{q}$.
To define the degeneracy we first introduce mapping $\eta_j \colon \{0,1, \ldots, q+1 \} \to \{0,1, \ldots, q\}$ given by
\[
\eta_j (\ell)=
\left\{
	\begin{array}{ll}
		\ell & \mbox{for } \ell \leq j, \\
		\ell-1 & \mbox{for } \ell > j.
	\end{array}
\right.
\]
Every mapping $\eta_j$ defines a map $C^*(\eta_j )\:C^*(\stdsimp{q}) \to C^*(\stdsimp{q+1})$.The degeneracy $s_j \s$ is now defined to be $C^*(\eta_j)(\sigma)$ 
(see~\cite[§~23]{may}).

It follows from our model of Eilenberg-MacLane space, that elements of $K(\pi_2, 1)_2$ can be identified with labelings of $1$-faces of a $2$-simplex by elements of $\pi_2$ that sum up to zero. 

As $\pi_2$ is an Abelian group, we use the additive notation for $\pi_2$. We identify the elements of $K(\pi_2, 1)_2$ with triples $(k_0, k_1, k_2)$, $k_i \in \pi_2$, $0\leq i\leq2$, such that $k_0 - k_1 + k_2 = 0 \in \pi_2$. 

\heading{Loop contraction on $F_3$.}
Let $X$ be a $0$-reduced, $1$-connected simplicial set with a given algorithm that computes the contraction on loops $c_0\: (GX)_0 \to (GX)_1$. 


In the rest of the section, we will assume $x \in X_1$. Then by our assumptions $c_0 \overline x = \overline{y_1}^{\epsilon_1} \cdots \overline{y_n}^{\epsilon_n}$, where $y_i \in X_2, \epsilon_i \in \Z$, $1\leq i\leq n$.
Let $k_i = \tau' (y_i)$.

We first show that in order to give a contraction on elements of the form $\overline{(x,0)}$ and $\overline{(x,k)}$, it suffices to have the contraction on elements of the form $\overline{(*,k)}$:


\heading{Contraction on element $(x,0)$.}
Let $\overline{(x,0)} \in G_0 (F_3)$. We define
\[
c_0 \overline{(x,0)} = \prod_{i = 1} ^{n}\big( c_0 \overline{(*,k_i)}^{-1} \overline{( s_1 d_2 y_i,(k_i, k_i, 0))}\cdot\overline{(y_i,0)}\big)^{\epsilon_i}.
\]

\heading{Contraction on element $(x,k)$.}
Suppose $\overline{(x,k)} \in (GF_3)_0$. The formula for the contraction is given using the formulae on contraction on $\overline{(x,0)}$ and $\overline{(*,k)}$ as follows
\[
c_0 \overline{(x,k)} = \overline{(s_0 x,(k, 0,-k))} \cdot s_0 \overline{(x, 0)}^{-1} \cdot s_0 \overline{(*, -k)} \cdot c_0 (\overline{(*,-k)})^{-1}\cdot c_0 (\overline{(x,0)}) 
\]
\heading{Contraction on element $(*,k)$.} 
We formalize the existence of the contraction as Proposition \ref{p:loop_contraction} given at the end of this section. Due to the fact that the proof is rather technical, we need to define and prove some preliminary results first:
\begin{definition}
Let $Z = \{z \in (GF_3)_1 \mid d_0 z = 1\}$ and let $W = \{d_1 z \mid z\in Z \}$ We define an equivalence relation $\sim$ on the elements of $W$ in the following way: We say that $w\sim w'$ if there exists $z \in Z$, $\alpha, \beta \in (GF_3)_1$ such that $d_1 z = w$, $\alpha z \beta \in Z$ and $d_1(\alpha z \beta) = w'$.
\end{definition}
\begin{lemma}\label{l:rules}
Let $w\in W$ such that
\begin{enumerate}
\item\label{rot1}  $w= \overline{(x,k)}^{\epsilon}\cdot \alpha $, where $\alpha \in (GF_3)_1$ Then $w = \overline{(x,k)}^{\epsilon}\cdot \alpha \sim \alpha \cdot (x,k)^{\epsilon} = w'$.
\item\label{invert}  $w =\overline{(*,k)}^{\epsilon} \cdot \alpha$, where $\alpha \in (GF_3)_0$. Then $w \sim w'= \overline{(*,-k)}^{-\epsilon} \cdot \alpha$.
\item\label{kxjoin} $w =  \overline{(*,- k)}^{-1} (x,0)\cdot \alpha$, where $\alpha \in (GF_3)_0$. Then $w \sim w' = \overline{(x,k)}\cdot \alpha$.
\item\label{kxreduce} $w = \overline{(x,0)}^{-1} \overline{(x,k)} \cdot \alpha$, where $\alpha \in (GF_3)_0$. Then $w \sim w'= \overline{(*,k)} \cdot \alpha$.
\item\label{klplus} $w = \overline{(*, -l)}^{-1} \overline{(*,k)} \cdot \alpha$, where $\alpha \in (GF_3)_0$. Then $w \sim w'= \overline{(*,k+l)}\cdot \alpha$.
\end{enumerate}
\end{lemma}
\begin{proof}
In all cases, we assume $z\in Z$ such that $d_1 z = w$ and we give a  formula  for $z' \in Z$ with $d_1 z' = w'$:
\begin{enumerate}
\item$z' = s_0 \overline{(x.k)}^{-\epsilon} \cdot z \cdot s_0 \overline{(x,k)}^{\epsilon}$.
\item $z' = \overline{(*,(k, 0, -k))}^{\epsilon} \cdot (s_0\overline{(*,k)})^{-\epsilon} \cdot z$.
\item $z' = (s_0\overline{(x,k)}) \cdot \overline{(s_0 x,(k, 0, -k))}^{-1} \cdot z$. 
\item $z' = (s_0\overline{(*,k)})\overline{(s_1 x,(k,k,0))}^{-1}\cdot z$.
\item $z' = \overline{(s_0(*,k + l))}\overline{(*,(k + l,k, - l))}^{-1}\cdot z$.
\end{enumerate}
\end{proof}
 
\begin{lemma}\label{l:sequence}
Let $z \in (GF_3)_1$, $z\in Z$  with
\[d_1 z = w  =  \overline{(*,-k_1)}^{-1}\cdot \overline{(x_1, 0)}^{\epsilon_1} \cdots \overline{(*,-k_n)}^{-1} \cdot\overline{(x_n, 0)}^{\epsilon_n}\]
where $\overline{x_1}^{\epsilon_1} \cdots \overline{x_n}^{\epsilon_n} = 1$ in $GX_0$, $x_i \in X$, $k_i \in \pi_2(X)$, $\epsilon_i \in \{1,-1\}$, $1\leq i\leq n$. Then $ w \sim  \overline{(\sum_{i = 1} ^{n} k_i, *)}$.
\end{lemma}
\begin{proof}
We achieve the proof using a sequence of equivalences given in Lemma~\ref{l:rules}.
Without loss of generality we can assume that  $x_1 = x_2^{-1}$ and $\epsilon_1, \epsilon_2 = 1$ (If this is not the case, we can use rule \eqref{rot1} and/or relabel the elements). Using \eqref{rot1} gives us
\begin{align*}
w  =& \overline{(*,-k_1)}^{-1} \cdot \overline{(x_2,0)}^{-1} \cdot\overline{(*, -k_2)}^{-1} \cdot \overline{(x_2,0)} \cdots \overline{(*,-k_n)}^{-1} \cdot \overline{(x_n,0)}^{\epsilon_n} \\
\sim &  \overline{(*, -k_2)}^{-1} \cdot \overline{(x_2,0)}\cdots \overline{(*,-k_n)}^{-1} \cdot \overline{(x_n,0)}^{\epsilon_n}\cdot \overline{(*,-k_1)}^{-1} \cdot \overline{(x_2,0)}^{-1}.
\end{align*}
Then successive use of \eqref{kxjoin},\eqref{rot1},\eqref{kxreduce}, \eqref{rot1} and finally \eqref{klplus} gives us
\begin{align*}
w  \sim &  \overline{(x_2,k_2)}\cdots \overline{(*,-k_n)}^{-1} \cdot \overline{(x_n,0)}^{\epsilon_n}\cdot \overline{(*,-k_1)}^{-1} \cdot \overline{(x_2,0)}^{-1}.\\
 \sim & \overline{(x_2,0)}^{-1}\cdot \overline{(x_2,k_2)}\cdots \overline{(*,-k_n)}^{-1} \cdot \overline{(x_n,0)}^{\epsilon_n}\cdot \overline{(*,-k_1)}^{-1} \\
 \sim & \overline{(*,k_2)}\cdots \overline{(*,-k_n)}^{-1} \cdot \overline{(x_n,0)}^{\epsilon_n}\cdot \overline{(*,-k_1)}^{-1}\\ 
 \sim & \overline{(*,k_1 +k_2)}\cdot \overline{(*,-k_3)}^{-1} \cdot \overline{(x_3, 0)}\cdots \overline{(*,-k_n)}^{-1} \cdot \overline{(x_n,0)}^{\epsilon_n}
\end{align*}
multiple use or rules \eqref{invert} and \eqref{rot1} and gives us
\begin{align*}
w \sim & \overline{(*, - k_1 - k_2 - k_3)}^{-1} \cdot \overline{(x_3, 0)}\cdots \overline{(*,-k_n)}^{-1} \cdot \overline{(x_n,0)}^{\epsilon_n}
\end{align*}
So far, we have produced some element $z' \in Z \subseteq (GF_3)_1$ such that $d_0 z' = 1$, 
\[d_1 z'  = \overline{(*, - k_1 - k_2 - k_3)}^{-1} \cdot \overline{(x_3, 0)}\cdots \overline{(*,-k_n)}^{-1} \cdot \overline{(x_n,0)}^{\epsilon_n}\]
and further $\overline{x_3}^{\epsilon_3} \cdots \overline{x_n}^{\epsilon_n}= 1$ in $GX_0$. 

It follows that the construction described above can be applied iteratively  until all elements $\overline{(x_i, 0)}$ are removed and we obtain $w \sim \overline{( -\sum_{i = 1} ^{n} k_i, *)}^{-1} \sim \overline{(\sum_{i = 1} ^{n} k_i, *)}$.
\end{proof}
\begin{proposition}\label{p:loop_contraction}
Let $k\in \pi_2(X)$. Then there is an algorithm that computes an element $z\in (GF_3)_1$ such that 
$d_0 z = \overline{(*,k)}$ and $d_1 z = 1$.
\end{proposition}
\begin{proof}
Given an element $k\in \pi_2 \cong H_2(X)$, one can compute a cycle $\gamma \in Z_2 (X)$ such that
\[
 [\gamma]  = k \in  \pi_2(X) \cong H_2(X) \cong H_2(K(\pi_2, 2)) \cong \pi_2(K(\pi_2, 2)),
\]
were the middle isomorphism is induced by $\varphi_2$ and the other isomorphisms follow from Hurewicz theorem.

If one considers $\gamma \in \widetilde{AX}_1$ then by Lemma~\ref{l:arrow2} one can algorithmically compute a spherical element $\gamma' = \overline{y_1}^{\epsilon_1}\cdots \overline{y_n}^{\epsilon_n} \in \widetilde{GX}_1$ where $y_i \in X_2$ and $\tau' y_i = k_i \in \pi_2 (X)$, such that $d_0 \gamma' = 1 = d_1 \gamma'$ and $\sum_{i=1}^{n} \epsilon_i \cdot k_i  = k$.

We define $z' \in (GF_3)_1 $ by
\[
z' = (\prod_{i = 1} ^{n} \overline{(s_0 d_0 y_i,(k_i, 0, -k_i))}^{\epsilon_i}) \cdot (\prod_{i = 1} ^{n} \overline{(y_i,(k_i, 0, -k_i))}^{\epsilon_i})^{-1}.
\]
Observe that $d_0(z') = 1$ and 
\[d_1 z'  = \big(\overline{(*, -k_1)}^{-1} \cdot \overline{(d_0 y_1,0 )}\big)^{\epsilon_1}\cdots \big(\overline{(*, -k_n)}^{-1} \cdot \overline{(d_0 y_n,0)}\big)^{\epsilon_n}.\]
We apply Lemma~\ref{l:sequence} on $z'$ and get an element $z'' \in (GF_3)_1$ with the property $d_0 z'' = 1$ and $d_1 z'' = \overline{(*, k)}$.
We define $z = s_0 \overline{(*, k)} \cdot (z'')^{-1}$. Thus $d_0 z = \overline{(*,k)}$ and $d_1 z = 1$.
\end{proof}
\heading{Computational complexity.}
We first observe that that formulas for $c_0$ on a general element $\overline{(x,k)}$ depend polynomially on the size of $c_0 (\overline{x})$ and the size of contractions on $\overline{(*,k)}$. Hence it is enough to analyse the complexity of the algorithm described in Proposition~\ref{p:loop_contraction}:

The computation of $\gamma'$ is obtained by the polynomial-time Smith normal form algorithm presented in~\cite{KannanBachem} and the polynomial-time algorithm in Lemma~\ref{l:arrow2}. The size of $z'$ depends polynomially (in fact linearly) on size of $\gamma'$. The algorithm described in Lemma~\ref{l:sequence} runs in a linear time in the size of $z'$.

To sum up, the algorithm computes the formula for contraction on the elements of $GF_i$ in time polynomial in the input ($\size\, X + \size\, c_0 (GX)$).




%
\section{Reconstructing a Map to the Original Simplicial Complex}
\label{sec:map-into-complex}
This section contains the proof of Lemma~\ref{l:lift_to_X}.
\label{s:edgewise}
\heading{Edgewise subdivision of simplicial complexes.}
In~\cite{Edelsbrunner:1999}, the authors present, for $k\in\N$, the \emph{edgewise subdivision} $\Esd_k(\Delta^m)$
of an $m$-simplex $\Delta^m$ that generalizes the two-dimensional sketch displayed in Figure~\ref{fig:edgewise}. 
This subdivision has several nice properties: in particular, the number of simplices of $\Esd_k(\Delta^m)$ grows polynomially with $k$. 
Explicitly, the subdivision can be represented as follows.
\begin{itemize}
\item The vertices of $\Esd_k(\Delta^m)$ are labeled by coordinates $(a_0,\ldots, a_{m})$ such that $a_j\geq 0$ and $\sum_j a_j=k$.
\item Two vertices $(a_0,\ldots a_m)$ and $(b_0,\ldots, b_m)$ are \emph{adjacent}, if there is a pair $j<k$ such that $|b_j-a_j|=|b_k-a_k|=1$
and $a_i=b_i$ for $i\neq j,k$. 
\item Simplices of $\Esd_k(\Delta^m)$ are given by tuples of vertices such that each vertex of a~simplex is adjacent to each other vertex.
\end{itemize}
We define the \emph{distance} of two vertices to be the minimal number of edges between them.
\begin{figure}
\begin{center}
\includegraphics{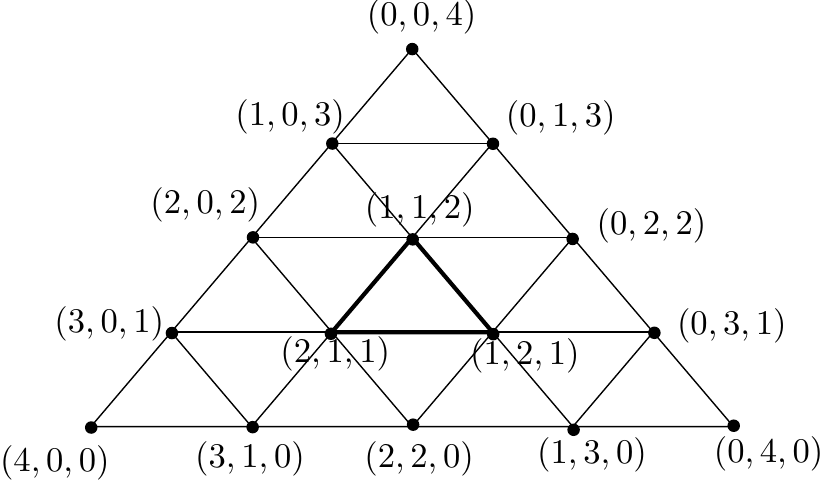}
\end{center}
\caption{Edgewise subdivision of a~$2$-simplex for $k=4$. 
In this case, there exists a copy of the $2$-simplex completely in the ``interior'', defined by vertices $(2,1,1)$, $(1,2,1)$ and
$(1,1,2)$. All other vertices are at the ``boundary'': more formally, their coordinatates contain a~zero.}
\label{fig:edgewise}
\end{figure}
An edgewise $k$-subdivision of $\Delta^m$ induces an edgewise $k$-subdivision of all faces, hence we may naturally define 
an~edgewise subdivision of any simplicial complex. 

\heading{Constructing the map $\Esd_k(\Sigma)\to X^{sc}$.}
Let $R$ be a chosen root in the tree $T$. 
We denote the tree-distance of a~vertex $W$ from $R$ by $\mathrm{dist}_T(W)$. 
Let 
$$
l:=\max\{\mathrm{dist}_T(V):\,\,V\,\text{ is a vertex of } X^{sc}\}
$$ 
be the maximal tree-distance of some vertex from $R$.
For each vertex $V$ of $X^{sc}$, there is a~unique path in the spanning tree that goes from $R$ into $V$. 
Further, we define the maps
$M(j): (X^{sc})^{(0)}\to (X^{sc})^{(0)}$ from vertices of $X^{sc}$ into vertices of $X^{sc}$ such that 
\begin{itemize}
\item $M(j)(V):=V$ if $j\geq \mathrm{dist}_T(V)$, and
\item $M(j)(V)$ is the vertex on the unique tree-path from $R$ to $V$ that has tree-distance $j$ from $R$, if $j < \mathrm{dist}_T(V)$.
\end{itemize}
If, for example, $R-U-V-W$ is a~path in the tree, then $M(0)(W)=R$, $M(1)(W)=U$ etc.
Clearly, $M(l)=M(l+1)=\ldots$ is the identity map, as $l$ equals the longest possible tree-distance of some vertex.

Assume that $d$ is the dimension of $\Sigma$ and $k:=l(d+1)+1$. 
We will define $f': \Esd_k(\Sigma)\to X^{sc}$ simplexwise. 
Let $\tau\in\Sigma$ be an~$m$-simplex and $f(\tau)=\tilde{\sigma}\in{X}$ be its 
image in the simplicial set ${X}$.  If ${\sigma}$ is the degeneracy of the base-point $*\in X$, then we define 
$f'(x):=R$ for all vertices $x$ of $\Esd_k(\tau)$: in other words, $f'$ will be constant on the subdivision of $\tau$.
Otherwise, $\tilde{\sigma}$ is not the degeneracy of a~point and has a~unique lift $\sigma\in X^{ss}$. (Recall that $X:=X^{ss}/T$.)
Let $(V_0,\ldots, V_m)$ be the vertices of $\sigma$ (order given by orientation): 
these vertices are not necessarily different, as $\sigma$ may be degenerate. 

In the algorithm, we will need to know which faces of $\sigma$ are in the tree $T$. We formalize this as follows:
let $S\subseteq 2^m$ be the family of all subsets of $\{0,1,\ldots, m\}$ such that 
\begin{itemize}
\item For each $\{i_0,\ldots, i_j\}\in S$, $\{V_{i_0}, \ldots, V_{i_j}\}$ is in the tree (that is, it is either an edge or a~single vertex),
\item Each set in $S$ is maximal wrt. inclusion.
\end{itemize}
Elements of $S$ correspond to maximal faces of $\sigma$ that are in the tree, in other words, to faces of $\tilde{\sigma}$ 
that are degeneracies of the~base-point.

\begin{definition}
\label{d:shade}
Let $\Delta^m$ be an oriented~$m$-simplex, represented as a~sequence of vertices $(e_0,\ldots, e_m)$.
For any face $s\subseteq \{e_0,\ldots, e_m\}$, we define the \emph{extended face $\mathcal{E}(s)$ in $\Esd_k(\Delta^m)$} 
to be the set of vertices $(x_0,\ldots, x_m)$ in $\Esd_k(\Delta^m)$ that have nonzero coordinates only on positions $i$ such that $e_i\in S$.
\end{definition}
The geometric meaning of this is illustrated by Figure~\ref{fig:extended}.

\begin{definition}
\label{def:ext-tree}
For $S\subseteq 2^m$, we define the \emph{extended tree} $\mathcal{E}(T)$ to be the union of the extended faces $\mathcal{E}(s)$ in $\Esd_k(\Delta^m)$
for all $s\in S$.
The edge-distance of a~vertex $x$ in $\Esd_k(\Delta^m)$ from $\mathcal{E}(T)$ will be denoted by $\mathrm{dist}_{ET}(x)$.
\end{definition}
\begin{figure}
\begin{center}
\includegraphics{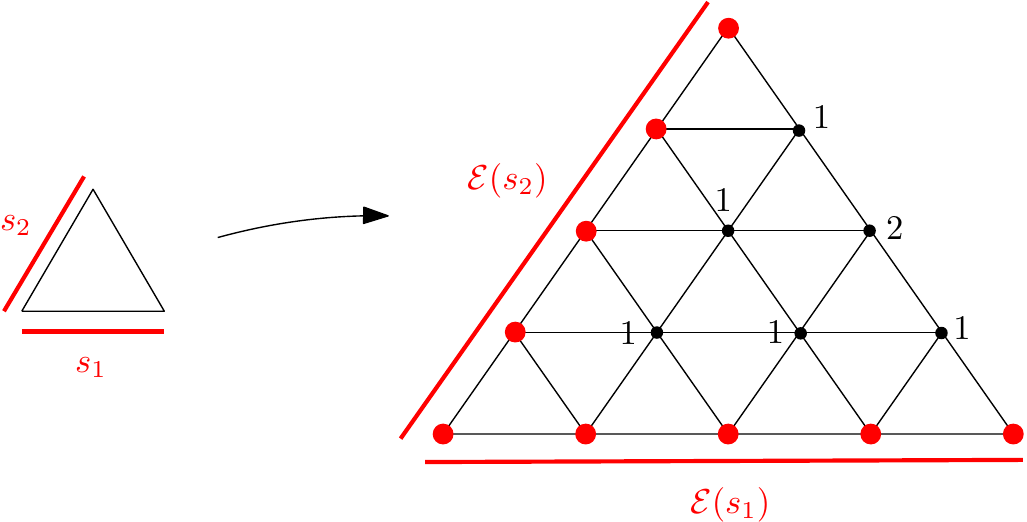}
\end{center}
\caption{Illustration of extended faces. Here $S=\{s_1,s_2\}$ corresponds to the lower- and left-face of a~$2$-simplex. The extended faces
$\mathcal{E}(s_1)$ and $\mathcal{E}(s_2)$ are sets of vertices of $\Esd_k(\Delta^2)$ that are on the lower- and left- boundary. The corresponding
extended tree $\mathcal{E}(T)$ is the union of all these vertices. The integers indicate edge-distances $\dist_{ET}$ of 
vertices in $\Esd_k(\Delta^2)$ from $\mathcal{E}(T)$.}
\label{fig:extended}
\end{figure}
In words, $\mathcal{E}(T)$ it is the union of all vertices in parts of the boundary of $\Esd_k(\Delta^m)$ that correspond 
to the faces of $\sigma$ that are in the tree, see Fig.~\ref{fig:extended}. The number $\mathrm{dist}_{ET}(x)$ 
is~the distance to $x$ from those boundary parts that correspond to faces of $\sigma$ that are in the tree.

To define a~simplicial map from $\Esd_k(\tau)$ to $X^{sc}$,
we need to label vertices of $\Esd_k(\tau)$ by vertices of $X^{sc}$ such that the induced map takes simplices in 
$\Esd_k(\tau)$ to simplices in $X^{sc}$. 
Recall that $V_0,\ldots,V_m$ are the vertices of $\sigma$. For $x=(x_0,\ldots, x_m)$, we denote by $\argmax\,x$
the smallest index of a~coordinate of $x$ among those with maximal value 
(for instance, $\argmax\,(4,2,1,4,0)=0$, as the first $4$ is on position $0$). The geometric meaning of $V_{\argmax x}$ is illustrated
by Figure~\ref{fig:argmax}.
\begin{figure}
\begin{center}
\includegraphics{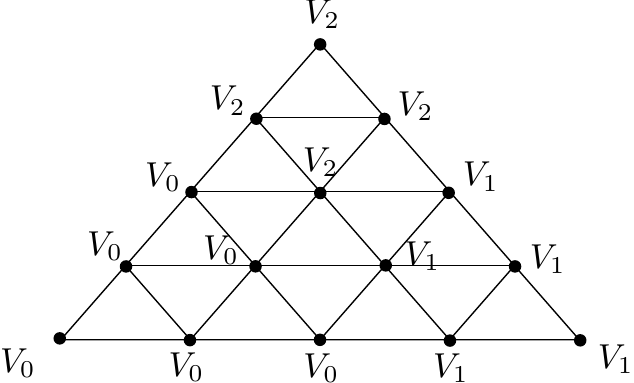}
\end{center}
\caption{Labelling vertices of $\Esd_k(\Delta^2)$ by $V_{\argmax x}$.}
\label{fig:argmax}
\end{figure}

Now we are ready to define the map $f': \Esd_k(\tau)\to X^{sc}$. It is defined on vertices $x$ with coordinates $(x_0,\ldots,x_m)$ by
\begin{equation}
\label{e:def_f'}
f'(x_0,\ldots, x_m) := M(\mathrm{dist}_{ET}(x))(V_{\argmax \,\,x}).
\end{equation}
Geometrically, most vertices $x$ will be simply mapped to $V_j$ for which the $j$'th coordinate of $x$ is dominant. 
In particular, a~unique $m$-simplex ``most in the interior of $\Esd_k(\tau)$'' with coordinates 
\begin{equation}
\label{e:inner_simplex}
\begin{pmatrix}
j+1\\
j\\
\ldots \\
j \\
\hline 
j+1 \\
\ldots \\
j+1
\end{pmatrix}^T,
\begin{pmatrix}
j\\
j+1\\
\ldots \\
j \\
\hline 
j+1 \\
\ldots \\
j+1
\end{pmatrix}^T,
\ldots,
\begin{pmatrix}
j\\
j\\
\ldots \\
j+1 \\
\hline 
j+1 \\
\ldots \\
j+1
\end{pmatrix}^T,
\begin{pmatrix}
j\\
j\\
\ldots \\
j \\
\hline 
j+2 \\
\ldots \\
j+1
\end{pmatrix}^T, \ldots,
\begin{pmatrix}
j\\
j\\
\ldots \\
j \\
\hline 
j+1 \\
\ldots \\
j+2
\end{pmatrix}^T
\end{equation}
for suitable $j$ will be labeled by $V_0,V_1,\ldots, V_m$; in other words, it will be mapped to $\sigma$.\footnote{If $\dim(\tau)=d$ 
is maximal, then $j=l$ and this most-middle simplex has particularly nice coordinates $(l+1,l,\ldots, l), \ldots, (l,\ldots, l,l+1)$.} 

However, vertices $x$ close to those boundary parts
of $\Esd_k(\tau)$ that correspond to the tree-parts of $\sigma$, will be mapped closer to the root $R$ and all the extended tree $\mathcal{E}(T)$
will be mapped to $R$. One illustration is in Figure~\ref{fig:lift_to_X}.
\begin{figure}[h!]
\begin{center}
\includegraphics[width=13 cm]{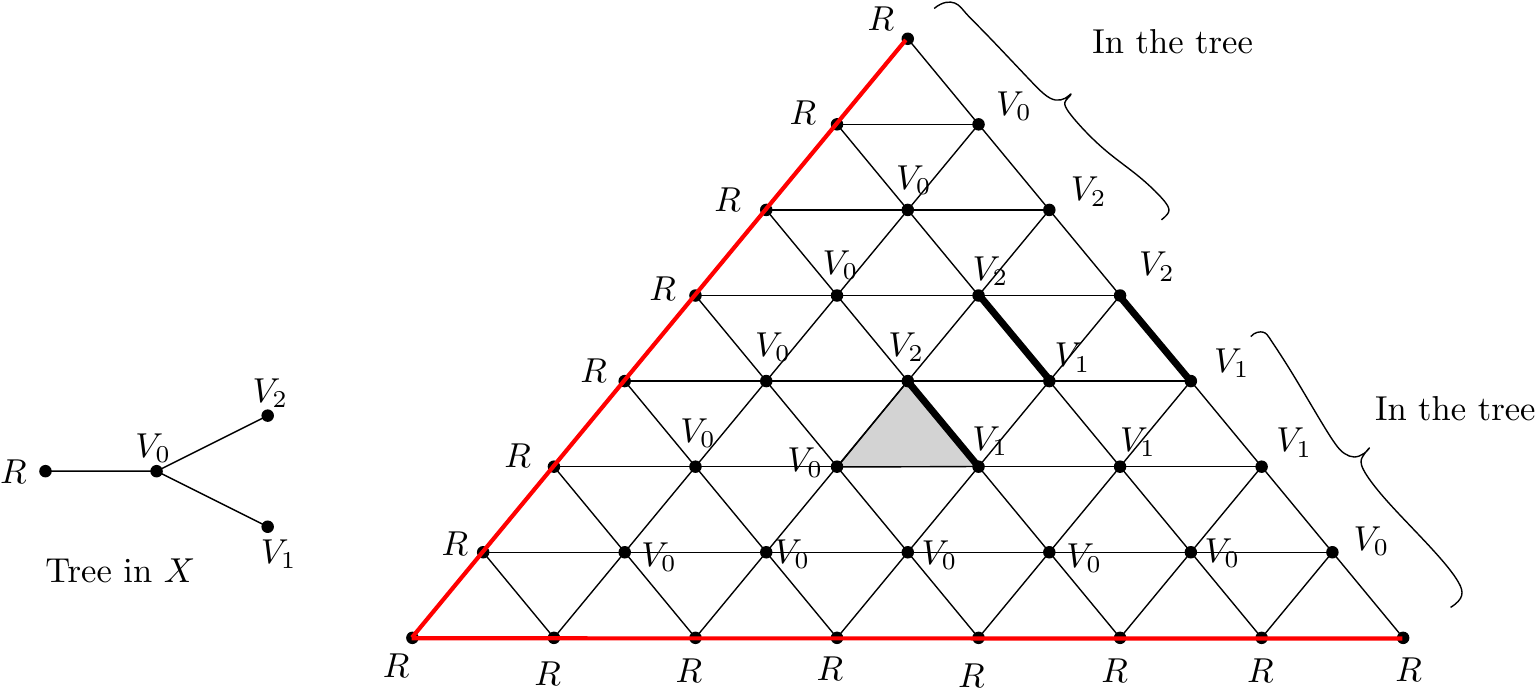}
\end{center}
\caption{Example of the labeling induced by formula (\ref{e:def_f'}). We assume that $f(\tau)=\tilde{\sigma}$ where $\sigma$ is a~simplex of $X^{sc}$
with three different vertices $V_0 V_1 V_2$. In this example, the tree connects $R-V_0-V_1$ as well as $R-V_0-V_2$ and the edge $V_1V_2$ is \emph{not}
in the tree. On the right, we give the induced labeling of vertices of $\Esd_k(\tau)$ which determines a~simplicial map to $X^{sc}$. 
The bottom and left faces
of $\sigma$ are in the tree, hence the bottom and left extended faces in $\Esd_k(\tau)$ are all mapped into $R$. The right face of $\sigma$ is
the edge $V_1 V_2$ that is not in the tree: the corresponding right extended face in $\Esd_k(\tau)$ is mapped to a~loop $R-V_0-V_1-V_2-V_0-R$, where
$V_1 V_2$ is the only part that is \emph{not} in the tree. The most interior simplex in $\Esd_k(\tau)$ is highlighted and is the only
one mapped to $\sigma$.}
\label{fig:lift_to_X}
\end{figure}

\heading{Computational complexity.} Assuming that we have a~given encoding of $\Sigma,f,{X},X^{sc}$ and a~choice of $T$ and $R$, defining a~simplicial map
$f':\Esd_k(\Sigma)\to X^{sc}$ is equivalent to labeling vertices of $\Esd_k(\Sigma)$ by vertices of $X^{sc}$. 
Clearly, the maximal tree-distance $l$ of some vertex depends only polynomially on the size of $X^{sc}$ and can be computed in polynomial time,
as well as the maps $M(0),\ldots, M(l)$. Whenever $j>l$, we can use the formula $M(j)=\mathrm{id}$.
Further, $k=l(d+1)+1$ is linear in $l$, assuming the dimension $d$ is fixed.
If $\tau\in\Sigma$ is an~$m$-simplex, then the number of vertices in $\Esd_k(\tau)$ is polynomial\footnote{Here the assumption on the fixed dimension 
$d$ is crucial.} in~$k$, and their coordinates 
can be computed in polynomial time.
Finding the lift $\sigma$ of $f(\tau)=\tilde{\sigma}$ is at most a~linear operation in $\size(X^{sc})+\size(\tilde{\sigma})$.
Converting $\sigma\in X^{ss}$ into an ordered sequence $(V_0,V_1,\ldots, V_m)$ amounts to computing its vertices
$d_0 d_1 \ldots \hat{d_i}\ldots, d_m \sigma$, where $d_i$ is omitted.
Collecting information on faces of $\sigma$ that are in the tree and the set of vertices $\mathcal{E}(T)$ is straight-forward: note that
assuming fixed dimensions, there are only constantly many faces of each simplex to be checked.
If $s=\{i_0,\ldots, i_j\}$ is a~face, then the edge-distance of a~vertex $x$ from $\mathcal{E}(s)$ equals to 
$\sum_{u} x_{i_u}$.
Applying formula (\ref{e:def_f'}) to $x$ requires to compute the edge-distance of $x$ from $\mathcal{E}(T)$:
this equals to the minimum of the edge-distances of $x$ from $\mathcal{E}(s)$ for all faces $s$ of $\sigma$ that are in the tree.
Computing $\argmax x$ is a~trivial operation. 
Finally, the number of simplices $\tau$ of $\Sigma$ is bounded by the size of $\Sigma$, so
applying (\ref{e:def_f'}) to each vertex of $\Esd_k(\Sigma)$ only requires polynomially many steps in $\size(\Sigma,f,X^{sc},T,{X})$.

\heading{Correctness.} What remains is to prove that formula (\ref{e:def_f'}) defines a~well-defined simplicial map and that
$|\Esd_k(\Sigma)|\to |X^{sc}|\to |{X}|$ is homotopic to $|\Sigma|\to |{X}|$.
\begin{lemma}
\label{l:correctness}
The above algorithm determines a~well-defined simplicial map $\Esd(\Sigma)\to X^{sc}$.
\end{lemma}
\begin{proof}
First we claim that formula (\ref{e:def_f'}) defines a~global labeling of vertices of $\Esd_k(\Sigma)$ by vertices of $X^{sc}$.
For this we need to check that if $\tau'$ is a~face of $\tau$, then (\ref{e:def_f'}) maps vertices of $\Esd_k(\tau')$ 
compatibly. This follows from the following facts, each of them easily checkable:
\begin{itemize}
\item If $\tau'$ is spanned by vertices of $\tau$ corresponding to $s\subseteq \{0,\ldots, m\}$, then
a vertex $x':=(x_0,\ldots,x_j)$ in $\Esd_k(\tau')$ has coordinates $x$ in $\Esd_k(\tau)$
equal to zero on positions $\{0\ldots, m\}\setminus s$ and to $x_0,\ldots, x_j$ on other positions, successively.
\item If $V_k':=V_{i_k}$ for $s=(i_0,\ldots,i_j)$ are the vertices of the corresponding face of $\sigma$, then 
$$V_{\argmax x'}'=V_{\argmax x}$$
\item The extended tree $\mathcal{E'}(T)$ in $\Esd_k(\tau')$ equals the intersection of the extended tree in $\Esd_k(\tau)$
with $\mathcal{E}(\tau')$
\item The distance $\mathrm{dist}_{ET}(x')$ in $\Esd_k(\tau')$ equals $\mathrm{dist}_{ET}(x)$ in $\Esd_k(\tau)$.
\end{itemize}

Further, we need to show that this labeling defines a~well-defined simplicial map, that is, it maps simplices to simplices.
We claim that each simplex in $\Esd_k(\tau)$ is mapped either to some subset of $\{V_0,\ldots, V_m\}$ or to
some edge in the tree $T$, or to a~single vertex.

We will show the last claim by contradiction. Assume that some simplex is \emph{not} mapped to a~subset of~$\{V_0,\ldots, V_m\}$, 
and also it is \emph{not} mapped to an~edge of the tree and \emph{not} mapped to a~single vertex.  Then 
there exist two vertices $x$ and $y$ in this simplex that are labeled by $U$ and $W$ in $X^{sc}$, such that either $U$ or $W$ is not in
$\{V_0,\ldots, V_m\}$, $UW$ is not in the tree, and $U\neq W$. 

The fact that at least one of $\{U,W\}$ does not belong to $\{V_0\ldots, V_m\}$, implies that
$\dist_{ET}(x)<l$ or $\dist_{ET}(y)<l$ (as $M(j)$ maps each $V_{\argmax x}$ to itself for $j\geq l$). 

Without loss of generality, assume that $\argmax x=0$ and $\argmax y=1$. Then the coordinates of $x$ and $y$ are either
$$
x=(j+1,j,x_3,\ldots, x_m),\quad y=(j,j+1,x_3,\ldots, x_m)
$$
such that $x_i\leq j+1$ for all $i\geq 3$, or 
$$
x=(j,j,x_3,\ldots, x_m),\quad y=(j-1,j+1,x_3,\ldots, x_m)
$$
for some $j$ such that $x_i\leq j$ for all $i\geq 3$. 

We claim that $V_0\neq V_1$ and that the edge $V_0V_1$ is~\emph{not} in the tree. 
This is because there exists a~tree-path from $R$ via $U$ to $V_0$ and
also a tree-path from $R$ via $W$ to $V_1$ (and $U\neq W$): both $V_0=V_1$ as well as a~tree-edge $V_0 V_1$ would create a~circle in the tree.
In coordinates, this means that vertices $(*,*,0,0,\ldots, 0)$ are not contained in $\mathcal{E}(T)$, apart of $(k,0,0,\ldots, 0)$ and 
$(0,k,0,\ldots, 0)$. So, any vertex in $\mathcal{E}(T)$ has a~zero on either the zeroth or the first coordinate.
This immediately implies that $\dist_{ET}(x)\geq j$ and $\dist_{ET}(y)\geq j$. Keeping in mind that coordinates of $x$ (and $y$)
has to sum up to $k=l(d+1)+1$, the smallest possible value of $j$ is $j=l$ (if $m=d$ is maximal), 
in which case $x=(l+1, l,l,\ldots ,l)$ and $y=(l,l+1,\ldots, l)$.
This choice, however, would contradict the fact that either $\dist_{ET}(x)<l$ or $\dist_{ET}(y)<l$.
Therefore we have a~strict inequality $j>l.$
Finally, we derive a~contradiction having either $\dist_{ET}(x)\geq j>l>\dist_{ET}(x)$, or a~similar inequality for $y$.

This completes the proof that each simplex is either mapped to 
a~subset of $\{V_0,\ldots, V_m\}$ or to an~edge in the tree or to a~single vertex:
the image is a~simplex in $X^{sc}$ in either case.
\end{proof}
\begin{lemma}
\label{l:homotopy_f_f'}
The geometric realisations of $pf': \Esd_k(\Sigma)\to {X}$ and $f: \Sigma\to {X}$ are homotopic.
\end{lemma}
\begin{proof}
%
First we reduce the general case to the case when all maximal simplices in $\Sigma$ (wrt. inclusion) have the same dimension $d$.
If this were not the case, we could enrich any lower-dimensional maximal simplex $\tau=\{x_0,\ldots, x_j\}\in\Sigma$ by new vertices
$y_{j+1}^\tau,\ldots, y_d^\tau$ and produce a~maximal $d$-simplex 
$$\tilde{\tau}=\{x_0,\ldots, x_j, y_{j+1}^\tau,\ldots, y_d^\tau\}.$$ 
Thus we produce a~simplicial complex $\tilde{\Sigma}\supseteq\Sigma$ with the required property. 
Whenever $f(\tau)$ is mapped to $\tilde{\sigma}$ where $\sigma=(V_0,\ldots,V_j)$, we define $f(\tilde{\tau})$ to be 
$s_j^{d-j}\tilde{\sigma}$, a~degenerate simplex with lift $(V_0,\ldots,V_j,V_j,\ldots, V_j)$. The map $f': \tilde{\Sigma}\to X^{sc}$
is constructed from $f: \tilde{\Sigma}\to X$ as above and if we prove that $|f|$ is homotopic to $|pf'|$ as maps $|\tilde{\Sigma}|\to |X|$, 
it immediately follows that their restrictions are homotopic as maps $|\Sigma|\to |X|$ as well.

Further, assume that all maximal simplices have dimension $d$.
Let $\tau\in\Sigma$ be a~$d$-dimensional simplex and let $\tau^{int}$ be the simplex in $\Esd_k(\tau)$ 
spanned by the vertices $$(l+1,l,\ldots,l), \ldots, (l,\ldots,l,l+1),$$ 
that is, the simplex in the interior of $\tau$ that is mapped by $pf'$ to $\tilde{\sigma}$.
Let $H_\tau(\cdot, 1): |\tau|\to|\tau|$ be a~linear map that takes $|\tau|$ linearly to $|\tau^{int}|$ via mapping the $i$'th vertex
to $(l,\ldots, l+1,1\ldots, l)$ where the $l+1$ is on position $i$.
Further, let $H_{\tau}$ be a~linear homotopy $|\tau|\times[0,1]\to|\tau|$ between the identity $H_\tau(\cdot, 0)=\mathrm{id}$ and $H_\tau(\cdot, 1)$.
The composition $|p f'| H_{\tau}$ then gives a~homotopy $|\tau|\times[0,1]\to |X|$ between the restrictions 
$(|pf'|)|_{|\tau|}$ and $(|f|)|_{|\tau|}$.
For a~general $x\in |\Sigma|$, there exists a~maximal $d$-simplex $|\tau|$ such that $x\in |\tau|$ and we define a~homotopy
$$
(x,t)\,\mapsto |p f'| H_\tau(x,t).
$$
It remains to show that this map is independent on the choice of $\tau$. 

Let as denote the (ordered) vertices of $\tau$ by $\{v_0,v_1,\ldots,v_d\}$ and let $\delta\subseteq\tau$ be one of its faces: further,
let $w_i$ be the vertex of $\tau^{int}$ with barycentric coordinates $(l,\ldots,l,l+1,l,\ldots,l)/k$ in $|\tau|$ such that the $l+1$ is in position $i$.
The homotopy $H_{\tau}$ sends points in $|\delta|$ onto the span of points $w_i$ for which $v_i\in\delta$. 
For $y\in|\delta|$, the $j$-th barycentric coordinate of $H_{\tau}(y,t)$ is equal to $t\,(l/k)$ 
for each $j\notin \delta$.
In particular, the $j$-th coordinate of $H_{\tau}(y,t)$ is between $0$ and $l/k$ for $j\notin\delta$, and hence it is not the ``dominant'' coordinate.
It follows that each $z:=H_{\tau}(x,t)$
is contained in the interior of a~unique simplex $\Delta$ of $\Esd_k(\tau)$ such that 
$v_{\argmax x}\in\delta$ for all vertices $x$ of $\Delta$.

Let $i_0< i_1\ldots < i_k$ be the indices such that $v_{i_j}\in\delta$ and $j_1< \ldots < j_{d-k}$ be the remaining indices.
Let $\tau'=(v_0',\ldots,v_d')$ be another $d$-simplex containing $\delta$ as a~face. 
Assume, for simplicity, that the vertices of $\tau'$ are ordered so that vertices of $\delta$ have orders $i_0,\ldots,i_k$---such as 
it is in $\tau$.
Let $\sigma,\sigma'$ be the lift of $f(\tau)$, $f(\tau')$ respectively, and $V_i$, $V_i'$ the $i$-th vertex of $\sigma$, $\sigma'$ respectively.

We define a~``mirror'' map $m: |\tau|\to|\tau'|$, which to a~point with barycentric coordinates $(x_0,\ldots,x_d)$ 
with respect to $\tau$ assigns a~point in $|\tau'|$ with the same barycentric coordinates with respect to $\tau'$. 
Clearly, $H_{\tau'}(y,t)=m(H_{\tau}(y,t))$ for $y\in|\tau|$ and whenever $z$ is in the interior of a~simplex $\Delta\in\Esd_k(\tau)$, 
then $m(z)$ is in the interior of $m(\Delta)$, where vertices of $\Delta$ and $m(\Delta)$ have the same barycentric coordinates 
with respect to $\tau$ and $\tau'$, respectively. 
If, moreover, $\Delta$ is such that each of its vertices $r$ have coordinates $\leq l/k$ on positions $j_1,\ldots, j_{d-k}$, 
then $V_{\argmax r}=V_{\argmax m(r)}'$.


To summarize these properties, $H_{\tau}(y,t)$ and $H_{\tau'}(y,t)$ satisfy that\footnote{
In general, vertices of $\delta$ may have different order in $\tau$ and $\tau'$ and the assumption on compatible ordering was chosen
only to increase readability. If $i_0'< \ldots <i_{k}'$ are such that $v_{i_j'}'=v_{i_j}$ (orders of $\delta$-vertices wrt. $\tau'$) 
and $j_1'<\ldots <j_{d-k}'$ are positions of the remaining vertices in $\tau'$, then $m$ is defined so that it maps $x\in |\tau|$
with $\tau$-coordinates $(x_0,\ldots, x_d)$ into $x'\in |\tau'|$ with coordinates $x_{i_j'}'=x_{i_j}$ and $x_{j_k}'=x_{j_k}$.
}
\begin{itemize}
\item they have the same coordinates wrt. $\tau$, $\tau'$, respectively, 
\item they are in the interior of simplices $\Delta\in\Esd_k(\tau)$, $\Delta'\in\Esd_k(\tau')$ 
whose vertices have the same coordinates wrt. $\tau$, $\tau'$, respectively, 
\item the $\argmax$ labeling induces the same labeling of vertices of $\Delta$, $\Delta'$ by vertices of $\delta$, respectively.
\end{itemize}
\begin{figure}
\begin{center}
\includegraphics{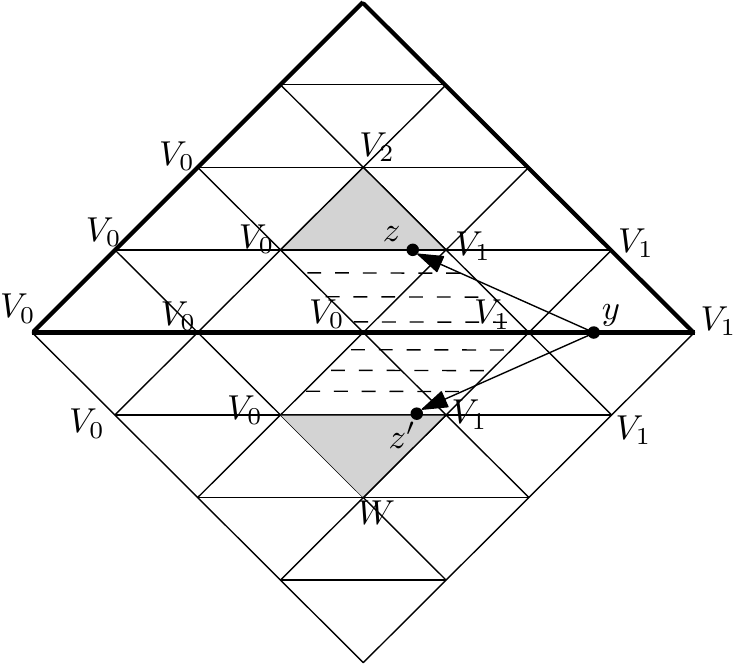}
\caption{The homotopy $H_{\tau}$ takes $y$ linearly into $z$ and $H_{\tau'}$ takes $y$ into $z'$.
Due to the symmetry represented by the horizontal line, $|pf'|$ maps $H_{\tau}(y,t)$ into the same point of $X$ as 
$|pf'| H_{\tau'}(y,t)$.}
\label{fig:homot}
\end{center}
\end{figure}
The map $pf'$ takes each $m$-simplex $\Delta$ in $\Esd_k(\tau)$ with vertices $t_u$ labeled by $V_{\argmax t_u}$
onto $p(V_{\argmax t_0},\ldots, V_{\argmax t_m})$ and it follows from the above properties that $m(\Delta)$ is mapped
to the same simplex.
We conclude that $|pf'| H_{\tau}(y,t)=|pf'|H_{\tau'}(y,t)$ for each $y\in |\delta|$ and $t\in [0,1]$.
\end{proof}

\heading{Acknowledgements.}
We would like to thank Marek Kr\v{c}\'al, Luk\'a\v{s} Vok\v{r}\'{\i}nek and Sergey Avvakumov for helpful conversations and comments.



\begin{thebibliography}{10}

\bibitem{Adyan:AlgorithmicUnsolvabilityRecognitionPropertiesGroups-1955}
S.~I. Adyan.
\newblock Algorithmic unsolvability of problems of recognition of certain
  properties of groups.
\newblock {\em Dokl. Akad. Nauk SSSR (N.S.)}, 103:533--535, 1955.

\bibitem{Anick-homotopyhard}
D.~J. Anick.
\newblock {The computation of rational homotopy groups is {\#}$\wp$-hard}.
\newblock {Computers in geometry and topology, Proc. Conf., Chicago/Ill. 1986,
  Lect. Notes Pure Appl. Math. 114, 1--56}, 1989.

\bibitem{Berger_thesis}
C.~Berger.
\newblock {\em An effective version of the Hurewicz theorem}.
\newblock Theses, Universit{\'e} Joseph-Fourier - Grenoble I, 1991.
\newblock URL: \url{https://tel.archives-ouvertes.fr/tel-00339314}.

\bibitem{Berger_paper}
Clemens Berger.
\newblock Un groupo{\"\i}de simplicial comme mod{\`e}le de l'espace des
  chemins.
\newblock {\em Bulletin de la Societ{\'e} mathematique de France},
  123(1):1--32, 1995.

\bibitem{Brown}
E.~H. B{rown (jr.)}.
\newblock {Finite computability of Postnikov complexes}.
\newblock {\em Ann. Math. (2)}, 65:1--20, 1957.

\bibitem{surv}
M.~Cadek, M.~Krc{\'a}l, J.~Matou{\v s}ek, Luk{\'a}s Vok{\v r}\'{\i}nek, and Uli
  Wagner.
\newblock Extending continuous maps: polynomiality and undecidability.
\newblock In {\em STOC}, pages 595--604, 2013.

\bibitem{post}
M.~{\v{C}}adek, M.~Kr\v{c}\'{a}l, J.~Matou\v{s}ek, F.~Sergeraert,
  L.~Vok\v{r}\'{\i}nek, and U.~Wagner.
\newblock Computing all maps into a sphere.
\newblock {\em J. ACM}, 61(3):17:1--17:44, June 2014.

\bibitem{ext-hard}
M.~{\v{C}}adek, M.~Kr\v{c}\'al, J.~Matou\v{s}ek, L.~Vok\v{r}\'{\i}nek, and
  U.~Wagner.
\newblock Extendability of continuous maps is undecidable.
\newblock {\em Discr. Comput. Geom.}, 51(1):24--66, 2013.

\bibitem{polypost}
M.~{\v{C}}adek, M.~Kr\v{c}\'al, J.~Matou\v{s}ek, L.~Vok\v{r}\'{\i}nek, and
  U.~Wagner.
\newblock Polynomial-time computation of homotopy groups and {P}ostnikov
  systems in fixed dimension.
\newblock {\em Siam Journal on Computing}, 43(5):1728--1780, 2014.

\bibitem{aslep}
Martin {\v{C}}adek, Marek Kr{\v{c}}{\'a}l, and Luk{\'a}{\v{s}}
  Vok{\v{r}}{\'\i}nek.
\newblock Algorithmic solvability of the lifting-extension problem.
\newblock {\em Discrete {\&} Computational Geometry}, pages 1--51, 2017.

\bibitem{EdelsbrunnerHarer:ComputationalTopology-2010}
H.~Edelsbrunner and J.~L. Harer.
\newblock {\em Computational topology}.
\newblock American Mathematical Society, Providence, RI, 2010.

\bibitem{Edelsbrunner:1999}
Herbert Edelsbrunner and Daniel~R. Grayson.
\newblock Edgewise subdivision of a simplex.
\newblock In {\em Proceedings of the Fifteenth Annual Symposium on
  Computational Geometry}, SCG '99, pages 24--30, New York, NY, USA, 1999. ACM.

\bibitem{weinberger_quantitative}
Steve Ferry and Shmuel Weinberger.
\newblock Quantitative algebraic topology and lipschitz homotopy.
\newblock {\em Proceedings of the National Academy of Sciences},
  110(48):19246--19250, 2013.

\bibitem{Filakovsky-tensor}
M.~Filakovsk{\'y}.
\newblock Effective chain complexes for twisted products.
\newblock Preprint, 2012.
\newblock URL: \url{arXiv: 1209.1240}.

\bibitem{VokriFil-homotopic}
M.~Filakovsk\'y and L.~Vok\v{r}\'{\i}nek.
\newblock Are two given maps homotopic? {A}n algorithmic viewpoint, 2013.
\newblock Preprint.
\newblock URL: \url{arXiv:1312.2337}.

\bibitem{nondec}
P.~Franek and M.~Kr\v{c}\'{a}l.
\newblock Robust satisfiability of systems of equations.
\newblock {\em J. ACM}, 62(4):26:1--26:19, 2015.

\bibitem{Freedman:Geometric-complexity-of-embeddings-in-Rd-2014}
Michael Freedman and Vyacheslav Krushkal.
\newblock Geometric complexity of embeddings in $\mathbb{R}^d$.
\newblock {\em Geometric and Functional Analysis}, 24(5):1406--1430, 2014.

\bibitem{FritschPiccinini:CellularStructures-1990}
R.~Fritsch and R.~A. Piccinini.
\newblock {\em Cellular structures in topology}, volume~19 of {\em Cambridge
  Studies in Advanced Mathematics}.
\newblock Cambridge University Press, Cambridge, 1990.
\newblock URL: \url{http://dx.doi.org/10.1017/CBO9780511983948}, \href
  {http://dx.doi.org/10.1017/CBO9780511983948}
  {\path{doi:10.1017/CBO9780511983948}}.

\bibitem{geoghegan2007}
R.~Geoghegan.
\newblock {\em Topological Methods in Group Theory}.
\newblock Graduate Texts in Mathematics. Springer New York, 2007.
\newblock URL: \url{https://books.google.at/books?id=BwX6gblqV8MC}.

\bibitem{goerssjardine}
P.~G. Goerss and J.~F. Jardine.
\newblock {\em Simplicial homotopy theory}.
\newblock Birkh{\"a}user, Basel, 1999.

\bibitem{gromov_quantitative}
M.~Gromov.
\newblock Quantitative homotopy theory.
\newblock {\em Prospects in Mathematics: Invited Talks on the Occasion of the
  250th Anniversary of Princeton University (H. Rossi, ed.)}, pages 45--49,
  1999.

\bibitem{Haefliger:PlongementsDifferentiablesDomaineStable-1962}
A.~Haefliger.
\newblock Plongements diff\'erentiables dans le domaine stable.
\newblock {\em Comment. Math. Helv.}, 37:155--176, 1962/1963.

\bibitem{Hatcher}
A.~Hatcher.
\newblock {\em Algebraic {T}opology}.
\newblock Cambridge University Press, Cambridge, 2001.

\bibitem{fKenzo}
J.~Heras, V.~Pascual, J.~Rubio, and F.~Sergeraert.
\newblock {fKenzo: a user interface for computations in algebraic topology}.
\newblock {\em J. Symb. Comput.}, 46(6):685--698, 2011.

\bibitem{Jardine:SimplicialApproximation-2004}
J.~F. Jardine.
\newblock Simplicial approximation.
\newblock {\em Theory Appl. Categ.}, 12:No. 2, 34--72, 2004.

\bibitem{Kan:Hurewicz}
D.~Kan.
\newblock The {H}urewicz {T}heorem.
\newblock In {\em International Symposium of Algebraic Topology, Autonomous
  University of Mexico and UNESCO}, pages 225--231, 1958.

\bibitem{Kan:1957}
D.~M. Kan.
\newblock A combinatorial definition of homotopy groups.
\newblock {\em Annals of Mathematics}, 67(2):282--312, 1958.

\bibitem{KannanBachem}
R.~Kannan and A.~Bachem.
\newblock Polynomial algorithms for computing the {S}mith and {H}ermite normal
  forms of an integer matrix.
\newblock {\em SIAM J. Computing}, 8:499--507, 1981.

\bibitem{Kochman}
S.~O. Kochman.
\newblock {\em {Stable homotopy groups of spheres. A computer-assisted
  approach.}}
\newblock Lecture Notes in Mathematics 1423. Springer-Verlag, Berlin etc.,
  1990.

\bibitem{pKZ1}
M.~Kr\v{c}\'al, J.~Matou\v{s}ek, and F.~Sergeraert.
\newblock Polynomial-time homology for simplicial {Eilenberg--MacLane} spaces.
\newblock {\em Foundat. of Comput. Mathematics}, 13:935--963, 2013.

\bibitem{MabillardWagner:Elim_II_SoCG-2016}
Isaac Mabillard and Uli Wagner.
\newblock Eliminating higher-multiplicity intersections, {II}. {T}he deleted
  product criterion in the $r$-metastable range.
\newblock In {\em Proc. 32nd International Symposium on Computational Geometry
  (SoCG 2016)}, pages 51:1--51:12, 2016.

\bibitem{MatousekTancerWagner:HardnessEmbeddings-2011}
J.~Matou{\v{s}}ek, M.~Tancer, and U.~Wagner.
\newblock {Hardness of embedding simplicial complexes in $\R^d$}.
\newblock {\em J. Eur. Math. Soc.}, 13(2):259--295, 2011.

\bibitem{Mat-homotopyW1}
J.~Matou\v{s}ek.
\newblock Computing higher homotopy groups is {$W[1]$}-hard.
\newblock {\em Fundamenta Informaticae}, 2014.

\bibitem{Matousek:Embeddability-in-the-3-sphere-is-decidable-2014}
Ji\v{r}\'{\i} Matou\v{s}ek, Eric Sedgwick, Martin Tancer, and Uli Wagner.
\newblock Embeddability in the 3-sphere is decidable.
\newblock In {\em Proceedings of the Thirtieth Annual ACM Symposium on
  Computational Geometry}, SOCG'14, pages 78--84, New York, NY, USA, 2014.

\bibitem{Matveev:AlgorithmicTopology-2007}
S.~Matveev.
\newblock {\em Algorithmic {T}opology and {C}lassification of 3-{M}anifolds}.
\newblock Springer, 2007.

\bibitem{may}
J.~P. May.
\newblock {\em Simplicial {O}bjects in {A}lgebraic {T}opology}.
\newblock Chicago Lectures in Mathematics. University of Chicago Press,
  Chicago, IL, 1992.
\newblock Reprint of the 1967 original.

\bibitem{Munkres}
J.~R. Munkres.
\newblock {\em Elements of Algebraic Topology}.
\newblock Addison-Wesley, Reading, MA, 1984.

\bibitem{quillen1967homotopical}
D.G. Quillen.
\newblock {\em Homotopical Algebra}.
\newblock Lecture Notes in Mathematics. Springer Berlin Heidelberg, 1967.

\bibitem{Rabin:RecursiveUnsolvabilityGroupTheoreticProblems-1958}
M.~O. Rabin.
\newblock Recursive unsolvability of group theoretic problems.
\newblock {\em Ann. of Math. (2)}, 67:172--194, 1958.

\bibitem{Ravenel}
D.~C. Ravenel.
\newblock {\em Complex Cobordism and Stable Homotopy Groups of Spheres (2nd
  ed.)}.
\newblock Amer. Math. Soc., 2004.

\bibitem{Real96}
P.~Real.
\newblock An algorithm computing homotopy groups.
\newblock {\em Mathematics and Computers in Simulation}, 42:461---465, 1996.

\bibitem{RomeroRubioSergeraert}
A.~Romero, J.~Rubio, and F.~Sergeraert.
\newblock {Computing spectral sequences}.
\newblock {\em J. Symb. Comput.}, 41(10):1059--1079, 2006.

\bibitem{SergRomEffHmtp}
A.~Romero and F.~Sergeraert.
\newblock Effective homotopy of fibrations.
\newblock {\em Applicable Algebra in Engineering, Communication and Computing},
  23(1-2):85--100, 2012.

\bibitem{Romero2016}
A.~Romero and F.~Sergeraert.
\newblock A {B}ousfield--{K}an algorithm for computing the effective homotopy
  of a space.
\newblock {\em Foundations of Computational Mathematics}, pages 1--32, 2016.

\bibitem{RubioSergeraert:ConstructiveAlgebraicTopology-2002}
J.~Rubio and F.~Sergeraert.
\newblock Constructive algebraic topology.
\newblock {\em Bull. Sci. Math.}, 126(5):389--412, 2002.

\bibitem{SergRub-homtypes}
J.~Rubio and F.~Sergeraert.
\newblock Algebraic models for homotopy types.
\newblock {\em Homology, Homotopy and Applications}, 17:139--160, 2005.

\bibitem{SergerGenova}
J.~Rubio and F.~Sergeraert.
\newblock Constructive homological algebra and applications.
\newblock Preprint, \texttt{arXiv:1208.3816}, 2012.

\bibitem{Schoen-effectivetop}
R.~Sch{\"o}n.
\newblock {\em Effective Algebraic Topology}.
\newblock Memoirs of the American Mathematical Society. American Mathematical
  Society, 1991.

\bibitem{Sergeraert:ComputabilityProblemAlgebraicTopology-1994}
F.~Sergeraert.
\newblock The computability problem in algebraic topology.
\newblock {\em Adv. Math.}, 104(1):1--29, 1994.

\bibitem{smith-mstructures}
J.~R. Smith.
\newblock {m}-{S}tructures determine integral homotopy type.
\newblock Preprint, \texttt{arXiv:math/9809151v1}, 1998.

\bibitem{Soare:ComputabilityDifferentialGeometry-2004}
R.~I. Soare.
\newblock Computability theory and differential geometry.
\newblock {\em Bull. Symbolic Logic}, 10(4):457--486, 2004.

\bibitem{vokrinek:oddspheres}
L.~Vok{\v{r}}{\'\i}nek.
\newblock Decidability of the extension problem for maps into odd-dimensional
  spheres.
\newblock {\em Discrete {\&} Computational Geometry}, 57(1):1--11, 2017.

\bibitem{Weber67}
C.~Weber.
\newblock Plongements de poly{\`e}dres dans le domaine metastable.
\newblock {\em Comment. Math. Helv.}, 42:1--27, 1967.

\bibitem{Zomorodian:TopologyComputing-2005}
A.~J. Zomorodian.
\newblock {\em Topology for {C}omputing}, volume~16 of {\em Cambridge
  Monographs on Applied and Computational Mathematics}.
\newblock Cambridge University Press, Cambridge, 2005.

\end{thebibliography}
\end{document}